\documentclass[11pt]{article}
\usepackage {a4wide}
\usepackage{amssymb,amsfonts,amsthm}
\usepackage{graphicx}
\usepackage{subfigure}
\usepackage{amssymb}
\usepackage{amsthm}
\usepackage{amsmath}
\usepackage{bm}
\usepackage[mathscr]{eucal}
\usepackage{cancel}
\usepackage{multirow}
\usepackage{wasysym}
\usepackage{color}
\usepackage{hyperref}
\theoremstyle{plain}
\newtheorem{theorem}{Theorem}

\newtheorem{lemma}{Lemma}
\newtheorem{remark}{Remark}
\newtheorem{proposition}{Proposition}

\def\hsp5{\hspace{5mm}}

\newcommand{\ts}{\mathsf{t}}
\newcommand{\xs}{\mathsf{x}}
\title{\sc Dynamics of quadratic $f(R)$ cosmology with a perfect fluid: Jordan and Einstein frames}
\begin{document}
\author{
\sc Artur Alho$^{1}$\thanks{Electronic address: {\tt artur.alho@tecnico.ulisboa.pt}},\,
    Margarida Lima$^{1,2}$\thanks{Electronic address: {\tt margarida.a.lima@tecnico.ulisboa.pt}}\, and
    Filipe C. Mena$^{1,3}$\thanks{Electronic address: {\tt filipecmena@tecnico.ulisboa.pt}} \\
$^{1}${\small\em Center for Mathematical Analysis, Geometry and Dynamical Systems,}\\
{\small\em Instituto Superior T\'ecnico, Universidade de Lisboa,}
{\small\em Av. Rovisco Pais, 1049-001 Lisboa, Portugal.}\\
$^{2}${\small\em OKEANOS - Instituto de Investigação em Ciências do Mar
	Universidade dos Açores}\\
	{\small\em Rua Prof. Doutor Frederico Machado, 4, 9901-862 Horta, Faial, A\c{c}ores, Portugal}\\
$^{3}${\small\em Centro de Matemática, Universidade do Minho, 4710-057 Braga, Portugal}\\
}
\date{}
\maketitle
\begin{abstract}
We investigate the global dynamics of the field equations of (pure) quadratic theories of gravity which generalise Einstein's theory
in spatially flat homogeneous and isotropic cosmological models with a perfect fluid. We introduce global and regular 3-dimensional dynamical systems' formulations, on both the Jordan frame and the conformally related Einstein frame.
The analysis in the Jordan frame explores the monotonicity properties of the interior flow which, together with the characterisation of the orbit structure on the 2-dimensional invariant boundaries and the desingularisation of non-hyperbolic fixed points, provides a global description of the flow and its limit sets. In the Einstein frame, the analysis uses the skew-product structure of the Einstein state space and the characterisation of the orbit structure on the 2-dimensional invariant boundaries. Furthermore, by obtaining asymptotic expansions we identify the solutions that are global conformally mapped from the Jordan frame to the Einstein frame and those that are not.
%
\end{abstract}
\newpage
\tableofcontents
\newpage
\section{Introduction}
Geometric theories of gravity alternative to Einstein's theory have been increasingly more popular to investigate physical models in astrophysics and cosmology. 
In fact, generalised theories may account for observational data  that challenges Einstein's theory such as the accelerated expansion of the universe and the rotational curves of galaxies \cite{Ferreira, yunes-siemens}. 

Amongst the existing modified theories, $f(R)$ theory of gravity is one of the simplest generalisations of General Relativity since it just incorporates higher order curvature terms in the Einstein-Hilbert action. 
This theory has been the subject of many applications which provide observational signatures distinct from General Relativity, for example in the spectra of galaxy clustering, in the cosmic microwave background and gravitational lensing phenomena (see~\cite{LivRev10, Sotiriou} for reviews). 

In $f(R)$ theory however, the evolution equations form a 4th order system of partial differential equations and a reduction is often needed to analyse its solutions. In this respect there has been recent important work showing the well-posedness of the initial value problem in $f(R)$ theories~\cite{LeFlochMa17a}. Nevertheless due to their mathematical complexity, even in the simplest cases, the global dynamics has hardly been rigorously analysed yet, which has resulted in several misconceptions and raised debates concerning their physical viability~\cite{AGPT07,APT07a,APT07b,CDCT05,CLCD08,GLD08,CTD09,AGD12}. It is therefore our purpose here to obtain a global description of the dynamics of pure quadratic $f(R)$ gravity in a cosmological setting.

In most cosmological models, the evolution equations reduce to ordinary differential equations since the models are taken to be spatially homogeneous.  This is then a particularly interesting  arena allowing the interplay between different disciplines of geometry, dynamics and physics. It is also a naturally suited setting to use and develop techniques of non-linear dynamical systems and, in fact, dynamical systems' approaches have been highly successful to describe globally the spacetime dynamics (see e.g. \cite{waiell97, col03,Hugwain02,Hug10,ugg13,alhetal15,alhetal22,alhetal20,Alho2025}).  

In General Relativity, the global dynamics of spatially homogeneous and isotropic models containing a perfect fluid with a linear equation of state is by now well understood \cite{waiell97, col03}. At the core of those studies is the introduction of bounded (dimensionless) scale-invariant variables based on the Hubble function. However, a direct generalisation of this procedure to $f(R)$ gravity is not possible, since in general the Hubble function vanishes during the evolution.   

A general framework to reformulate the set of evolution equations of $f(R)$ gravity to a regular dynamical system on a compact state-space was introduced in~\cite{ACU16}. A detailed analysis of the vacuum model $f(R)=R+\alpha R^2$, $\alpha>0$, was then performed and in particular this allowed to identify the solutions in the Jordan frame that are past asymptotically incomplete in the Einstein frame. Here, we are instead interested in the model $f(R)=\alpha R^2$ and the addition of a perfect fluid with a linear equation of state. Quadratic models with $f(R)=\alpha R^2$ in vacuum have been often used to study the physics of black holes \cite{Kehagias}, cosmological models \cite{CDCT05, Kounnas, Barker, Leon} and phenomenological approaches to relativistic quantum field theory (see  \cite{Hell, Garcia} and references therein).

Comparing the present problem with the one in \cite{ACU16}, besides the functional form of $f(R)$ being different, the inclusion of a perfect fluid results in a higher dimensional dynamical system, as well as leading to 
interactions between the fluid and the scalar field representation in the Einstein frame. This leads to new mathematical difficulties such as the appearance of non-hyperbolic fixed points located on the intersection of the state-space invariant boundaries, and whose desingularisation has to be done using successive  blow-ups.

Finally, it is often assumed that the global dynamics of $f(R)$ in the Jordan and Einstein frames are qualitatively equivalent \cite{LivRev10}, however this is not the case as noted e.g. in \cite{ACU16} for vacuum, after several previous discussions in different physical contexts  (see e.g. \cite{Faraoni, Gionti, Saa, Postma}). Here, we shall also investigate which solutions are globally contained in the two frames and show that the Einstein state-space is a trapping region of the Jordan state-space, and therefore there are solutions in the Einstein frame that can be past conformally extended in the Jordan frame.

The plan of the paper is as follows: In Section \ref{field-eqs} we review the derivation of the field equations of $f(R)$ gravity specialised to the present models. In Section~\ref{subsec:dynsysjord} we analyse the global dynamics in the Jordan frame and present the main result of the section in Theorem~\ref{GlobalTheoJordan}, 
which relies on monotonicity properties of the interior state-space flow and the orbit structure on its invariant boundaries. A crucial step in the proof consists in understanding the flow in a neighbourhood of two non-hyperbolic fixed points located on the intersection of the invariant boundaries and whose blow-ups are presented in Section~\ref{sec-blow-up}.
In Section~\ref{subsec:dynsysein} we analyse the global dynamics in the Einstein frame (Theorem \ref{theorem-a}) and, in Section~\ref{sec:Map}, we identify the solutions in the Jordan frame with those in the Einstein frame giving a complete characterisation of their $\alpha$ and $\omega$-limit sets.

We use units such that $8\pi G=c=1$, where $G$ is Newton's gravitational constant and $c$ the speed of light.

\section{The field equations of $f(R)$ gravity}
\label{field-eqs}
In this section we briefly review the derivation of the field equations of $f(R)$ gravity starting from action principles and specializing to spatially flat and isotropic cosmological models. We split the section in two parts: Firstly, we obtain the equations in the Jordan frame and then we consider its conformally related Einstein frame. The latter frame is widely used in the physics literature since it corresponds to a transformation of the $f(R)$ theories into General Relativity by including a scalar field minimally coupled to gravity. 
\subsection{Field equations in the Jordan frame}
On a 4-dimensional spacetime $(M,g)$ with Lorentzian metric $g$,
the action for $f(R)$ gravity theories can be written as
\begin{equation}\label{fRaction}
\mathcal{S}_{\mathrm{J}} = \int\left\{\frac{f(R)}{2}+\mathcal{L}_{\mathrm{pf}} \right\} dV_g 
\end{equation}
where $f(R)$ is some function of the Ricci scalar curvature $R$ and  $\mathcal{L}_{\mathrm{pf}}$ the matter Lagrangian density, which we assume to be of a perfect fluid. Einstein's theory of General Relativity is obtained by setting $f(R)=R$. 

Varying the action with respect to the metric, yields the equations of motion (see e.g.~\cite{LivRev10}):
\begin{subequations}
	\begin{align}
	F(R) R_{\mu\nu}-(\nabla_\mu \nabla_\nu -g_{\mu\nu}\nabla_\alpha \nabla^\alpha)F(R) - \frac{1}{2}g_{\mu\nu}f(R)&=T_{\mu\nu}^{(\mathrm{pf})}, \label{motion_eqs}
	\\
	\nabla_{\mu}T^{(\mathrm{pf})\mu }_{~\nu}&=0, \label{cons_eqs}
	\end{align}
\end{subequations}
where the rank 2 tensors $R_{\mu\nu}$ and $T_{\mu\nu}^{(\mathrm{pf})}$ denote the Ricci curvature and stress-energy, respectively, the greek indices $\mu,\nu,...=0,1,2,3$ and
\begin{equation}
F(R):=\frac{d f(R)}{d R}.
\end{equation}
As the matter content, we consider a perfect fluid with energy density $\rho_\mathrm{pf}>0$,  pressure $p_\mathrm{pf}$ and stress-energy tensor
\begin{equation}
\label{stress}
T^{(\mathrm{pf})}_{\mu\nu}= (\rho_\mathrm{pf}+p_\mathrm{pf})u_{\mu} u_{\nu}+p_\mathrm{pf}g_{\mu\nu},
\end{equation}
where $u^\mu$ is the fluid 4-velocity field with $u^{\mu}u_{\mu}=-1$.
As in most cosmological models, we assume a spatially flat homogeneous and isotropic metric in co-moving coordinates
\begin{equation}\label{metric}
g=-dt^2 + a^2(t)\delta_{ij}dx^{i}dx^{j},
\end{equation}
where $a(t)$ is a positive function called scale factor and $u^{\mu}=(1,0,0,0)$. We recall that the Hubble function is
\begin{equation}\label{DefH}
H:=\frac{\dot{a}}{a}
\end{equation}
where an overdot denotes the derivative with respect to time $t$ and the Ricci scalar for the metric~\eqref{metric} is known to be
\begin{equation}\label{DefR}
R=\dot{H}+2H^2.
\end{equation}
In turn the deceleration parameter $q$, which will be used later, is defined by the relation $\dot{H} = -(1 + q)H^2$, giving
\begin{equation}
q := 1- \frac{R}{6H^2}.
\end{equation}
Substituting \eqref{stress} and \eqref{metric} in \eqref{motion_eqs} gives
\begin{subequations}
	\begin{align}
	6H\big( F_{,R}\dot{R} + F(R) H\big) + f(R) -F(R)R  &= 2\rho_\mathrm{pf} \label{00_eq} \\
	F(R)\bigg( 2\frac{\ddot{a}}{a}+H^2 \bigg)+ F_{,RR}\dot{R}^2 + F_{,R}\ddot{R} - 2HF_{,R}\dot{R} + \frac{1}{2}\big( f(R)-F(R) R \big)&=- p_\mathrm{pf} ,	\label{ii_eq}
	\end{align}
\end{subequations}
where
\begin{equation*}
F_{,R}=\frac{d F(R)}{d R},\qquad  F_{,RR}=\frac{d^2 F(R)}{d R^2}.
\end{equation*}
Using relations~\eqref{DefH} and \eqref{DefR}, equation~\eqref{ii_eq} can be written as 
\begin{equation}
F_{,R}\ddot{R}=-3HF_{,R}\dot{R}- F_{,RR}\dot{R}^2-\frac{1}{3}\big( 2f(R)-F(R)R \big)  -\bigg(p_\mathrm{pf}-\frac{\rho_\mathrm{pf}}{3}\bigg)  .
\label{11_eq}
\end{equation}
In turn, the conservation of the stress-energy tensor, equation~\eqref{cons_eqs}, yields
\begin{equation}\label{Euler}
\frac{d\rho_\mathrm{pf}}{dt}=-3H (\rho_\mathrm{pf}+p_\mathrm{pf}).
\end{equation}
Given an equation of state relating the fluid pressure $p_\mathrm{pf}$ and the energy density $\rho_\mathrm{pf}$, then equations~\eqref{DefH}, \eqref{DefR},~\eqref{11_eq} and~\eqref{Euler} form a close system of ODEs for the state-vector $(a,H,R,\rho_\mathrm{pf})$, subject to the constraint~\eqref{00_eq}.

\subsection{Field equations in the Einstein frame}
The so-called Einstein frame formulation of $f(R)$ models is based on a conformal transformation of the Jordan frame metric $g$ to the Einstein frame metric $\tilde{g}$:
\begin{equation}
\label{fr}
\tilde{g}_{\mu\nu} = F g_{\mu\nu}, \qquad F = \frac{df}{dR}>0.
\end{equation}
Thus $F=0$ constitutes the boundary between the Einstein frame state-space and the Jordan frame state-space. Since in general $F=0$ is not an invariant subset in the Jordan frame, there are solutions that pass through the $F=0$ surface in the Jordan state-space. This means that there are solutions in the Einstein representation that can be extended in the Jordan frame. 
Under the conformal transformation~\eqref{fr} together with the identification
\begin{equation}\label{EFscalarfield}
\phi = \sqrt{\frac{3}{2}}\ln F, \qquad V(\phi) = \frac{RF - f}{2F^2},
\end{equation}
the action~\eqref{fRaction} of $f(R)$ theories is transformed to an action in the Einstein-Hilbert form as in General Relativity with a minimally coupled (to gravity) scalar field, and coupled to the perfect fluid Lagrangian density (see also \cite{LivRev10})
\begin{equation}\label{EHaction}
\mathcal{S}_E = \int \left\{\frac{\tilde{R}}{2}-\frac{\tilde{g}^{\mu\nu}}{2}\left(\tilde{\nabla}_{\mu}\phi\right)
\left(\tilde{\nabla}_{\nu}\phi\right)- V(\phi)+ F^{-2}(\phi)\tilde{\mathcal{L}}_{\mathrm{pf}} \right\}  dV_{\tilde g},
\end{equation}
which leads to the Einstein equations
\begin{equation}
\tilde{R}_{\mu\nu}-\frac{1}{2}\tilde{g}_{\mu\nu}\tilde{R}=\tilde{T}^{(\phi)}_{\mu\nu}+\tilde{T}^{(\mathrm{pf})}_{\mu\nu}, \label{Einstein-eqs-EF}
\end{equation}
where the perfect fluid energy-momentum tensor in this frame is given by
\begin{equation}\label{SETFluid}
\tilde{T}^{(\mathrm{pf})\mu}{}_\nu=\text{diag}\left(-\tilde{\rho}_\mathrm{pf},\tilde{p}_\mathrm{pf},\tilde{p}_\mathrm{pf},\tilde{p}_\mathrm{pf}\right)=\text{diag}\left(-\frac{\rho_\mathrm{pf}}{F^2},\frac{p_\mathrm{pf}}{F^2},\frac{p_\mathrm{pf}}{F^2},\frac{p_\mathrm{pf}}{F^2}\right),
\end{equation}
and $\tilde{T}^{(\phi)}_{\mu\nu}$ denotes the energy-momentum tensor of the scalar field, given by 
\begin{equation}\label{Scalar-Field-EM}
	\tilde{T}^{(\phi)}_{\mu\nu}=\tilde{\nabla}_\mu\phi \tilde{\nabla}_\nu\phi-\tilde{g}_{\mu\nu}\left( \frac{1}{2} \tilde{g}^{\alpha\beta}\tilde{\nabla}_\alpha\phi \tilde{\nabla}_\beta\phi+V(\phi)  \right). 
\end{equation}
In fact by varying the Einstein-frame action (\ref{EHaction}) with respect to $\tilde{g}_{\mu\nu}$, yields the scalar field energy-momentum tensor (\ref{Scalar-Field-EM}), whereas the perfect fluid contribution acquires the conformal factor $F^{-2}(\phi)$, leading to (\ref{SETFluid}). 

From the Bianchi identities and the Einstein equations (\ref{Einstein-eqs-EF}), it follows that the total energy-momentum tensor is conserved:
\begin{equation}\label{eq:cons_total}
\tilde{\nabla}_\mu \!\left( \tilde{T}^{(\phi)\mu}{}_{\nu} + \tilde{T}^{(\mathrm{pf})\mu}{}_{\nu}\right)=0.
\end{equation}
However, due to the non-minimal coupling between the scalar field and the matter, the individual components are not separately conserved. Taking into account the conformal dependence $F^{-2}(\phi)$ in the matter Lagrangian, one finds that
\begin{equation}
\tilde{\nabla}_\mu \tilde{T}^{(\phi)\mu}{}_{\nu}=\frac{1}{2}\beta(\phi)\tilde{g}^{\mu\sigma}\tilde{T}^{(\mathrm{pf})}_{\mu\sigma}\tilde{\nabla}_\nu \phi, \qquad \tilde{\nabla}_\mu \tilde{T}^{(\mathrm{pf})\mu}{}_\nu=-\frac{1}{2}\beta(\phi)\tilde{g}^{\mu\sigma}\tilde{T}^{(\mathrm{pf})}_{\mu\sigma}\tilde{\nabla}_\nu \phi,
\end{equation}
which describes the exchange of energy between the scalar field and the perfect fluid.

Furthermore, the variation of the action (\ref{EHaction}) with respect to the field $\phi$, considering the definition of the energy-momentum tensor, yields the modified non-linear wave-equation
\begin{equation}
\tilde{\Box}\phi+\lambda(\phi)V(\phi)-\frac{1}{2}\beta(\phi)\tilde{T}^{(\mathrm{pf})}=0,
\end{equation}
where $\tilde{\Box}\phi=\tilde \nabla_\mu\tilde \nabla^\mu \phi$  and $\tilde{T}^{(\mathrm{pf})}=\tilde{g}^{\mu\sigma}\tilde{T}^{(\mathrm{pf})}_{\mu\sigma}$ is the trace of the perfect fluid energy-momentum tensor with
\begin{equation}
	\lambda(\phi)  =-\frac{d\ln{(V)}}{d\phi}=-\frac{V_{\phi}}{V} \qquad\text{and}\qquad  \beta(\phi)  =\frac{d\ln{(F)}}{d\phi}= \frac{F_{\phi}}{F}.
\end{equation}
It turns out that the matter coupling is constant and given by
\begin{equation}
	\beta=\sqrt{\frac{2}{3}}
\end{equation}
and in general
\begin{equation}
	\lambda(\phi)=\left(\frac{RF-2f}{RF-f}\right)\frac{F_{,R}}{F}\frac{dR}{d\phi}. 
\end{equation}
If one restricts to spatially flat homogeneous and isotropic metrics $\tilde g = -d\tilde{t}^2+\tilde{a}^2(\tilde{t})\delta_{ij}d\tilde{x}^{i}d\tilde{x}^{j}$, the conformal transformation~(\ref{fr}) implies $d\tilde{t} = \sqrt{F}dt, \tilde{a} = \sqrt{F}a, \tilde{x}^{i}=x^{i}.$
It follows that
\begin{equation}\label{EinH}
\tilde{H} = \frac{1}{\tilde{a}}\frac{d\tilde{a}}{d\tilde{t}} = \frac{1}{\sqrt{F}}\left(H + \frac{F_{,R}\dot{R}}{2F}\right),
\end{equation}
where an overdot still signifies derivation with respect to the Jordan proper time $t$ and 
\begin{equation}
\tilde{q} = -1+\frac{3}{\left(2H + \frac{F_{,R}}{F}\dot{R}\right)^2} \left(\left(\frac{F_{,R}}{F}\dot{R}\right)^2+\frac{2(\rho_\mathrm{pf}+p_\mathrm{pf})}{3F}\right).    
\end{equation}
The above transformations yield the following evolution equations in the Einstein frame
\begin{subequations}\label{EFsystenm}
	\begin{align}
	\frac{d\tilde{a}}{d\tilde{t}} &= \tilde{H}\tilde{a}, \label{Edota}\\
	\frac{d\tilde{H}}{d\tilde{t}}  & = - \frac12\left(\left(\frac{d\phi}{d\tilde{t}}\right)^2 + \tilde{p}_\mathrm{pf} + \tilde{\rho}_\mathrm{pf}\right), \label{Ray1} \\
	\frac{d^2\phi}{d\tilde{t}^2}&=-3\tilde{H}\frac{d\phi}{d\tilde{t}}+\lambda(\phi)V(\phi)-\frac{3}{2}\beta(\phi) \left(\tilde{p}_\mathrm{pf}-\frac{\tilde{\rho}_\mathrm{pf}}{3}\right), \label{waveEq} \\
	\frac{d\tilde{\rho}_\mathrm{pf}}{d\tilde{t}} & = -3\tilde{H}\left(\tilde{p}_\mathrm{pf}+\tilde{\rho}_\mathrm{pf}\right) + \frac{1}{2}\beta(\phi)(3\tilde{p}_\mathrm{pf}-\tilde{\rho}_\mathrm{pf})\frac{d\phi}{d\tilde{t}}, \label{Edotrho}
	\end{align}
\end{subequations}
with the constraint
\begin{equation} \label{Gauss1} 
3\tilde{H}^2 =
\left[\frac{1}{2}\left(\frac{d\phi}{d\tilde{t}}\right)^2 + V(\phi)+\tilde{\rho}_\mathrm{pf}\right].
\end{equation}
Given a function $f(R)$, the conformal transformation~\eqref{EFscalarfield} uniquely determines the form of the potential $V(\phi)$ and therefore of the function $\lambda(\phi)$. Recall that for $f(R)$ gravity theories $\beta=\sqrt{2/3}$. Once an equation of state $\tilde{p}_\mathrm{pf}=\tilde{p}_\mathrm{pf}(\tilde{\rho}_{\mathrm{pf}})$ is given, then equations~\eqref{EFsystenm} form a closed system of evolution equations for $(\tilde{a},\tilde{H},\phi,\tilde{\rho}_\mathrm{pf})$, subject to the constraint~\eqref{Gauss1}.
\subsubsection{Monomial $f(R)$ gravity}
The choice
\begin{equation}
f(R)= \alpha^{p-1}R^{p},\quad \alpha>0,\quad p\in\mathbb{Z}\setminus\{0,1\} 
\end{equation}
implies that
\begin{equation}\label{Monomial}
\phi = \sqrt{\frac{3}{2}}\ln{\left(p (\alpha R)^{p-1}\right)}\qquad,\qquad V(\phi)=\frac{(p-1)\alpha^{p-2}}{2p^2}\left(pe^{-\sqrt{\frac{2}{3}}\phi}\right)^{\frac{p-2}{p-1}}
\end{equation}
The potential is positive and constant for $p=2$ or exponentially decreasing for $p>2$. This yields
\begin{equation}\label{lambdaRp}
\lambda=\sqrt{\frac{2}{3}}\left(\frac{p-2}{p-1}\right)\geq 0\quad \text{for}\quad p\geq 2.
\end{equation}
For $p<1$ the model leads to a negative potential, which is a well-known feature of $f(R)$ models in the Einstein frame. 
Note that General Relativity is recovered in the limit $p\rightarrow 1$. 

For monomials $f(R)$ gravity models $F(R) = p\alpha^{p-1} R^{p-1}$ and for $p$ even $F(R)<0$ whenever $R<0$. Moreover $F(R)=0$ at $R=0$, and in general $\dot{F}\Big|_{R=0}\neq0$, so that $F=0$ is not an invariant subset on the constrained state-space $(H,\dot{R},R,\rho_\mathrm{pf})$, meaning that there are solutions in the Jordan frame that are not globally mapped to the Einstein frame.

The goal in this paper is to give a global description of the solution space of the simplest quadratic $f(R)$ model, with $p=2$, and to describe the dynamical properties of all solutions, in both the Jordan and Einstein frames. In particular, we aim to identify the solutions that are conformally incomplete in the Einstein frame.
%
\section{Dynamics in the Jordan frame}\label{subsec:dynsysjord}
In this work we consider
\begin{equation}
f(R) = \alpha R^{2}, \qquad   \alpha>0.
\end{equation}
The matter content is assumed to be a perfect fluid with a linear equation of state
\begin{equation}
\label{pf}
p_\mathrm{pf} = (\gamma_\mathrm{pf} - 1)\rho_\mathrm{pf}, \qquad \frac{2}{3} <\gamma_\mathrm{pf} <2.
\end{equation}
For instance $\gamma_\mathrm{pf}=1$ and $\gamma_\mathrm{pf}=4/3$ correspond to dust matter and radiation, respectively. Considering \eqref{pf}, then equations~\eqref{DefH}, \eqref{DefR},~\eqref{11_eq} and~\eqref{Euler} reduce to
\begin{subequations}\label{Jordanorigeq}
	\begin{align}
	\dot{a} &= Ha, \label{aeq}\\
	\dot{H} &= -2H^2 + \frac{R}{6}, \label{dotH}\\
	\ddot{R}&=-3H\dot{R}-\frac{1}{2\alpha}\left(\gamma_\mathrm{pf}-\frac{4}{3}\right)\rho_\mathrm{pf}, \label{ddotR}\\
	\dot{\rho}_\mathrm{pf} &= -3\gamma_\mathrm{pf} H\rho_\mathrm{pf},\label{dotrho}
	\end{align}
	where $\dot{R}$ is regarded as an independent variable in order to obtain a first order system of evolution equations. The system must satisfy the constraint~\eqref{00_eq} which takes the form
	\begin{equation}
	-12 H  \big(\dot{R}+RH \big) + R^{2}+\frac{2}{\alpha}\rho_\mathrm{pf}=0.
	\label{constJor}
	\end{equation}
\end{subequations}
It follows that $\rho_\mathrm{pf} \propto a^{-3\gamma_\mathrm{pf}}$ and one can
therefore reduce the above system by solving either for $a$ in terms of $\rho_\mathrm{pf}$ or vice-versa. We choose to consider $\rho_\mathrm{pf}$ as variable. The equations~\eqref{dotH}, \eqref{ddotR} and~\eqref{dotrho}
then determine a closed system of first order equations for $(H,\dot{R},R,\rho_\mathrm{pf})$, which due to the
constraint~\eqref{constJor} yield a dynamical system describing a flow on a 3-dimensional state-space for each value of the parameter $\gamma_\mathrm{pf}$, while the vacuum state-space $\rho_\mathrm{pf}=0$ forms a 2-dimensional invariant boundary. 

Besides the trivial Minkowski solution 
\begin{equation}
\mathrm{M}:\qquad \rho_\mathrm{pf}=\dot{R}=R=H=0,
\end{equation}
the system~\eqref{Jordanorigeq} also admits the vacuum self-similar solution
\begin{equation}\label{SSsol}
\mathrm{R}:\qquad \rho_\mathrm{pf}=0,\qquad\dot{R}=0, \qquad R=0,\qquad H=\frac{1}{2t}
\end{equation}
for which $a(t)\propto \sqrt{t}$, $t\in(0,+\infty)$ and $q=1$, resembling the flat Friedmann-Lema\^itre solution with a radiation fluid of General Relativity. 
Vacuum quadratic $f(R)$ cosmologies also have explicit de-Sitter like solutions, 
\begin{equation}\label{dSsol}
\mathrm{L_{dS}}:\qquad\rho_\mathrm{pf}=0,\qquad\dot{R}=0,\qquad R=R_*,\qquad H=\sqrt{\frac{R_*}{12}},
\end{equation}
with $R_*$ a positive constant. In this case, the scale factor is simply $a(t)\propto e^{\sqrt{\frac{R_*}{12}}t}$, $t\in(-\infty,+\infty)$ and the deceleration parameter $q=-1$.
\subsection{Regular dynamical system on a compact state-space}
To obtain a regular system of equations on a compact state-space in the Jordan frame we follow closely the formulation in~\cite{ACU16}. 
As shown in~\cite{ACU16}, the vacuum state-space is a 2-dimensional double cone with the Minkowski solution as its apex, while the matter state-space consists of
the cone interior. This can be seen by making the following global (invertible) transformation
$(H,\dot{R},R) \rightarrow (\ts,\xs,R)$, 
\begin{subequations}
\begin{align}
H&=\sqrt{\frac{\alpha}{12}}(\ts-\xs), \label{var1.1} \\
\dot{R}+RH &=\frac{1}{\sqrt{12\alpha}}(\ts+\xs). \label{var2.1}	
\end{align}
\end{subequations}
which brings the constraint~\eqref{constJor} to the form
\begin{subequations}\label{tsxseq}
\begin{equation}\label{txconstr}
\ts^2 = \xs^2 + R^{2} + \frac{2}{\alpha}\rho_\mathrm{pf},
\end{equation}
while the evolution equations are given by
\begin{align}
\dot{\ts} &=\frac{1}{\sqrt{12\alpha}}\left[R-\alpha(\ts-\xs)^2-3\left(\gamma_\mathrm{pf}-\frac{2}{3}\right)\rho_\mathrm{pf}\right], \\
\dot{\xs} &=\frac{1}{\sqrt{12\alpha}}\left[-R+\alpha(\ts-\xs)^2-3\left(\gamma_\mathrm{pf}-\frac{2}{3}\right)\rho_\mathrm{pf}\right], \\
\dot{R} &= \frac{1}{\sqrt{12\alpha}}\left[\ts+\xs-\alpha(\ts-\xs)R\right],\label{dotR Eq} \\
\dot{\rho}_\mathrm{pf} &= -\frac{1}{\sqrt{12\alpha}}3\gamma_\mathrm{pf}\alpha(\ts-\xs)\rho_\mathrm{pf}.
\end{align}
\end{subequations}
The constraint equation~\eqref{txconstr} shows that the state-space consists of two disconnected invariant subsets ($\ts>0$ and $\ts<0$) and the Minkowski fixed point $\mathrm{M}$  $(\ts=\xs=R=\rho_\mathrm{pf}=0)$ of the dynamical
system~\eqref{tsxseq}. Furthermore, the system~\eqref{Jordanorigeq} is invariant under the transformation $(t,H) \rightarrow -(t,H)$ and, as a consequence, the system~\eqref{tsxseq} is invariant under the transformation $(t,\ts,\xs) \rightarrow -(t,\ts,\xs)$. Then it suffices to investigate the dynamics on the invariant upper part of the cone in order to obtain a complete picture of the dynamics. 

Since the conformal factor $F=2\alpha R$ is positive for $R>0$, only this half of the upper cone state-space is conformal to the Einstein frame. Moreover, from~\eqref{dotR Eq}, it follows that
\begin{equation}
\dot{F}\Big|_{F=0}=\sqrt{\frac{\alpha}{3}}(\ts+\xs)
\end{equation}
and hence $F=0$ (or $R=0$) is not an invariant subset except at $\ts=-\xs$. This orbit ending at $\mathrm{M}$ corresponds to the self-similar solution~\eqref{SSsol}, where the evolution of $\xs$ and $\ts$ describes the evolution of $H$. The de-Sitter solutions \eqref{dSsol} are now associated with the line of fixed points given by (see Figure~\ref{fig:ConeJF}):
\begin{equation}
\rho_\mathrm{pf}=0,\qquad R=R_*,\qquad \ts=\frac{1}{2}\sqrt{\frac{R_*}{\alpha}}(1+\alpha R_*),\qquad \xs=-\frac{1}{2}\sqrt{\frac{R_*}{\alpha}}(1-\alpha R_*)
\end{equation}
and $R_*$ a positive constant.
\begin{figure}
	\centering
		\includegraphics[trim={2.2cm 0.2cm 1.5cm 1cm},clip,width=0.3\textwidth]{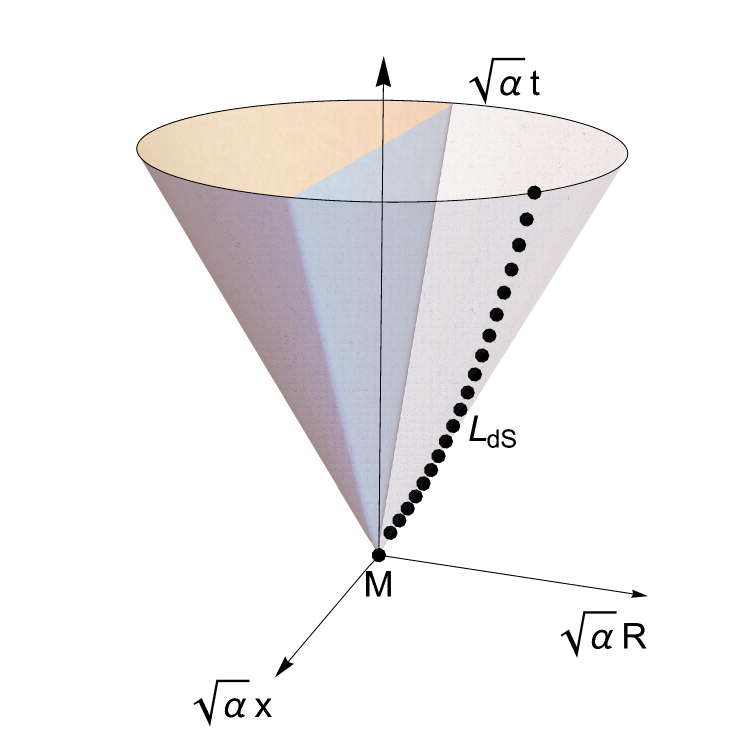}
	\caption{The future cone state-space of the Jordan frame. The shaded region with $R<0$ corresponds to a conformal factor $F<0$.}
	\label{fig:ConeJF}
\end{figure}

Now, removing the Minkowski fixed point from the state-space analysis allows us to introduce the following dimensionless bounded variables for the remaining non-trivial upper cone state-space:
\begin{equation}
(X,S,T,\Omega_\mathrm{pf}) = \left(\frac{\xs}{\ts},-\frac{R}{\ts},\frac{1}{1 +\alpha \ts},\frac{2\rho_\mathrm{pf}}{\alpha \ts^2}\right).
\end{equation}
Note that this amounts to a projection where all circles on the cone with constant $\ts$ now become the unit circle given by $X^{2} + S^2 =1$, where the different circles are parameterised by the value of $T$.

By introducing a new time variable $\bar{t}$ defined by
\begin{equation}
\label{tbar}
\frac{dt}{d\bar{t}} = \sqrt{12\alpha}T,
\end{equation}
and using the constraint
\begin{equation}
\Omega_\mathrm{pf} = 1- X^2 - S^2\,,
\end{equation}
to globally solve for $\Omega_\mathrm{pf}$, we obtain the following unconstrained regular dynamical system:
\begin{subequations}\label{dynsysTXS}
	\begin{align}
		X^{\prime} &= TS(1+X)+(1-T)(1-X)\left(1-X^2-\frac{3}{2}\left(\gamma_\mathrm{pf}-\frac{2}{3}\right)\left(1- X^2 - S^2\right)\right), \\
		S^{\prime} &= -(1+X-S^2)T-S(1-T)\left[(1-X)X-\frac{3}{2}\left(\gamma_\mathrm{pf}-\frac{2}{3}\right)\left(1- X^2 - S^2\right)\right], \\
		T^{\prime} & = T(1-T)\left[ST+(1-T)\left((1-X)^2+\frac{3}{2}\left(\gamma_\mathrm{pf}-\frac{2}{3}\right)(1-X^2-S^2)\right)\right],
	\end{align}
\end{subequations}
where the prime denotes derivatives with respect to $\bar{t}$. The above regular dynamical system is global (in the sense that the state-space covers
all of the $(H,\dot{R},R,\rho_\mathrm{pf})$ state-space, with the exception of the Minkowski fixed point) and forms our \emph{master equations} for dealing with the present models in the Jordan frame.

It is of interest to complement the above system with the auxiliary equation
\begin{equation}
\Omega^{\prime}_\mathrm{pf}=2 \Omega_\mathrm{pf} \left[ TS - \frac{3}{2}(1-T) \left( \left( \gamma_\mathrm{pf}-\frac{2}{3}+\frac{2}{3}X \right)\left( 1-X \right) - \ \left( \gamma_\mathrm{pf} -\frac{2}{3} \right)\Omega_\mathrm{pf}   \right) \right].
\end{equation}
The Jordan state-space $\mathbf{S}_\mathrm{J}$ consists of an open and bounded cylinder with unit radius: 
\begin{equation}
\mathbf{S}_\mathrm{J} = \{(X,S,T)\in\mathbb{R}^3: \Omega_\mathrm{pf}=1-X^2-S^2 > 0,\quad 0<T<1\}.
\end{equation}
The interior of the cylinder is an invariant subset that describes the matter solutions, $\Omega_\mathrm{pf}>0$, while the unit radius invariant boundary $\mathbf{S}_{\mathrm{Jvac}}$ describes the vacuum solutions $\Omega_\mathrm{pf}=0$ and gives a constrained system of equations. Moreover, we can extend the state-space to include the invariant boundaries $\{T=0\}$ and $\{T=1\}$ and obtain an extended compact state-space $\bar{\bf S}_\mathrm{J}=\mathbf{S}_\mathrm{J}\cup\partial\mathbf{S}_\mathrm{J}$, where $\partial\mathbf{S}_\mathrm{J}=\{\Omega_{\mathrm{pf}}=0\}\cup\{T=0\}\cup\{T=1\}$, see Figure~\ref{fig:CylJF}. This is crucial since all possible asymptotic states reside on these boundaries, as shown by the following Proposition:
\begin{proposition}\label{MonotoneF}
	The $\alpha$-limit sets for all interior orbits reside at the invariant boundary $\{T=0\}$, and the $\omega$-limit sets at the invariant boundaries $\{T=1\}$ and/or $\{\Omega_\mathrm{pf}=0\}$.
\end{proposition}
\begin{proof}
	Let $M: \mathbf{S}_\mathrm{J} \rightarrow (0,+\infty)$ be the $C^1$ function
	\begin{equation}
	M(X,S,T) = \frac{(1-T)^2\Omega_\mathrm{pf}}{T^2}=\frac{(1-T)^2}{T^2}(1-X^2-S^2)
	\end{equation}
	which satisfies 
	\begin{equation}
	M^\prime = -3\gamma_\mathrm{pf}(1-T)(1-X)M 
	\end{equation}
	and therefore, for $\gamma_{\mathrm{pf}}>0$, $M$ is strictly monotonically decreasing in the interior of the state-space. Since $M$ attains its minimum at the boundaries $\{T=1\}$ and $\{\Omega_\mathrm{pf}=0\}$ and tends to $+\infty$ as $T\rightarrow0$, the result follows by the \emph{monotonicity principle}.
\end{proof}
For completeness, we note that the original variables relate to our bounded variables as:
\begin{subequations}
	\begin{align}
	R &= -\frac{1}{\alpha}\frac{(1-T)S}{T},\quad 
	\dot{R} = \frac{1}{\alpha\sqrt{12\alpha}}\frac{(1-T)}{T}\left(1 + X + \frac{(1-T)(1-X)S}{T}\right), 
	\\
	H &= \frac{1}{\sqrt{12\alpha}}\frac{(1-T)(1-X)}{T}
	,\quad \rho_\mathrm{pf} =\frac{1}{2\alpha}\left(\frac{1-T}{T}\right)^2 \Omega_\mathrm{pf}
	\end{align}
	\label{Usual-Jordan}
\end{subequations}
and
\begin{subequations}
	\begin{align}
	X &=- \frac{H-\alpha(\dot{R}+HR)}{H+\alpha(\dot{R}+HR)} 
	,\quad S =-\sqrt{\frac{\alpha}{3}}\frac{ R}{\left(H+\alpha(\dot{R}+HR)\right)}
	\\
	T &= \frac{1}{1+\sqrt{3\alpha}\left(H+\alpha(\dot{R}+HR)\right)}
	,\quad \Omega_\mathrm{pf} = \frac{2\rho_\mathrm{pf}}{3\left(H+\alpha(\dot{R}+HR)\right)^2}.
	\end{align}
	\label{Jordan-Usual}
\end{subequations}
It is also worth noticing that the invariant boundary $\{T=0\}$ is associated with the limit $H\rightarrow +\infty$ and that $H$ is not only zero on the invariant boundary $\{T=1\}$ but also when $X=1$, which is not an invariant subset of the dynamical 
systems \eqref{dynsysTXS}. Therefore solution trajectories passing through $X=1$ will have $H=0$, 
and subsequently $H>0$ again, since $ X\leq1$.
%
\subsection{Fixed points}
The existence of a monotone function prevents the existence of interior fixed points, recurrent or periodic orbits, and implies that all attractor sets must reside on the invariant boundaries. 
The fixed points of system~\eqref{dynsysTXS} on the extended state-space can be found in Table~\ref{FP}. 
\begin{table}[ht!]
	\begin{center}
		\resizebox{\textwidth}{!}{
			\begin{tabular}{|c|c|c|c|c|c|c|c|}
				\hline            & & & & & & &\\ [-2ex]
				\begin{tabular}{c} {\bf Fixed}\\{\bf points}\end{tabular} & $X$ & $S$ & $T$ & $\Omega_\mathrm{pf}$ &{\bf Eigenvalues} & {\bf Eigenvectors} & {\bf Restrictions} \\ [1ex]
				\hline\hline      & & & & & & & \\[-2ex]
				$\ \mathrm{R}_0$    & $-1$ & $0$ & $0$ & $0$ & \begin{tabular}{c} $8(1-\frac{3}{4}\gamma_\mathrm{pf}),$\\  $2,$\\ $4$\\ \end{tabular} & \begin{tabular}{c} $e_1,$\\ $e_2,$\\ $e_3$ \end{tabular} & \\[1ex] \hline
				$\mathrm{N}_0$    & $1$ & $0$ & $0$ & $0$ & $0$, $0$, $0$ & &\\[1ex]\hline
				$\mathrm{N}_1$    & $-1$ & $0$ & $1$ & $0$ & $0$, $0$, $0$ & &\\[1ex]\hline
				$\mathrm{L}_\mathrm{R}$    & $X_*$ & $0$ & $0$ & $1-X^2_*$ & \begin{tabular}{c} $0,$\\ $1-X_*,$\\ $2(1-X_*)$ \end{tabular} & \begin{tabular}{c} $e_1,$\\ $e_2,$\\ $(1+X_*)e_2-(1-X_*)e_3$ \end{tabular}  &  \begin{tabular}{c}  $\gamma_\mathrm{pf}=\frac{4}{3},$\\ $X_*\in(-1,1)$\end{tabular}  \\ [1ex]\hline
				$\mathrm{L}_\mathrm{dS}$    & $X_0$ & $-\sqrt{1-X^2_0}$ & $\frac{(1-X_0)^2}{(1-X_0)^2+\sqrt{1-X^2_0}}$ & $0$ & $0$, $\lambda_+$, $\lambda_-$ &  & $X_0\in(-1,1)$ \\[1ex]
				\hline
		\end{tabular} }
	\end{center}\vspace{-0.5cm}
	\caption{Fixed  points of the dynamical system~\eqref{dynsysTXS} in the state-space $\bar{\bf S}_\mathrm{J}$. We defined $\lambda_{\pm}=(-b\pm\sqrt{b^2-4ac})/2a$, with $a=\left(S_0-(X_0-1)^2\right)^3$, $b=-3(\gamma_\mathrm{pf}+1)(X_0-1)^2 k$, $k=2(X_0-1)^2(X_0+1)-S_0\left( 2(1-X_0)+{X_0}^2(3-X_0) \right)$, $S_0=-\sqrt{1-{X_0}^2}$ and c such that $4ac=\gamma_\mathrm{pf}\left( 6(X_0-1)^2 k(X_0) \right)^2$. Since $-1< X_0< 1$, then $a<0$, $b<0$ and $4ac>0$.}
	\label{FP}
\end{table}

The physical interpretation of the fixed points can be inferred from the value of $q$ which, in terms of our state-space variables, is given by
\begin{equation}
q=1+2\frac{TS}{(1-T)(1-X)^2}.
\end{equation}

For $\gamma_\mathrm{pf}\neq 4/3$, in the extended state-space, there exists the isolated fixed point $\mathrm{R}_0$ at the intersection of invariant vacuum subset $\{\Omega_\mathrm{pf}=0\}$ with the $\{T=0\}$ invariant boundary,
and corresponds to the vacuum self-similar solution with $q=1$ given in~\eqref{SSsol}. 

If $\gamma_\mathrm{pf}<4/3$ the isolated fixed point $\mathrm{R}_0$ is a hyperbolic source and the $\alpha$-limit point of a two-parameter family of orbits into $\mathbf{S}_\mathrm{J}$. On the other hand if $\gamma_\mathrm{pf}>4/3$, then $\mathrm{R}_0$ is a hyperbolic saddle with unstable manifold given by the vacuum invariant set.

When $\gamma_\mathrm{pf}=4/3$ there exists an extra line of fixed points $\mathrm{L}_\mathrm{R}$ on the $\{T=0\}$ invariant boundary parametrised by $X_*\in(-1,1)$. The extension of the line to the vacuum invariant set with $X_*=\pm1$ are the fixed points $\mathrm{N}_0$ and $\mathrm{R}_0$, respectively. 
The linearisation around $\mathrm{L}_\mathrm{R}$ yields one zero eigenvalue and two positive eigenvalues. Moreover since 
the eigenvector associated with the zero eigenvalue points along the line, then $\mathrm{L}_\mathrm{R}$ is normally hyperbolic. Therefore each fixed point on the line is the $\alpha$-limit point of a 1-parameter family of interior orbits, and the whole line is the $\alpha$-limit set of a 2-parameter family of interior orbits. It also follows that, when $\gamma_\mathrm{pf}=4/3$, no interior orbit originates from $\mathrm{R}_0$, although $\mathrm{R}_0$ is the $\alpha$-limit point of a 1-parameter family of orbits in the vacuum boundary state-space $\mathbf{S}_{\mathrm{J}\mathrm{vac}}$, which will be analysed in detail in Section~\ref{sec-vacuum-boundary}. 

In addition to the hyperbolic fixed point $\mathrm{R}_0$ for $\gamma_\mathrm{pf}\neq 4/3$, and the normally hyperbolic set $\mathrm{L}_\mathrm{R}\cup\mathrm{R}_0$ when $\gamma_\mathrm{pf}=4/3$, there exists another line of fixed points $\mathrm{L}_\mathrm{dS}$ which corresponds to the de-Sitter solutions~\eqref{dSsol} with $q=-1$. The linearised matrix of the system around $\mathrm{L}_\mathrm{dS}$ has one zero eigenvalue and two negative eigenvalues. The eigenvector associated to the zero eigenvalue points along the lines and hence $\mathrm{L}_\mathrm{dS}$ is normally hyperbolic, meaning that each point on the line is the $\omega$-limit point of a 1-parameter family of orbits in $\mathbf{S}_\mathrm{J}$, and the whole line is the $\omega$-limit set of a 2-parameter family of interior orbits.

The extension of $\mathrm{L}_{\mathrm{dS}}$ to the invariant boundaries  $\{T=0\}$ and $\{T=1\}$ are the two fixed points $\mathrm{N}_0$ and $\mathrm{N}_1$, respectively. Both $\mathrm{N}_0$ and $\mathrm{N}_1$ have all eigenvalues zero and $q$ becomes undetermined. The blow-up of these two non-hyperbolic fixed points is performed in Section~\ref{sec-blow-up}.
%
\subsection{Global dynamics and the conformal region to the Einstein frame}
In this section we state our main results concerning the dynamics in the Jordan frame and situate the conformal region to the Einstein frame in our global Jordan state-space formulation. We start with the simple result:
\begin{lemma}
	The	conformal region to the Einstein frame $\{F>0\}\cap \mathbf{S}_\mathrm{J}$ is future-invariant.
\end{lemma}
\begin{proof}
	The conformal factor $F$ is given in terms of the global state-space variables by
	\begin{equation}
	F(X,S,T)=-2\frac{(1-T)}{T}S,
	\end{equation}
    and obeys
    \begin{equation}
    F^{\prime}=2\frac{(1-T)}{T}\left[(1+X)T+(1-X)S(1-T)\right].
    \end{equation}
	In the interior of the state-space, $F=0$ is given by the surface $S=0$ from which follows that 
	\begin{equation}
	F^\prime\Big|_{S=0}= 2(1-T)(1+X) >0
	\end{equation}
	and so $F=0$ is a future invariant surface for the flow in $\bf{ S}_\mathrm{J}$ in the direction of growing $F$.
\end{proof}	

Hence the open set of interior solutions in the Jordan frame crossing the surface $S=0$ are \emph{past conformally incomplete} in the Einstein frame. We now state our main results describing the Jordan frame dynamics: 
\begin{theorem}\label{GlobalTheoJordan}
	[Global interior dynamics in the Jordan frame] Consider the dynamical system~\eqref{dynsysTXS}: 
	\begin{itemize}
		\item[(i)] If $2/3<\gamma_\mathrm{pf}< 4/3$, then all interior orbits converge to $\mathrm{R}_0$ as $\bar{t}\rightarrow-\infty$, and to $\mathrm{L}_{\mathrm{dS}}$ (a 1-parameter set to each fixed point on $\mathrm{L}_{\mathrm{dS}}$) as $\bar{t}\rightarrow+\infty$.
		\item[(ii)] If $\gamma_\mathrm{pf}=4/3$, then all interior orbits converge to $\mathrm{L}_\mathrm{R}$ (a 1-parameter set from each point on $\mathrm{L}_\mathrm{R}$) as $\bar{t}\rightarrow-\infty$, and to $\mathrm{L}_{\mathrm{dS}}$ (a 1-parameter set to each fixed point on $\mathrm{L}_{\mathrm{dS}}$) as $\bar{t}\rightarrow+\infty$.
		\item[(iii)] If $4/3<\gamma_\mathrm{pf}< 2$, then all interior orbits converge to $\mathrm{N}_0$ as $\bar{t}\rightarrow-\infty$, and to $\mathrm{L}_{\mathrm{dS}}$ (a 1-parameter set to each fixed point on $\mathrm{L}_{\mathrm{dS}}$) as $\bar{t}\rightarrow+\infty$.
	\end{itemize}
\end{theorem}
\begin{proof}
	From Proposition~\ref{MonotoneF}, \emph{all} interior orbits originate from the invariant boundary $\{T=0\}$ and must end at the invariant boundaries $\{T=1\}$ and/or $\{\Omega_{\mathrm{pf}}=0\}$. The blow-up of $\mathrm{N}_0$ given in Proposition~\ref{BupN0}, together with the analysis of the flow on the 2-dimensional invariant boundary $\{T=0\}$ given in Lemma~\ref{T0SolStr}, then shows that all possible $\alpha$-limit sets are fixed points on $\{T=0\}$ and, thus, they must be the hyperbolic fixed point $\mathrm{R}_0$, the normally hyperbolic line of fixed points $\mathrm{L}_\mathrm{R}$ or the hyperbolic fixed point $\mathrm{P}$ on the blow-up of $\mathrm{N}_0$, depending on whether $\gamma_{\mathrm{pf}}<4/3$, $\gamma_{\mathrm{pf}}=4/3$ or $\gamma_{\mathrm{pf}}>4/3$, respectively, as follows from their local stability analysis. Conversely, the blow-up of the non-hyperbolic fixed point $\mathrm{N}_1$ given in Proposition~\ref{BupN1}, together with the description of the flow on the invariant boundaries $\{T=1\}$ and $\{\Omega_\mathrm{pf}=0\}$, given by Lemma~\ref{T1SolStr} and~\ref{Om0SolStr}, respectively, entails that the only possible $\omega$-limit set is the normally hyperbolic line of de-Sitter fixed points $\mathrm{L}_\mathrm{dS}$, where each point on the line attracts a 1-parameter family of interior orbits.
\end{proof}
As a corollary, and since \emph{all} interior orbits have as $\omega$-limit set the line of de-Sitter fixed points with $S<0$, in particular the orbits lying in the region $S>0$ ($F<0$), i.e. the shaded region in Figure~\ref{fig:CylJF}, must eventually cross the surface $S=0$.
\begin{figure}[!htb]
	\centering
		\subfigure[ $2/3<\gamma_\mathrm{pf}<4/3.$]{\label{fig:cilinder1}
	       \includegraphics[trim={3cm 4.2cm 5cm 5cm},clip,width=0.32\textwidth]{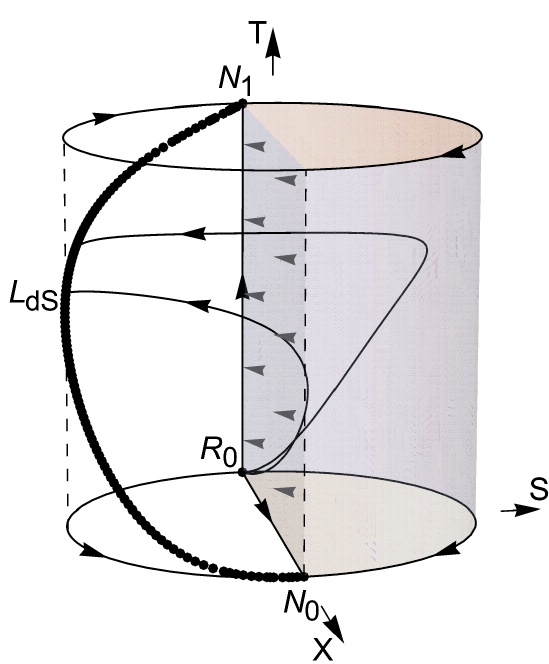}}
		\subfigure[ $\gamma_\mathrm{pf}=4/3$.]{\label{fig:cilinder2}
			\includegraphics[trim={3.2cm 4.2cm 5.2cm 5.2cm},clip,width=0.32\textwidth]{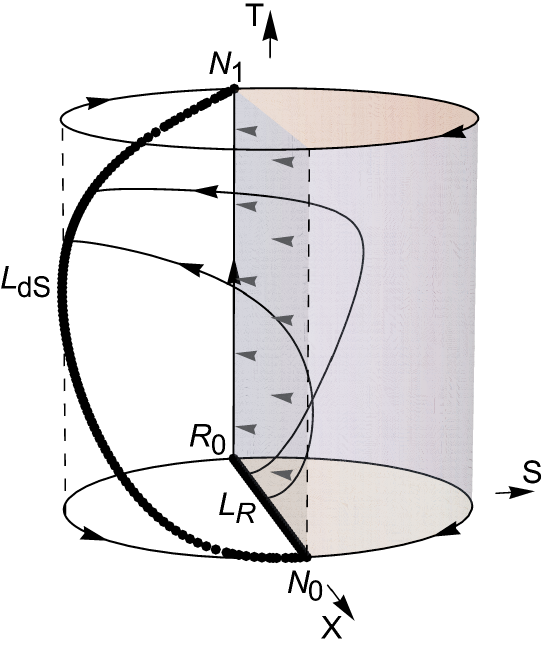}}
		\subfigure[ $4/3<\gamma_\mathrm{pf}<2$.]{\label{fig:cilinder3}
			\includegraphics[trim={3.3cm 4.6cm 5.6cm 5.6cm},clip,width=0.32\textwidth]{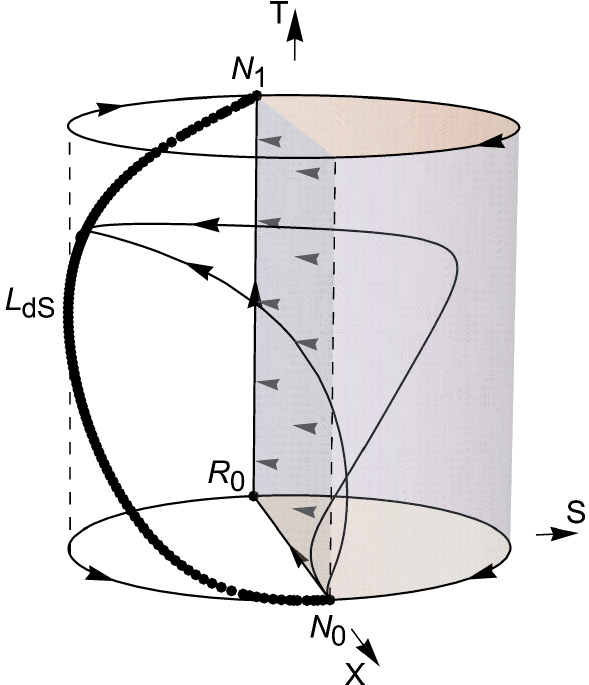}}	
	\caption{The global cylinder state-space in the Jordan frame. The shaded region corresponds to a conformal factor $F<0$.}
	\label{fig:CylJF}
\end{figure}
\subsection{Invariant boundary $\{T=0\}$}
On the $\{T=0\}$ invariant boundary the induced flow is given by
\begin{subequations}
	\begin{align}
	X^{\prime} &= (1-X)\left(1-X^2-\frac{3}{2}\left(\gamma_\mathrm{pf}-\frac{2}{3}\right)\left(1- X^2 - S^2\right)\right), \\
	S^{\prime} &= -S\left[(1-X)X-\frac{3}{2}\left(\gamma_\mathrm{pf}-\frac{2}{3}\right)\left(1- X^2 - S^2\right)\right],   
	\end{align}
\end{subequations}
and the state-space is the unit disk with boundary given by the vacuum invariant set $X^2+S^2=1$. Since the flow is invariant under the transformation $S\rightarrow -S$, the state-space is symmetric with respect to the $X$-axis and $S=0$ is an invariant subset. The intersection of $S=0$ with the invariant boundary $\{\Omega_\mathrm{pf}=0\}$ yields the two fixed points on $\{T=0\}$:  $\mathrm{R}_0$ and $\mathrm{N}_0$  which are the only fixed points of the above system if $\gamma_\mathrm{pf}\neq 4/3$. The fixed point $\mathrm{N}_0$ is non-hyperbolic while, on $\{T=0\}$, $\mathrm{R}_0$ is a source if $2/3<\gamma_\mathrm{pf}<4/3$ and a saddle if $4/3<\gamma_\mathrm{pf}<2$. If $\gamma_\mathrm{pf}=4/3$, then $\mathrm{R}_0$ has a zero eigenvalue due to the appearance of the normally hyperbolic line of fixed points $\mathrm{L}_\mathrm{R}$ whose extension to the vacuum invariant boundary $\{\Omega_\mathrm{pf}=0\}$ are the fixed points $\mathrm{R}_0$ and $\mathrm{N}_0$. 

The next lemma gives a characterisation of the flow in a neighbourhood of the non-hyperbolic fixed point $\mathrm{N}_0$ on the $\{T=0\}$ invariant boundary:
\begin{lemma}[Blow-up of $\mathrm{N}_0$ on $\{T=0\}$]\label{BUP_N0_T0}
	On the $\{T=0\}$ invariant boundary, the flow in a neighbourhood of $\mathrm{N}_0$ is as depicted in Figure~\ref{fig:BUPN0}, i.e. when $2/3<\gamma_\mathrm{pf}<4/3$, $\mathrm{N}_0$ is an isolated fixed point with a parabolic sector; When $4/3<\gamma_\mathrm{pf}<2$, $\mathrm{N}_0$ is an isolated fixed point with two elliptic sectors; When $\gamma_\mathrm{pf}=4/3$, $\mathrm{N}_0$ is the end point of the normally hyperbolic line of fixed points $\mathrm{L}_\mathrm{R}$, where the orbits originating from the line end at $\mathrm{N}_0$.
\end{lemma}
\begin{proof}
	The proof follows from the blow-up of $\mathrm{N}_0$ done in Section~\ref{BUP_N0} and given by Proposition~\ref{BupN0} when restricted to the invariant boundary $\{T=0\}$ (see Remark~\ref{BupN0T0}).
\end{proof}
\begin{figure}[!htb]
	\centering
		\includegraphics[width=0.965\textwidth, trim=1.1cm 7.1cm 17.5cm 0.7cm, clip]{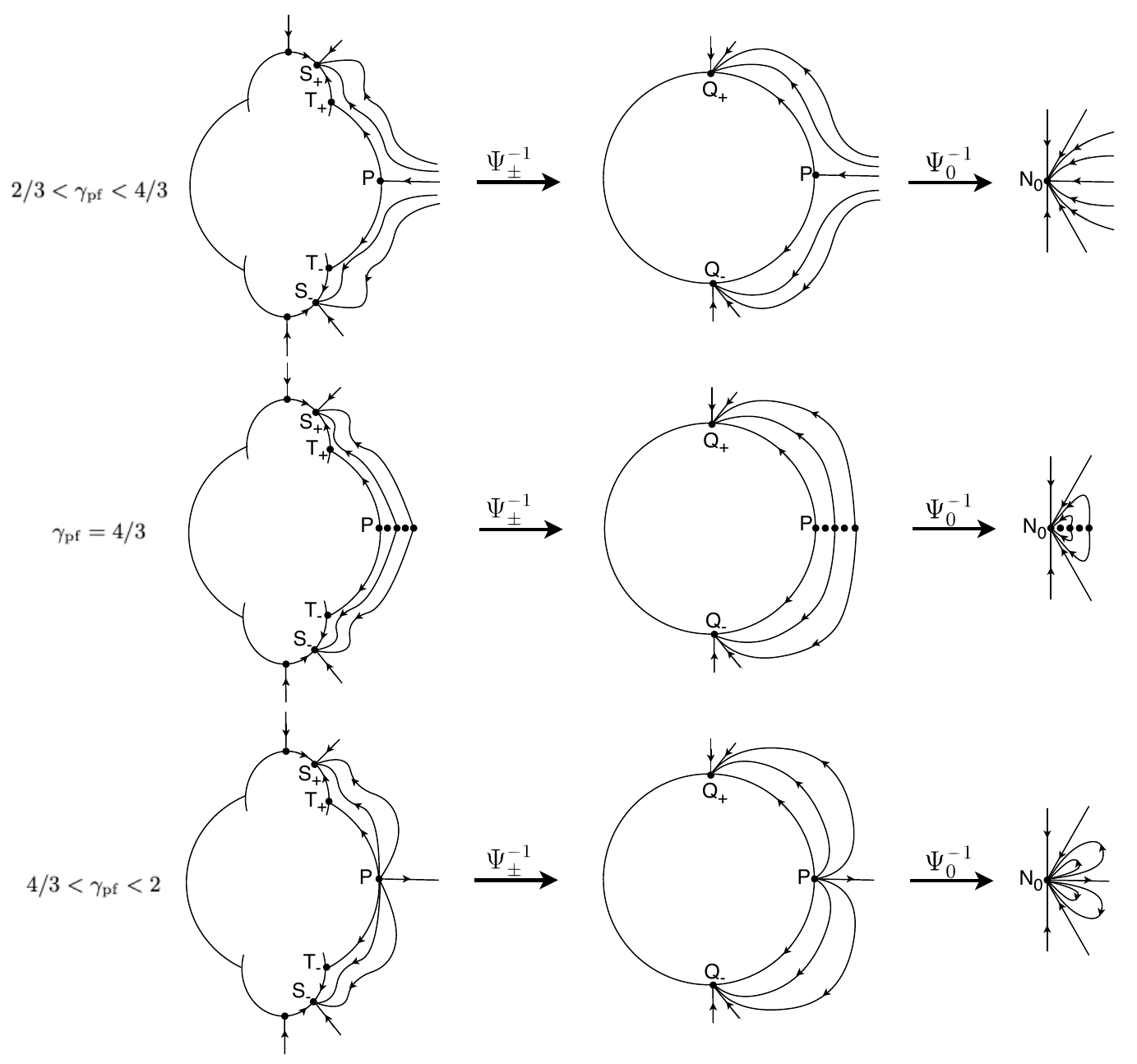}
	\caption{Successive blow-ups of $\mathrm{N}_0$ on $\{T=0\}$.}\label{fig:BUPN0}
\end{figure}
\begin{lemma}[Orbit structure of $\{T=0\}$]\label{T0SolStr}
	The orbit structure on the $\{T=0\}$ invariant boundary is as depicted in Figure \ref{fig:T0}.
\end{lemma}
\begin{proof}
	The 1-dimensional invariant boundary $\Omega_\mathrm{pf}=1-X^2-S^2=0$ consists of two heteroclinic orbits $\mathrm{R}_0\rightarrow\mathrm{N}_0$ as follows from $X^\prime|_{\Omega_\mathrm{pf}=0}=(1-X)(1-X^2)>0$. Moreover $S=0$ is also invariant for the flow and $X^\prime|_{S=0}=2(1-\frac{3}{4}\gamma_\mathrm{pf})(1-X)(1-X^2)$, which shows that this subset is a heteroclinic orbit $\mathrm{R}_0\rightarrow \mathrm{N}_0$ if $\gamma_\mathrm{pf}<4/3$ and $\mathrm{N}_0\rightarrow \mathrm{R}_0$ if $\gamma_\mathrm{pf}>4/3$, while the case $\gamma_\mathrm{pf}=4/3$ results in the line of fixed points $\mathrm{L}_\mathrm{R}$. Since on each invariant subset $S>0$ and $S<0$ there are no interior fixed points, it then follows by the Poincar\'e-Bendixson theorem, together with the hyperbolicity of $\mathrm{R}_0$, the normal hyperbolicity of $\mathrm{L}_\mathrm{R}$ and Lemma~\ref{BUP_N0_T0}, that each of these invariant subsets consist of:  heteroclinic orbits $\mathrm{R}_0\rightarrow \mathrm{N}_0$ if $\gamma_\mathrm{pf}<4/3$; heteroclinic orbits  originating from the normally hyperbolic line $\mathrm{L}_\mathrm{R}$ (a single orbit from each fixed point on the line) and ending at $\mathrm{N}_0$ if $\gamma_\mathrm{pf}=4/3$; homoclinic orbits of $\mathrm{N}_0$ if $\gamma_{\mathrm{pf}}>4/3$.
\end{proof}
\begin{figure}[!htb]
	\centering
		\subfigure[$2/3<\gamma_\mathrm{pf}<4/3$.]{\label{fig:DustU1}
			\includegraphics[trim={1.4cm 2cm 0.1cm 0.2cm},clip,width=0.31\textwidth]{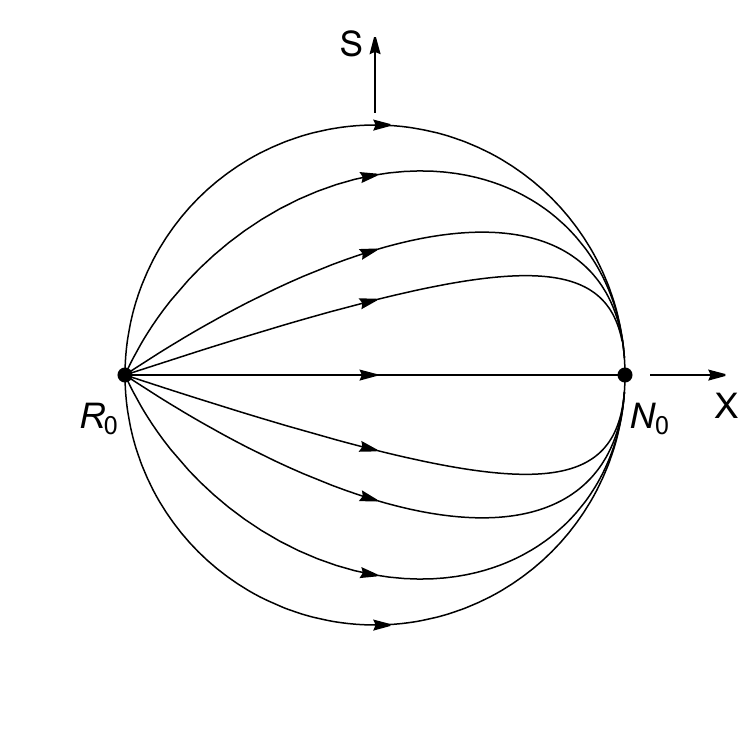}}  
		\subfigure[$\gamma_\mathrm{pf}=\frac{4}{3}$.]{\label{fig:RadiationU1}
			\includegraphics[trim={1.4cm 2cm 0.1cm 0.2cm},clip,width=0.31\textwidth]{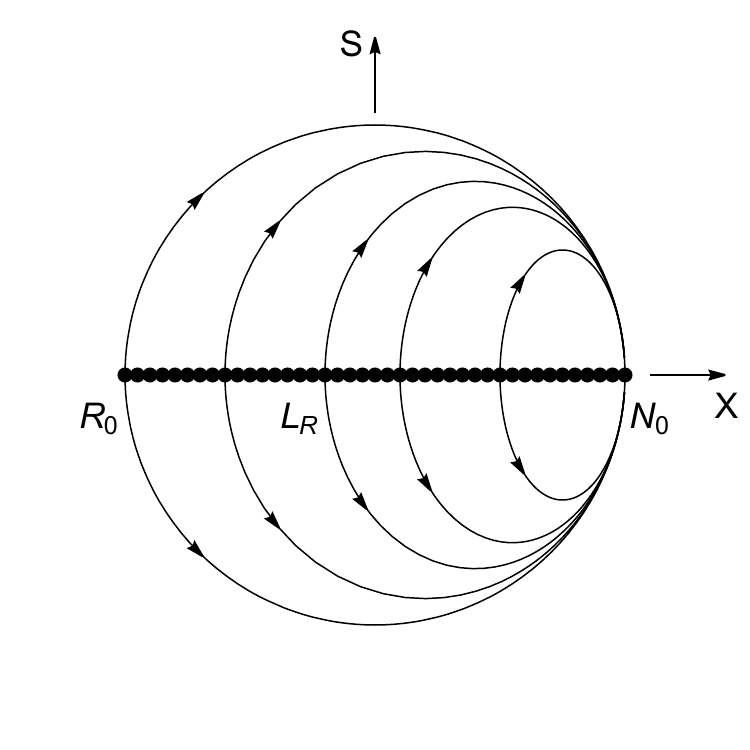}} 
		\subfigure[$4/3<\gamma_\mathrm{pf}<2$.]{\label{fig:32U1}
			\includegraphics[trim={1.4cm 2cm 0.1cm 0.2cm},clip,width=0.31\textwidth]{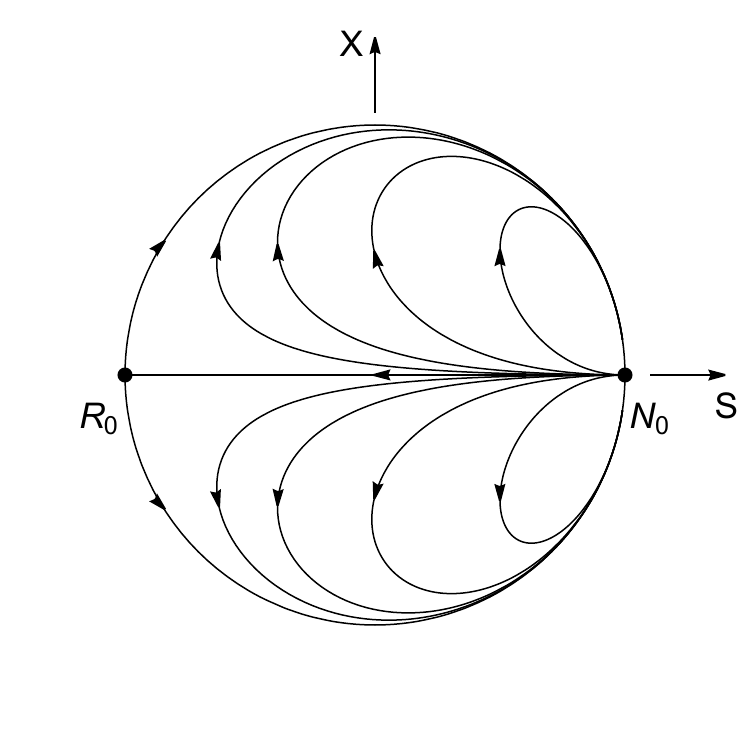}}  
	\caption{The $\{T=0\}$ invariant boundary subset of $\bar{\bf S}_\mathrm{J}$.}\label{fig:T0}
\end{figure}
%
\subsection{Invariant boundary $\{T=1\}$}
The flow induced on the invariant boundary $\{T=1\}$ is simply
\begin{subequations}
	\begin{align}
	X^{\prime} &= S(1+X), \\
	S^{\prime} &= -(1+X-S^2)
	\end{align}
\end{subequations}
and the state-space is the unit disk with boundary given by the vacuum invariant set $X^2+S^2=1$. The above system has a single fixed point $\mathrm{N}_1$ located at the vacuum invariant boundary. The linearisation around $\mathrm{N}_1$ yields two zero eigenvalues, so that $\mathrm{N}_1$ is non-hyperbolic and a characterisation of the flow in a neighbourhood of $\mathrm{N}_1$ is given by the next lemma:
\begin{lemma}[Blow-up of $\mathrm{N}_1$ on $\{T=1\}$]\label{BUP_N1_T1}
	On the invariant boundary $\{T=1\}$ the flow in a neighbourhood of $\mathrm{N}_1$ is as depicted in Figure~\ref{fig:BUPN1T1}, i.e. $\mathrm{N}_1$ has one elliptic sector.
\end{lemma}
\begin{proof}
The proof follows from the full blow-up of $\mathrm{N}_1$ given by Proposition~\ref{BupN1} in Section~\ref{sec:BUP_N1}, when restricted to the invariant boundary $\{T=1\}$ (see Remark~\ref{BupN1T1}).
\end{proof}
\begin{figure}[!htb]
	\centering		
		\includegraphics[width=0.6\textwidth, trim=2.1cm 13.2cm 16.7cm 1.1cm, clip]{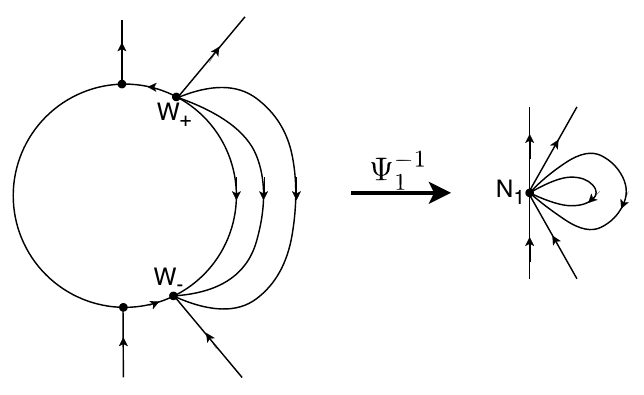}
	\caption{Blow-up of $\mathrm{N}_1$ on $\{T=1\}$.}\label{fig:BUPN1T1}
\end{figure}
\begin{lemma}[Orbit structure on $\{T=1\}$]\label{T1SolStr}
	The invariant boundary $\{T=1\}$ consists of homoclinic orbits of $\mathrm{N}_1$ as depicted in Figure~\ref{fig:T1}.
\end{lemma}
\begin{proof}
	Since $\mathrm{N}_1$ is located on the invariant boundary $\{\Omega_{\mathrm{pf}}=0\}$ is the only fixed point of the system, the proof follows by Lemma~\ref{BUP_N1_T1} and the use of the Poincar\'e-Bendixson theorem.
\end{proof}
\begin{figure}[ht!]
	\centering
		\includegraphics[trim={1.5cm 2.1cm 0.1cm 0.5cm},clip,width=0.35\textwidth]{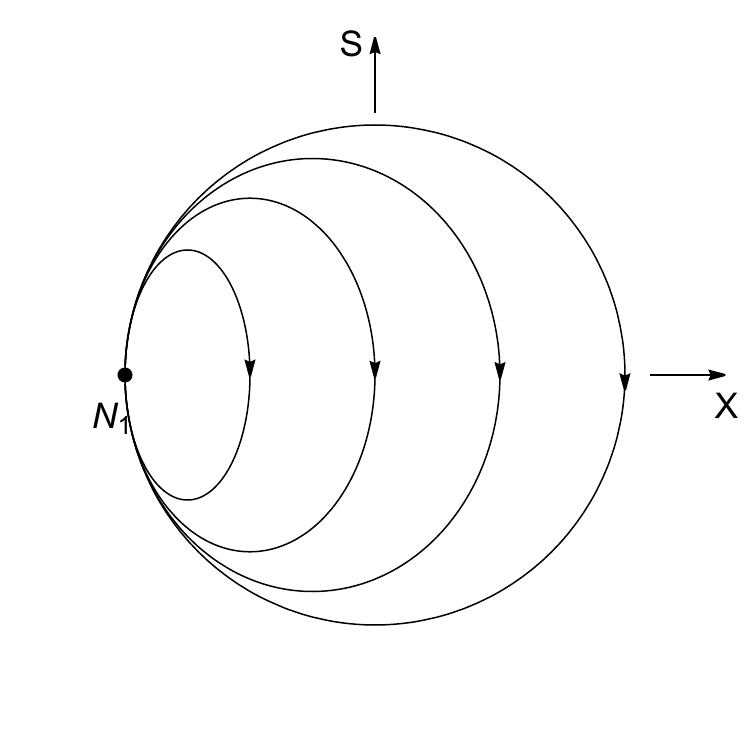} 
	\caption{The $\{T=1\}$ invariant boundary of $\bar{\mathbf{S}}_\mathrm{J}$.}\label{fig:T1}
\end{figure}
%
\subsection{Vacuum invariant boundary $\{\Omega_\mathrm{pf}=0\}$}
\label{sec-vacuum-boundary}
The induced flow on the vacuum invariant boundary, defined by $X^2+S^2=1$, is a constraint system of equations. The constraint equation can be globally solved by introducing an angular variable $\theta$ as
\begin{equation}
X = \cos{\theta},\qquad S = \sin{\theta},
\end{equation}
which results in the following regular unconstrained system of equations for the state vector $(\theta,T)$:
\begin{subequations}\label{Jordan2DDynSys}
	\begin{align}
	T^{\prime} &= T(1 - T)\left[T\sin{\theta} + (1-T)(1 - \cos{\theta})^2\right], 
	\\
	\theta^{\prime} &= -\left[T(1 + \cos{\theta}) + (1-T)(1 - \cos{\theta})\sin{\theta}\right].
	\end{align}
\end{subequations}
The above system of equations admits the set of interior fixed points $\mathrm{L}_\mathrm{dS}$ characterised by the values $(\theta_0,T_0)$ given by the solutions to the cubic equation 
\begin{equation}
T_0\left(1+\tan^2{\left(\frac{\theta_0}{2}\right)}\right)+2(1-T_0)\tan^3{\left(\frac{\theta_0}{2}\right)}=0,
\label{line}
\end{equation}
with $\theta_0\in(-\pi\pm 2n\pi,\pm 2n\pi)$, $n=0,1,2,...$, and that passes through $(\theta_0,T_0)=(-\pi/2\pm 2n\pi,1/2)$. The line $\mathrm{L}_\mathrm{dS}$ is normally hyperbolic being a $\omega$-limit set for interior orbits, while its extension to the invariant boundaries $\{T=0\}$ and $\{T=1\}$ result in the non-hyperbolic fixed points $\mathrm{N}_0$ with $\theta=-\pi\pm 2n\pi$ and $\mathrm{N}_1$ with $\theta=\pm 2n\pi$, respectively. Finally, the system has a single isolated hyperbolic fixed point $\mathrm{R}_0$ located at $\{T=0\}$, with $\theta=-\pi\pm 2n\pi$, which is a source.
\begin{lemma}[Blow-up of $\mathrm{N}_0$ on $\{\Omega_\mathrm{pf}=0\}$]\label{BUP_N0_Om0}
	On the invariant boundary $\{\Omega_\mathrm{pf}=0\}$, the flow in a neighbourhood of $\mathrm{N}_0$ is as depicted in Figure~\ref{fig:BUPN0Omega0}, i.e. no interior orbit in $\bar{\mathbf{S}}_{\mathrm{Jvac}}$ converges to $\mathrm{N}_0$.
\end{lemma}
\begin{proof}
	The proof follows from the full blow-up of $\mathrm{N}_0$ performed in Section~\ref{BUP_N0} (see Figure~\ref{fig:BUP3D_N0}), when restricted to the invariant boundary $\{\Omega_{\mathrm{pf}}=0\}$, see Remark~\ref{BupN0T0}.
\end{proof}
\begin{figure}[ht!]
	\centering
		\includegraphics[width=1\textwidth, trim=1.5cm 15.5cm 5.5cm 1.2cm, clip]{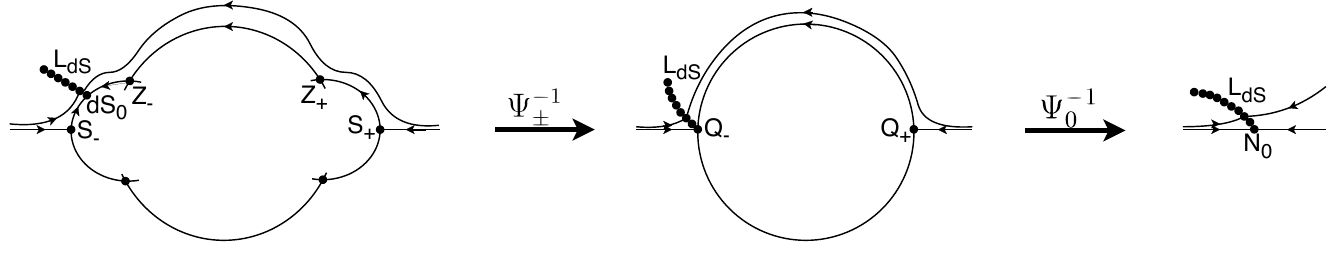}
	\caption{Successive blow-ups of $\mathrm{N}_0$ on $\bar{\mathbf{S}}_\mathrm{Jvac}$.}
	\label{fig:BUPN0Omega0}
\end{figure}
\begin{lemma}[Blow-up of $\mathrm{N}_1$ on $\{\Omega_\mathrm{pf}=0\}$]\label{BUP_N1_Om0}
	On the invariant boundary $\{\Omega_\mathrm{pf}=0\}$, the flow in a neighbourhood of $\mathrm{N}_1$ is as depicted in Figure~\ref{fig:BUPN1Omega0}, i.e. a unique interior orbit in $\bar{\mathbf{S}}_{\mathrm{Jvac}}$ converges to $\mathrm{N}_1$. 
\end{lemma}
\begin{proof}
	The proof follows from the full blow-up of $\mathrm{N}_1$ given by Proposition~\ref{BupN1} in Section~\ref{sec:BUP_N1} (see Figure~\ref{fig:BUP3D_N1}), when restricted to the invariant boundary $\{\Omega_{\mathrm{pf}}=0\}$, see Remark~\ref{BupN1Om0}.
\end{proof}
	\begin{figure}[ht!]
	\centering
		\includegraphics[width=0.7\textwidth, trim=1.5cm 15.55cm 16.1cm 0.7cm, clip]{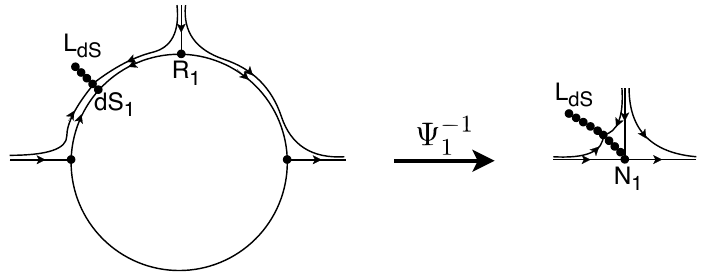}
	\caption{Blow-up of $\mathrm{N}_1$ on $\bar{\mathbf{S}}_\mathrm{Jvac}$.}
	\label{fig:BUPN1Omega0}
\end{figure}
\begin{lemma}[Orbit structure on $\{\Omega_\mathrm{pf}=0\}$]\label{Om0SolStr}
	The flow on the invariant boundary $\{\Omega_\mathrm{pf}=0\}$ is as depicted in Figure~\ref{fig:VBJF}, i.e. except for the single heteroclinic orbit $\mathrm{R}_0\rightarrow \mathrm{N}_1$ with constant $\theta=-\pi\pm2n\pi$, $n=0,1,2,...$, all interior orbits originate at $\mathrm{R}_0$ and end at $\mathrm{L}_{\mathrm{dS}}$ (each point on $\mathrm{L}_{\mathrm{dS}}$ is the $\omega$-limit point of two interior orbits).
\end{lemma}
\begin{proof}
At $\theta=-\pi\pm2n\pi$, $n=0,1,2,...$, we have that 
\begin{equation}
T^\prime|_{\theta=-\pi\pm 2n\pi}=2T(1-T)^2>0, \qquad \theta^\prime|_{\theta=-\pi\pm 2n\pi}=0,
\end{equation}
which consists of a heteroclinic orbit $\mathrm{R}_0\rightarrow \mathrm{N}_1$. Hence the vacuum invariant boundary state-space $\bar{\mathbf{S}}_{\mathrm{J}\text{vac}}$ is the union of two invariant subsets delimited by this heteroclinic orbit and the line of de-Sitter fixed points $\mathrm{L}_{\mathrm{dS}}$. On each of these invariant subsets there are no interior fixed points and hence, by the Poincaré-Bendixson theorem, all possible limit sets must be located on the boundaries of those invariant subsets. Since $\mathrm{R}_0$ is a hyperbolic source, it is the $\alpha$-limit point for a 1-parameter set of interior orbits, while $\mathrm{L}_\mathrm{dS}$, being normally hyperbolic, is the $\omega$-limit set of a 1-parameter set of interior orbits from each invariant subset, so each point on the line $\mathrm{L}_{\mathrm{dS}}$ is the $\alpha$-limit point of a single interior orbit on each invariant subset.
Finally, by the two previous lemmas~\ref{BUP_N0_Om0} and~\ref{BUP_N1_Om0}, no interior orbit in $\mathbf{S}_{\mathrm{J}\text{vac}}$ converges to the fixed point $\mathrm{N}_0$ and there is a unique interior orbit that has $\mathrm{N}_1$ as the $\omega$-limit point, which corresponds to the straight heteroclinic orbit $\mathrm{R}_0\rightarrow\mathrm{N}_1$ with constant $\theta=-\pi\pm 2n\pi$, $n=0,1,2,...$.
\end{proof}

\begin{figure}[ht!]
	\centering
		\includegraphics[trim={0cm 2.2cm 0cm 1.7cm},clip,width=0.5\textwidth]{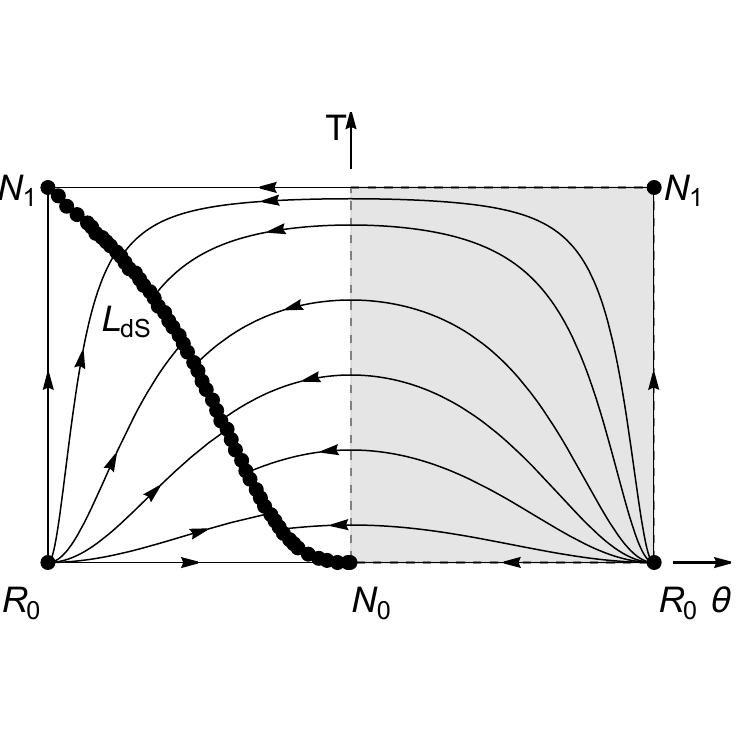}	
	\caption{The extended unwrapped vacuum invariant boundary state-space $\bar{\mathbf{S}}_\mathrm{Jvac}$. The shaded region has $F<0$ and is therefore not mapped to the Einstein frame state-space.}
	\label{fig:VBJF}
\end{figure}
The conformal factor $F=2\alpha R$ is given in terms of our state-space variables by
\begin{equation}
F(\theta,T)=-2\frac{(1-T)}{T}\sin{\theta}.
\end{equation}
We note that the dynamics in the Einstein frame is conformal to the region of the Jordan state-space given by $F>0$, i.e., $R>0$ and hence $S<0$, which is the strip $(\theta,T)\in(-\pi\pm 2n\pi,\pm 2n\pi)\times(0,1)$, $n=0,1,2,...$. The shaded region in Figure~\ref{fig:VBJF} has $F<0$ which is the interior region of the Jordan state-space that is not conformally mapped to the Einstein frame state-space. Moreover $F(\theta,T)$ satisfies
\begin{equation}
F^{\prime}=2\frac{(1-T)}{T}\left[(1+\cos{\theta})T+(1-T)(1-\cos{\theta})\sin{\theta}\right]
\end{equation}
and hence at $F=0$,  it follows that in the vacuum boundary ${\bf{ S}}_{\text{Jvac}}$,
\begin{equation}
\left(\frac{dF}{d\bar{t}}\right)_{\theta=\pm 2n\pi}= 4(1-T) >0, \qquad \left(\frac{dF}{d\bar{t}}\right)_{\theta=-\pi\pm 2n\pi}= 0.
\end{equation}
So $F=0$ is a future invariant surface for the flow in ${\bf{ S}}_{\text{Jvac}}$ in the direction of growing $F$ at $\theta=\pm 2n\pi$ while, at $\theta=-\pi\pm2n\pi$, it consists of the heteroclinic orbit $\mathrm{R}_0\rightarrow \mathrm{N}_1$, corresponding to the self-similar solution~\eqref{SSsol}
, where the evolution in $T$ describes the evolution in $H$ since $R=\dot{R}=0$.
Hence $\mathrm{R}_0\rightarrow \mathrm{N}_1$ splits the solutions which are globally mapped to the Einstein frame and those that are not. In fact the set of solutions that are not initially mapped to the Einstein frame, i.e. those lying in the shaded region in Figure~\ref{fig:VBJF}, must cross $\theta=\pm 2n\pi$, $n=0,1,2,...$. 
\section{Blow-up of non-hyperbolic fixed points}
\label{sec-blow-up}

\subsection{Blow-up of the fixed point $\mathrm{N}_0$}
\label{BUP_N0}
In this section we derive the blow-up of the non-hyperbolic fixed point $\mathrm{N}_0$. The starting point is to introduce the local variables 
\begin{equation}
(\bar{X},\bar{S},\bar{T})=\left(1-X,S,\frac{T}{1-T}\right),
\end{equation} 
so that $\mathrm{N}_0$ is now located at the origin of coordinates. Then, after changing the time variable to $N$ defined by
\begin{equation}
\frac{d}{dN}=\frac{1}{(1-T)}\frac{d}{d\bar{t}}=\sqrt{12\alpha} \bar{T}\frac{d}{dt}
\end{equation}
the system of equations~\eqref{dynsysTXS} results in 
\begin{subequations}
	\begin{align}
	\frac{d\bar{X}}{dN} &=-\bar{T}S(2-\bar{X})-\bar{X}\left(\bar{X}(2-\bar{X})+\left(1-\frac{3}{2}\gamma_\mathrm{pf}\right)(\bar{X}(2-\bar{X})-S^2)\right), \\
	\frac{dS}{dN} &=-(2-\bar{X}-S^2)\bar{T} -S\left[\bar{X}(1-\bar{X})+\left(1-\frac{3}{2}\gamma_\mathrm{pf}\right)\left(\bar{X}(2-\bar{X}) - S^2\right)\right], \\
	\frac{d\bar{T}}{dN} &=  \bar{T}\left[S\bar{T}+\bar{X}^2-\left(1-\frac{3}{2}\gamma_\mathrm{pf}\right)\left(\bar{X}(2-\bar{X}) - S^2\right)\right].
	\end{align}
\end{subequations}
We now employ the spherical blow-up method, i.e., we transform the fixed point at the origin to the unit 2-sphere 
$\mathbb{S}^2_0=\left\{(x,y,z): x^2+y^2+z^2=1\right\}$ 
and define the blow-up space manifold $\mathcal{B}_0\:=\mathbb{S}^2_0\times[0,\bar{u}^{(0)}]$ for some fixed $\bar{u}^{(0)}>0$. We further define the quasi-homogeneous blow-up map
\begin{equation}
\Psi_0\,:\quad \mathcal{B}_0\rightarrow \mathbb{R}^3,\qquad \Psi_0(x,y,z,u)=(ux,uy,u^2 z),
\end{equation} 
which after cancelling a common factor $u$, i.e. by changing the time variable to $\tau$ defined by $d/d\tau= u^{-1}d/dN$, leads to a desingularisation of the non-hyperbolic fixed point on the blow-up locus $\{u=0\}$. If there are still degenerated fixed points, then further blow-ups are needed. Since $\Psi_0$ is a diffeomorphism outside of the sphere $\mathbb{S}^2_0\times\{u=0\}$, which corresponds to the fixed point $(0,0,0)$, the dynamics on the blow-up space $\mathcal{B}_0\setminus \{\mathbb{S}^2_0\times \{u = 0\}\}$ is topological conjugate to $\mathbb{R}^3\setminus \{0,0,0\}$. 
Instead of using standard spherical coordinates on $\mathcal{B}_0$, it usually simplifies the computations if one uses different local charts $\kappa^{(0)}_{i} \,:\,\mathcal{B}_0\rightarrow\mathbb{R}^3$ and define the directional blow-up maps $\psi^{(0)}_{i}\,:\, \Psi_0\circ (\kappa^{(0)}_{i})^{-1}$ for which the resulting state vectors are simpler to analyse. 

Since the original extended state-space is defined by $\Omega_{\mathrm{pf}}\geq0$ and $T\in[0,1]$, then only the region of physical relevance corresponds to $\bar{X}\geq0$ and $\bar{T}\geq0$, and we just need to consider four charts $\kappa^{(0)}_{i}$ such that
\begin{subequations}
\label{ki}
	\begin{align}
	\psi_{(0)1} &=(u_{1}, u_{1} y_{1},u^2_1 z_1), \\
	\psi_{(0)2\pm} &=(u_{2\pm}x_{2\pm},\pm u_{2\pm},u^2_{2\pm} z_{2\pm}), \\
	\psi_{(0)3} &=(u_{3}x_{3},u_{3}y_3,u^2_3),
	\end{align}
\end{subequations}
where $\psi^{(0)}_{1}$, $\psi^{(0)}_{2+}$ and $\psi^{(0)}_{3}$ are called the directional blow-ups in the positive $x$, $y$ and $z$ directions, respectively, and $\psi^{(0)}_{2-}$ is the blow-up in the negative $y$-direction. 

We start the analysis by using chart $\kappa^{(0)}_{3}$, i.e. the directional blow-up in the positive $z$-direction. In these coordinates the equator of the sphere is located at infinity, which will be analysed using both charts $\kappa^{(0)}_{1}$ and $\kappa^{(0)}_{2}$. Using the transition charts $\kappa^{(0)}_{ij}=\kappa^{(0)}_{j}\circ(\kappa^{(0)}_{i})^{-1}$ we can then identify special invariant subsets and obtain a global picture of the blow-up solution space.
%
\paragraph{Blow-up in the positive $z$-direction.}
The blow-up in the positive $z$-direction amounts to the change of coordinates $(\bar{X},\bar{S},\bar{T})=(u_3 x_3,u_3 y_3,u^2_3)$, which after canceling a common factor $u_3$, i.e. by changing the time variable $d/d\bar{t} = u_3 d/d\tau_3$ yields the regular dynamical system
\begin{subequations}
	\begin{align}
	\frac{dx_3}{d\tau_3} &=-2y_3u_3-x_3\left[2x_3-\frac{1}{2}(u^2_3 y_3+u_3 x^2_3)+\frac{1}{2}\left(1-\frac{3}{2}\gamma_\mathrm{pf}\right)\left(x_3 (2-u_3x_3)-u_3y^2_3\right)\right] , \\
	\frac{dy_3}{d\tau_3} &= -2+u_3 x_3-y_3\left[2\left(1-\frac{3}{4}\gamma_\mathrm{pf}\right)x_3-\frac{1}{2}(u^2_3 y_3+u_3x^2_3)-\frac{1}{2}\left(1-\frac{3}{2}\gamma_\mathrm{pf}\right)u_3(x^2_3+y^2_3) \right], \\
	\frac{du_3}{d\tau_3} &=\frac{u_3}{2}\left[u^2_3 y_3+u_3 x^2_3-\left(1-\frac{3}{2}\gamma_\mathrm{pf}\right)\left(x_3 (2-u_3x_3)-u_3y^2_3\right)\right].
	\end{align}
\end{subequations}
Since the physical region of the state-space defined by $\Omega_{\mathrm{pf}}>0$ has $\bar{X}>0$, then we are only interested in the invariant region $x_3>0$ of the blow-up space, where the intersection of the invariant boundary $\{\Omega_{\mathrm{pf}}=0\}$ with the invariant set $\{u_3=0\}$ is just the invariant $y_3$-axis, i.e. $\{x_3=0\}$. 
\begin{lemma}\label{lemk3}
On the half-plane $\{(x_3,y_3,u_3): u_3=0 \wedge x_3\geq0\}$ there are no fixed points, recurrent or periodic orbits, and all possible $\alpha$ and $\omega$-limit sets are located at infinity, as depicted in Figure~\ref{fig:blowupzdS0}.
\end{lemma}
\begin{proof}
On the $\{u_3=0\}$ invariant subset the induced flow is given by
\begin{equation}
\frac{dx_3}{d\tau_3}=-\frac{3}{2}(2-\gamma_\mathrm{pf})x^2_3, \quad \frac{dy_3}{d\tau_3} = -2 -2 \left(1-\frac{3}{4}\gamma_\mathrm{pf}\right)x_3 y_3.
\end{equation}
Since $\gamma_{\mathrm{pf}}<2$, then $x_3$ is strictly monotonically decreasing except at the invariant $1$-dimensional boundary subset $\{x_3=0\}$ where $\frac{dy_3}{d\tau_3} = -2<0$. The proof follows by the \emph{monotonicity principle}.
\end{proof}	
\begin{figure}[ht!]
	\centering
		\includegraphics[trim={1.2cm 0.1cm 0.2cm 0.7cm},clip,width=0.22\textwidth]{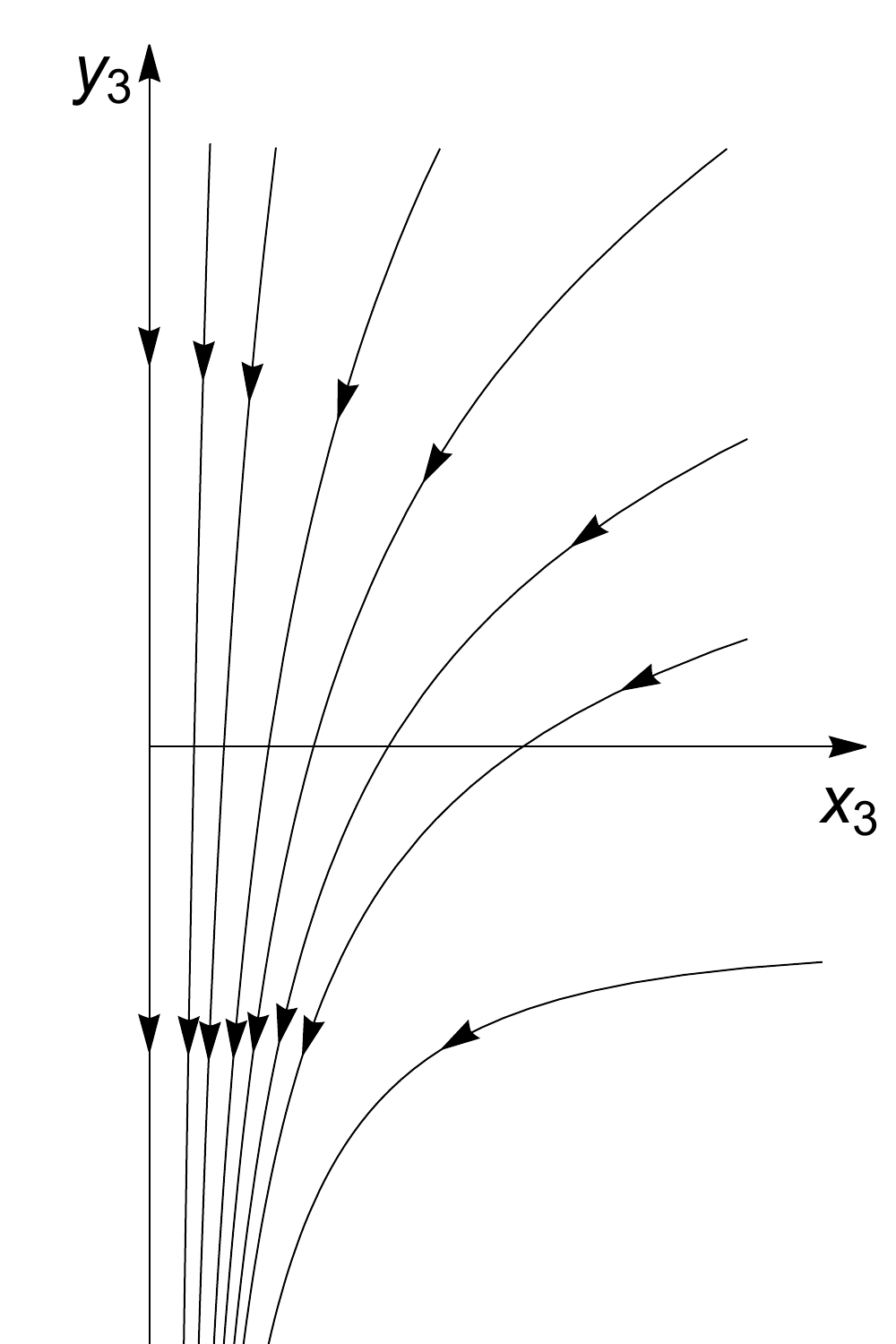} 	
	\caption{The invariant subset $\{u_3=0\}$ of the blow-up in the positive $z$-direction of the fixed point $\mathrm{N}_0$.}
	\label{fig:blowupzdS0}
\end{figure}
To understand the orbit structure on the equator of the upper quarter of the unit 2-sphere, we now perform the blow-ups in the positive $x$ and $y$-directions. 
\paragraph{Blow-up in the positive $x$-direction.}
In the region where $x_3$ become unbounded, we use the blow-up map in the positive $x$-direction consisting of the change of coordinates $(\bar{X},\bar{S},\bar{T})=(u_1,u_1 y_1,u^2_1 z_1)$ which, after canceling a common factor $u_1$, i.e. by changing the time variable $d/d\bar{t} = u_1 d/d\tau_1$, yields the regular dynamical system
\begin{subequations}
	\begin{align}
	\frac{du_1}{d\tau_1} &=-u_1\left\{u_1y_1z_1(2-u_1)+2\left(1-\frac{3}{4}\gamma_\mathrm{pf}\right)(2-u_1)-\left(1-\frac{3}{2}\gamma_\mathrm{pf}\right)u_1y^2_1\right\}, \\
	\frac{dy_1}{d\tau_1} &= y_1+(u_1-2)z_1+2u_1y_1^2z_1, \\
	\frac{dz_1}{d\tau_1} &=z_1\Big\{2+(4-u_1)u_1y_1z_1+2\left(1-\frac{3}{4}\gamma_\mathrm{pf}\right)(2-u_1)-\left(1-\frac{3}{2}\gamma_\mathrm{pf}\right)u_1y^2_1\Big\}.
	\label{eq:blowupxdS0}
	\end{align}
\end{subequations}
The above system has one isolated fixed point 
\begin{equation}
\label{pointP}
\mathrm{P}:\,\,(u_1,y_1,z_1)=(0,0,0).
\end{equation}
The linearised matrix of the system around $\mathrm{P}$ has eigenvalues $3\left(\gamma_\mathrm{pf}-\frac{4}{3}\right)$, $1$, $3(2-\gamma_\mathrm{pf})$ and associated eigenvectors $(1,0,0)$, $(0,1,0)$, $\left(0,2,3\gamma_{\mathrm{pf}}-5\right)$.
For $\gamma_\mathrm{pf}\in(\frac{2}{3},\frac{4}{3})$, the fixed point $\mathrm{P}$ is a hyperbolic saddle and for $\gamma_\mathrm{pf}\in(\frac{4}{3},2)$ a hyperbolic source. The value $\gamma_\mathrm{pf}=4/3$ consists of a bifurcation where the first eigenvalue becomes zero due to the appearance of the line of fixed points $\mathrm{L}_\mathrm{R}$, which in these coordinates is located at $(u_1,y_1,z_1)=(u_*,0,0)$, with $u_*>0$.

On the invariant boundary $\{u_1=0\}$, the fixed point  $\mathrm{P}$ is always a hyperbolic source, and there exists a straight orbit originating from $\mathrm{P}$ defined by
\begin{equation}
\label{eq74}
\frac{d}{d\tau_1}\left(-\frac{3}{2}(\gamma_\mathrm{pf}-\frac{5}{3})y_1+z_1\right)=-\frac{3}{2}(\gamma_\mathrm{pf}-\frac{5}{3})y_1+z_1=0.
\end{equation}
Since both charts $\kappa^{(0)}_{3}$ and $\kappa^{(0)}_{1}$ overlap in the regions $x_3>0$ and $z_1>0$ then, by Lemma~\ref{lemk3}, there are no fixed points or periodic orbits in the region $z_1>0$ of the invariant set $\{u_1=0\}$. The flow on the $\{u_1=0\}$ invariant set is therefore as depicted in Figure~\ref{fig:blowupx}.

\begin{figure}[ht!]
	\centering
		\subfigure[$2/3<\gamma_{\mathrm{pf}}<5/3$.]{\label{fig:blowupx-1}
			\includegraphics[trim={1cm 0.2cm 0.2cm 0.2cm},clip,width=0.32\textwidth]{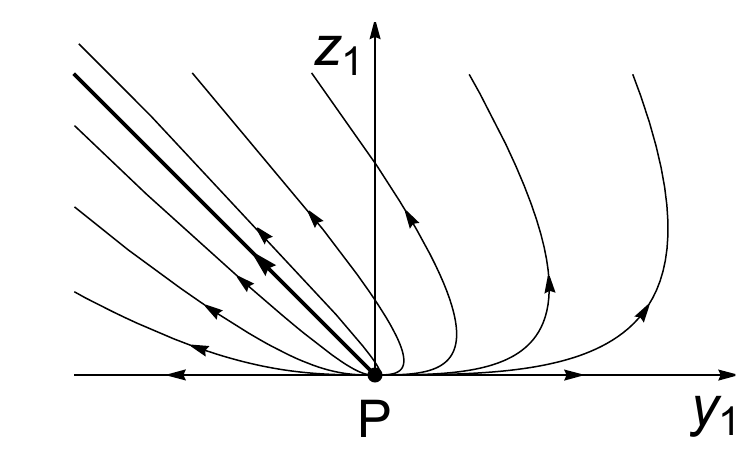}}
		\subfigure[$\gamma_{\mathrm{pf}}=5/3$.]{\label{fig:blowupx-2}
			\includegraphics[trim={1cm 0.2cm 0cm 0.2cm},clip,width=0.32\textwidth]{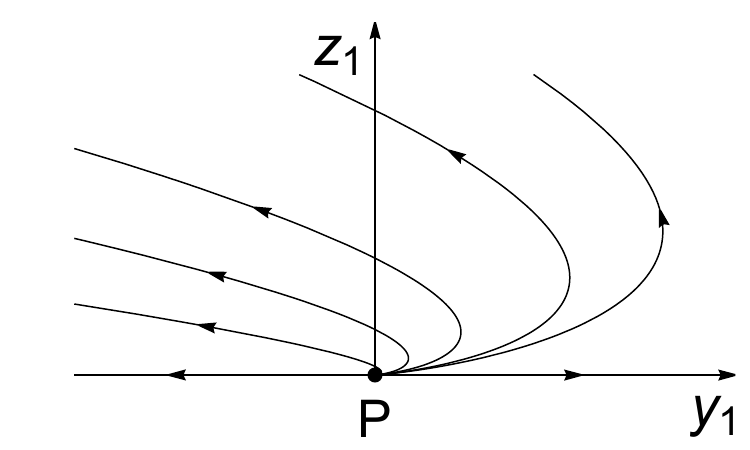}}
		\subfigure[$5/3<\gamma_{\mathrm{pf}}<2$.]{\label{fig:blowupx-3}
			\includegraphics[trim={1cm 0.2cm 0cm 0.2cm},clip,width=0.32\textwidth]{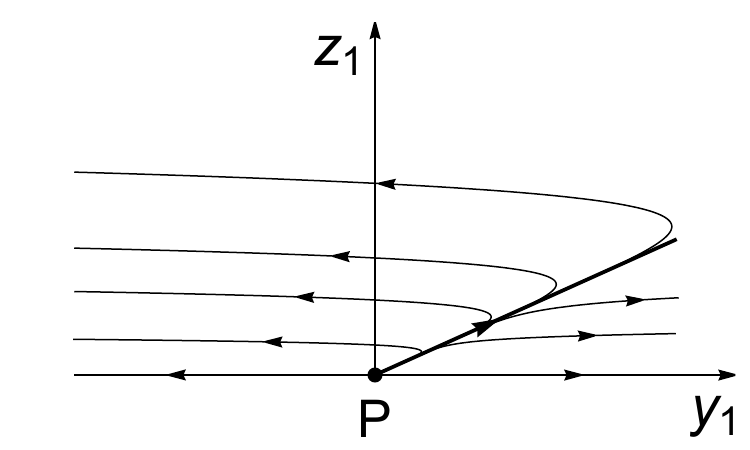}}	
	\caption{The invariant subset $\{u_1=0\}$ of the blow-up in the positive $x$-direction of the fixed point $\mathrm{N}_0$. The bold straight line corresponds to the orbit $-\frac{3}{2}(\gamma_\mathrm{pf}-\frac{5}{3})y_1+z_1=0$.}
	\label{fig:blowupx}
\end{figure}

On the invariant boundary $\{z_1=0\}$, $\mathrm{P}$ is a hyperbolic source when $\gamma_\mathrm{pf}>4/3$ and a saddle when $\gamma_\mathrm{pf}<4/3$. The case $\gamma=4/3$ is a bifurcation where one of the eigenvalues becomes zero due to the appearance of the line of fixed points $\mathrm{L}_\mathrm{R}$. The flow in a neighbourhood of $\mathrm{P}$ is depicted in Figure~\ref{fig:blowupxdS0}.
\begin{figure}[ht!]
	\centering
			\subfigure[$2/3<\gamma_{\mathrm{pf}}<4/3$.]{\label{fig:blowupxN0-1}
			\includegraphics[trim={0cm 0cm 0cm 0cm},clip,width=0.29\textwidth]{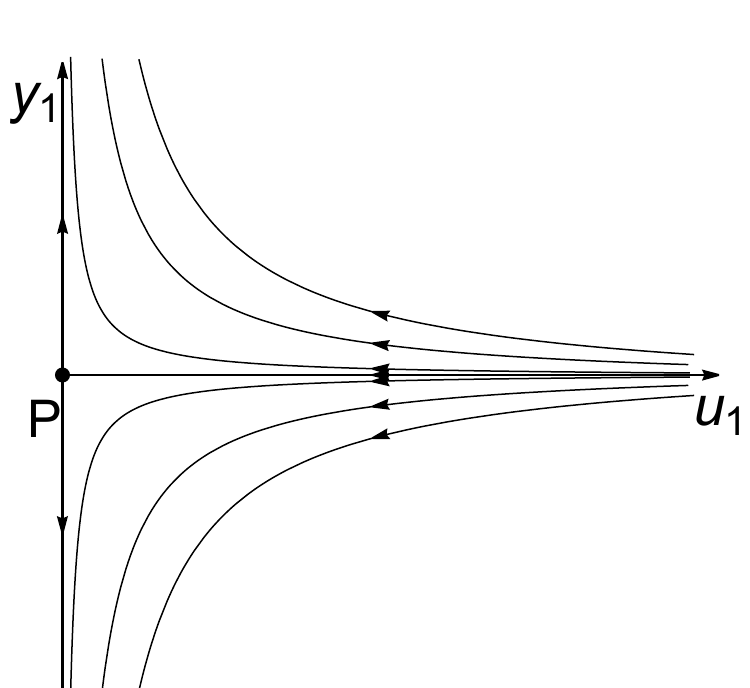}}\hspace{0.5cm}
		\subfigure[$\gamma_{\mathrm{pf}}=4/3$.]{\label{fig:blowupxN0-2}
			\includegraphics[trim={0cm 0cm 0cm 0cm},clip,width=0.29\textwidth]{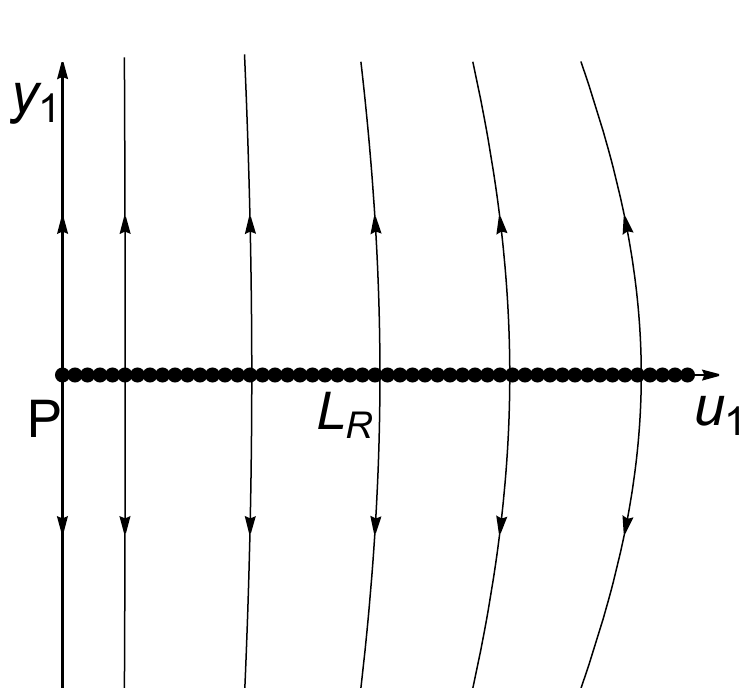}}\hspace{0.5cm}
		\subfigure[$4/3<\gamma_{\mathrm{pf}}<2$.]{\label{fig:blowupxN0-3}
			\includegraphics[trim={0cm 0cm 0cm 0cm},clip,width=0.29\textwidth]{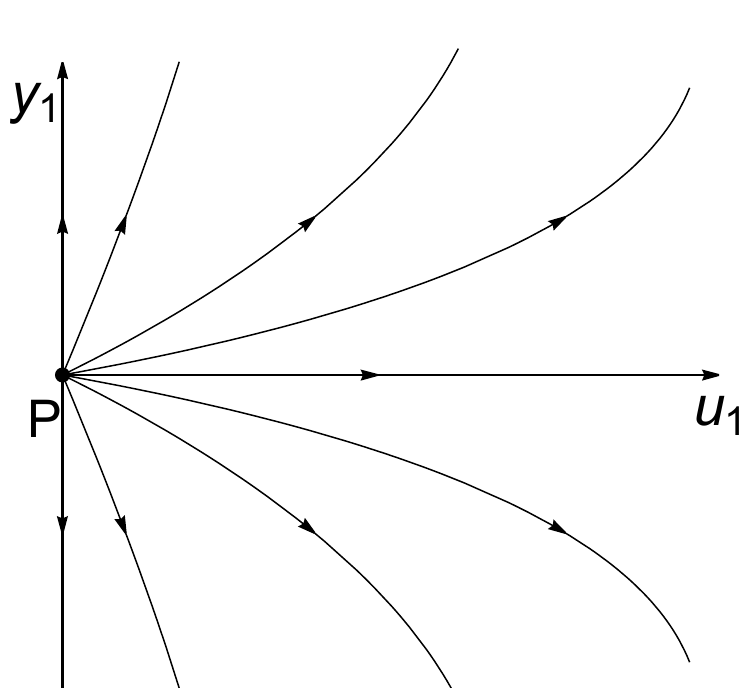}}		
	\caption{Flow in a neighbourhood of the equator of the unit 2-sphere given by the blow-up in the positive $x$-direction of the fixed point $\mathrm{N}_0$ when restricted to the invariant set $\{z_1=0\}$.}
	\label{fig:blowupxdS0}
\end{figure}
%
\paragraph{Blow-up in the $y$-direction.}
Finally we want to understand the flow when $y_3$ becomes unbounded. In this case we make use of the blow-up in the positive and negative $y$-directions (denoted by the $\pm$ subscripts, respectively) and use the change of coordinates
\begin{equation}
(\bar{X},\bar{S},\bar{T})=(u_{2\pm} x_{2\pm},\pm u_{2\pm}, u^2_{2\pm} z_{2\pm}),
\end{equation}
where the physical region of interest has $x_{2\pm}\geq0$ and $z_{2\pm}\geq0$. After changing the time variable to $\tau_{2\pm}$ defined by $d/d\bar{t} = u_{2\pm} d/d\tau_{2\pm}$, yields the regular dynamical system
\begin{subequations}
	\begin{align}
	\frac{dx_{2\pm}}{d\tau_{2\pm}} =&-\left\{x_{2\pm}^2\pm u_{2\pm}z_{2\pm}(2-u_{2\pm}x_{2\pm})\pm (-2+u_{2\pm}x_{2\pm}+u_{2\pm}^2)z_{2\pm}x_{2\pm} \right\},\\
	\frac{du_{2\pm}}{d\tau_{2\pm}} = &-u_{2\pm}\left\{\pm(2-u_{2\pm}x_{2\pm}-u_{2\pm}^2)z_{2\pm}+x_{2\pm}(1-u_{2\pm}x_{2\pm})\right.\nonumber\\
	&\left.+\left(1-\frac{3}{2}\gamma_\mathrm{pf}\right)(x_{2\pm}(2-u_{2\pm}x_{2\pm})-u_{2\pm}) \right\},\\
	\frac{dz_{2\pm}}{d\tau_{2\pm}} =&z_{2\pm}\left\{x_{2\pm}(2-x_{2\pm}u_{2\pm})\pm(4-2u_{2\pm}x_{2\pm}-u_{2\pm}^2)z_{2\pm}\right.\nonumber\\
	&\left.+\left(1-\frac{3}{2}\gamma_\mathrm{pf}\right)(x_{2\pm}(2-u_{2\pm}x_{2\pm})-u_{2\pm})\right\}.
	\end{align}
\end{subequations}
On each invariant subset $\{u_{2\pm}=0\}$  there exists one critical point at the origin
\begin{equation}
\mathrm{Q}_{\pm} \,: (x_{2\pm},u_{2\pm},z_{2\pm})=(0,0,0),
\end{equation}
which is non-hyperbolic. Hence one needs to perform a further blow-up on both $\mathrm{Q}_\pm$  which will be the subject of  the next section.
%
\subsubsection{Blow-up of the fixed points $\mathrm{Q}_\pm$}
\label{BUP_Qpm}
%
We again employ the spherical blow-up method, i.e. we transform the fixed points $\mathrm{Q}_\pm$ at the origin to the unit 2-sphere $\mathbb{S}^2_{\pm}=\left\{(v^{(\pm)},w^{(\pm)},s^{(\pm)}): (v^{(\pm)})^2+(w^{(\pm)})^2+(s^{(\pm)})^2=1\right\}$ 
and define the blow-up space manifold $\mathcal{B}_{\pm}\:=\mathbb{S}^2_{\pm}\times[0,\bar{u}^{(\pm)}]$ for some fixed $\bar{u}^{(\pm)}>0$. We further define the homogeneous blow-up map
\begin{equation}
\Psi_{\pm}\,:\quad \mathcal{B}_\pm\rightarrow \mathbb{R}^3,\qquad \Psi_{\pm}(v^{(\pm)},w^{(\pm)},s^{(\pm)},u^{(\pm)})=(u^{(\pm)}v^{(\pm)},u^{(\pm)}w^{(\pm)},u^{(\pm)} s^{(\pm)}),
\end{equation} 
which after canceling a common factor $u^{(\pm)}$, i.e. by changing the time variable to $\tau^{(\pm)}$ defined by $d/d\tau^{(\pm)}= (u^{(\pm)})^{-1}d/d\tau_{2\pm}$, leads to a desingularisation of the non-hyperbolic fixed points $\mathrm{Q}_\pm$ on the blow-up locus $\{u^{(\pm)}=0\}$.
In order to simplify the analysis of the blow-up we take local charts $\kappa^{(\pm)}_i \,:\,\mathcal{B}_\pm\rightarrow\mathbb{R}^3$ and define the directional blow-up maps $\psi^{(\pm)}_i\,:\, \Psi_\pm \circ (\kappa^{(\pm)}_i)^{-1}$. Since we are interested only in the region where $x_{2\pm}\geq0$ and $z_{2\pm}\geq0$, we just need to consider 3 charts $\kappa^{(\pm)}_i$ such that
\begin{subequations}
	\begin{align}
	\psi^{(\pm)}_{1} &=(u^{(\pm)}_{1}, u^{(\pm)}_{1} w^{(\pm)}_{1},u^{(\pm)}_1 s^{(\pm)}_1), \\
	\psi^{(\pm)}_{2} &=(u^{(\pm)}_{2}v^{(\pm)}_{2}, u^{(\pm)}_{2},u^{(\pm)}_2 s^{(\pm)}_2), \\
	\psi^{(\pm)}_{3} &=(u^{(\pm)}_{3}v^{(\pm)}_{3},u^{(\pm)}_{3} w^{(\pm)}_{3},u^{(\pm)}_{3} ),
	\end{align}
\end{subequations}
where $\psi^{(\pm)}_{1}$, $\psi^{(\pm)}_{2}$, and $\psi^{(\pm)}_{3}$ are the directional blow-ups in the positive directions.

\paragraph{Blow-up in the positive $s$-direction.} We start the analysis by using the chart $\kappa^{(\pm)}_{3}$, i.e. the directional blow-up in the positive $s$-direction. 
In the above coordinates the north pole of the 2-sphere is at the origin while the  equator is located at infinity, which will be analysed using both charts $\kappa^{(\pm)}_{1}$ and $\kappa^{(\pm)}_{2}$. 
In this case, and after the change of time variable $d/d\tau_{2\pm} = u^{(\pm)}_3 d/d\tau^{(\pm)}_3$, we arrive at
\begin{subequations}
	\begin{align}
	\frac{dv^{(\pm)}_3}{d\tau^{(\pm)}_3} &= \pm\left(u^{(\pm)^2}_3v^{(\pm)}_3w^{(\pm)}_3-2\right)\left(v^{(\pm)}_3+w^{(\pm)}_3\right)-v^{(\pm)}_3\Bigg( v^{(\pm)}_3\left(3-u^{(\pm)^2}_3v^{(\pm)}_3w^{(\pm)}_3\right) \nonumber\\
	&\quad+\left(1-\frac{3}{2}\gamma_{\mathrm{pf}}\right)\left(2v^{(\pm)}_3-w^{(\pm)}_3-u^{(\pm)^2}_3v^{(\pm)^2}_3w^{(\pm)}_3\right) \Bigg),\\
	\frac{dw^{(\pm)}_3}{d\tau^{(\pm)}_3} &=-w^{(\pm)}_3\Bigg\{\pm\left(6-u^{(\pm)^2}_3w^{(\pm)}_3\left(3v^{(\pm)}_3+2w^{(\pm)}_3\right)\right)+v^{(\pm)}_3\left(3-2u^{(\pm)^2}_3v^{(\pm)}_3w^{(\pm)}_3\right) \nonumber\\	
	&\quad +2\left(1-\frac{3}{2}\gamma_{\mathrm{pf}}\right)\left(2v^{(\pm)}_3-w^{(\pm)}_3-u^{(\pm)^2}_3v^{(\pm)^2}_3w^{(\pm)}_3\right) \Bigg\},\\
	\frac{du^{(\pm)}_3}{d\tau^{(\pm)}_3} &=u^{(\pm)}_3 \Bigg\{ \pm\left(4-u^{(\pm)^2}_3w^{(\pm)}_3\left(2v^{(\pm)}_3+w^{(\pm)}_3\right)\right)+v^{(\pm)}_3\left(2-u^{(\pm)^2}_3v^{(\pm)}_3w^{(\pm)}_3\right)  \nonumber \\
	&\quad +\left(1-\frac{3}{2}\gamma_{\mathrm{pf}}\right)\left(2v^{(\pm)}_3-w^{(\pm)}_3-u^{(\pm)^2}_3v^{(\pm)^2}_3w^{(\pm)}_3\right) \Bigg\}.
	\end{align}
\end{subequations}
On the $\{u^{(\pm)}_3=0\}$ invariant set, the vacuum invariant boundary $\{\Omega_{\mathrm{pf}}=0\}$ is the straight line $2v^{(\pm)}_3-w^{(\pm)}_3=0$ which is invariant as can be seen by computing
\begin{equation}
\frac{d}{d\tau^{(\pm)}_3}(2v^{(\pm)}_3-w^{(\pm)}_3)=-(2v^{(\pm)}_3-w^{(\pm)}_3)\left(2+3v^{(\pm)}_3+2\left(1-\frac{3}{2}\gamma_\mathrm{pf}\right)(v^{(\pm)}_3-w^{(\pm)}_3)\right).
\end{equation} 
On the blow-up of $\mathrm{Q}_{-}$ there exists a fixed point on the vacuum invariant boundary $\{2v^{(-)}_3-w^{(-)}_3=0\}$ given by
\begin{equation}
\mathrm{dS}_{0}\,\,:\quad (v^{(-)}_3,w^{(-)}_3,u^{(-)}_3)=\left(2, 4,0\right),
\end{equation}
with eigenvalues $-6\gamma_{\mathrm{pf}}$, $-6$, $0$ and eigenvectors $\left(1,\frac{2(6\gamma_{\mathrm{pf}}-7)}{3\gamma_{\mathrm{pf}}-4},0\right)$, $\left(\frac{1}{2},1,0\right)$, $(0,0,1)$, respectively. 
The zero eigenvalue is due to the fact that $\mathrm{dS}_0$ is the extension of the normally hyperbolic line of fixed points $\mathrm{L}_{\mathrm{dS}}$ to the invariant set $\{u^{(-)}_3=0\}$~\footnote{In these coordinates $\mathrm{L}_\mathrm{dS}$ can be parameterised by constant values $u^{(-)}_3=c_0$ and is the solution to the cubic equation $c_0 \left(v^{(-)}_3\right)^3+v^{(-)}_3-2=0$ and $w^{(-)}_3=\left(v^{(-)}_3\right)^2$ for $c_0>0$.}. On the $\{u^{(-)}_3=0\}$ invariant set, the fixed point $\mathrm{dS}_0$ is a hyperbolic sink and, therefore, the $\omega$-limit point of a 1-parameter family of orbits in the physical region $\Omega_\mathrm{pf}>0$, i.e., $w^{(-)}_3<2v^{(-)}_3$.

Furthermore, the systems have $\{w^{(\pm)}_3=0\}$ as invariant boundary subsets, whose intersection with the vacuum invariant boundary  $\{w^{(\pm)}_3=2v^{(\pm)}_3\}$ on $\{u^{(\pm)}_3=0\}$  are simply the fixed points at the origin
\begin{equation}
\mathrm{Z}_{\pm} \,: (v^{(\pm)}_3,w^{(\pm)}_3,u^{(\pm)}_3)=(0,0,0).
\end{equation}
The linearisation around $\mathrm{Z}_\pm$ yields the eigenvalues $\mp 2$, $\mp 6$, $\pm 4$ and eigenvectors $(1,0,0)$, $(1,2,0)$, $(0,0,1)$, respectively.
Hence $\mathrm{Z}_\pm$ are hyperbolic saddles although, on $\{u^{(\pm)}_3=0\}$, $\mathrm{Z}_-$ is a source and $\mathrm{Z}_+$ a sink. Finally, on the invariant subsets $\{w^{(\pm)}_3=0\}$  there are also the fixed points 
\begin{equation}
\mathrm{V}_{\pm} \,: (v^{(\pm)}_3,w^{(\pm)}_3,u^{(\pm)}_3)=\left(\pm \frac{2/3}{\gamma_{\mathrm{pf}}-\frac{5}{3}},0,0\right).
\end{equation}
The fixed point $\mathrm{V}_{-}$ exists for  $\frac{2}{3}<\gamma_\mathrm{pf}<\frac{5}{3}$, while  $\mathrm{V}_{+}$ only exists for $\frac{5}{3}<\gamma_\mathrm{pf}<2$. The linearisation around these fixed points yields the eigenvalues $\pm 2$, $\pm \frac{2(\frac{8}{3}-\gamma_{\mathrm{pf}})}{\gamma_{\mathrm{pf}}-\frac{5}{3}}$, $\mp \frac{2(2-\gamma_{\mathrm{pf}})}{\gamma_{\mathrm{pf}}-\frac{5}{3}}$ with associated eigenvectors $(1,0,0)$, $\left(\frac{3(3\gamma_{\mathrm{pf}}-4)}{2(6\gamma_{\mathrm{pf}}-13)},1,0\right)$, $(0,0,1)$.
Hence both $\mathrm{V}_{\pm}$ are hyperbolic saddles. However, on the invariant sets $\{u^{(\pm)}_3=0\}$, $\mathrm{V}_{-}$ is a saddle and the $\alpha$-limit point of a unique orbit in the region $w^{(-)}_3>0$, while  $\mathrm{V}_{+}$ is a source and the $\alpha$-limit point of a 1-parameter set of orbits in the region $w^{(+)}_3>0$. Figure~\ref{fig:vBupQ} shows the physical regions $0<w^{(\pm)}_3<2v^{(\pm)}_3$ and the trivial orbit structure of its invariant boundaries. In Lemma~\ref{LemmaCompQ_pm} we will show that the interior consists of heteroclinic orbits as seen in Figure~\ref{fig:vBupQ}.
\begin{figure}[ht!]
	\centering
		\subfigure[ $2/3<\gamma_\mathrm{pf} < 5/3$.]{\label{fig:s-direction-1}
\includegraphics[trim={0.7cm 1cm 0cm 0cm},clip,width=0.23\textwidth]{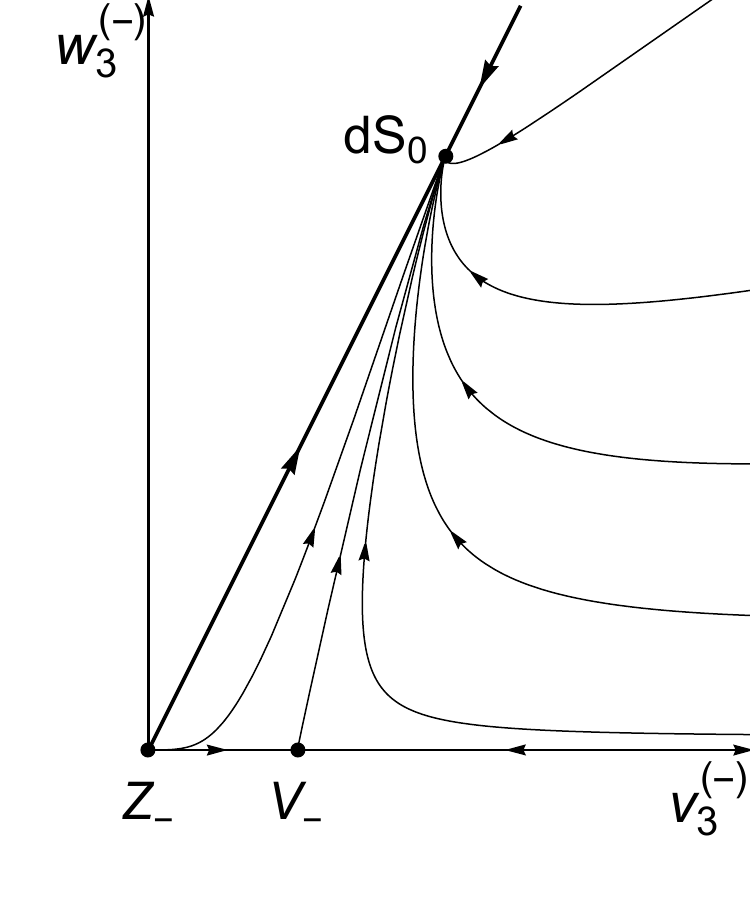}}\hspace{0.1cm}
\subfigure[ $5/3\leq\gamma_\mathrm{pf}< 2$.]{\label{fig:s-direction-2}
\includegraphics[trim={1.5cm 2cm 0cm 0cm},clip,width=0.23\textwidth]{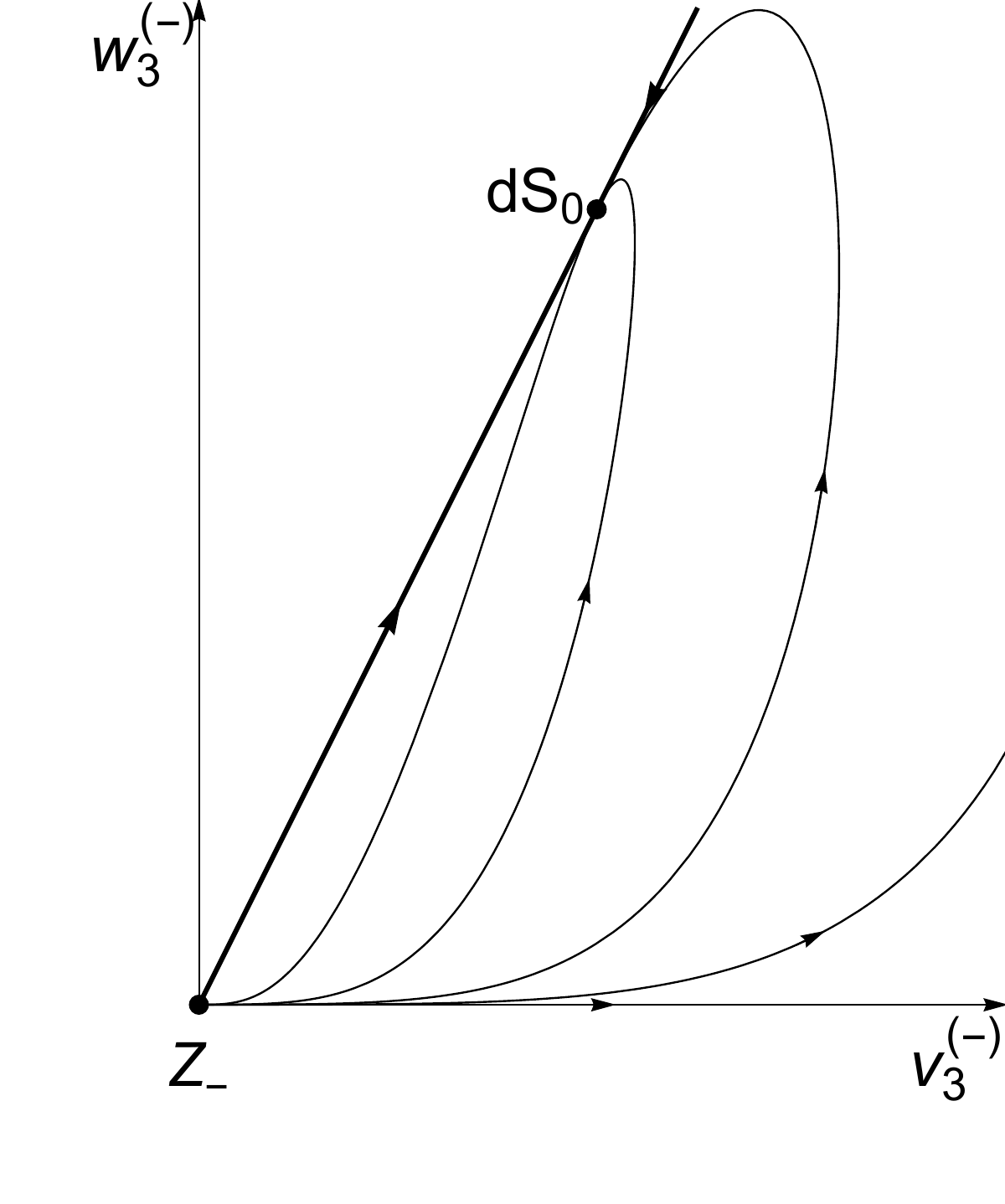}}\hspace{0.1cm}
\subfigure[ $2/3<\gamma_\mathrm{pf} \leq 5/3$.]{\label{fig:s-direction+1}
\includegraphics[trim={0.7cm 1cm 0cm 0cm},clip,width=0.23\textwidth]{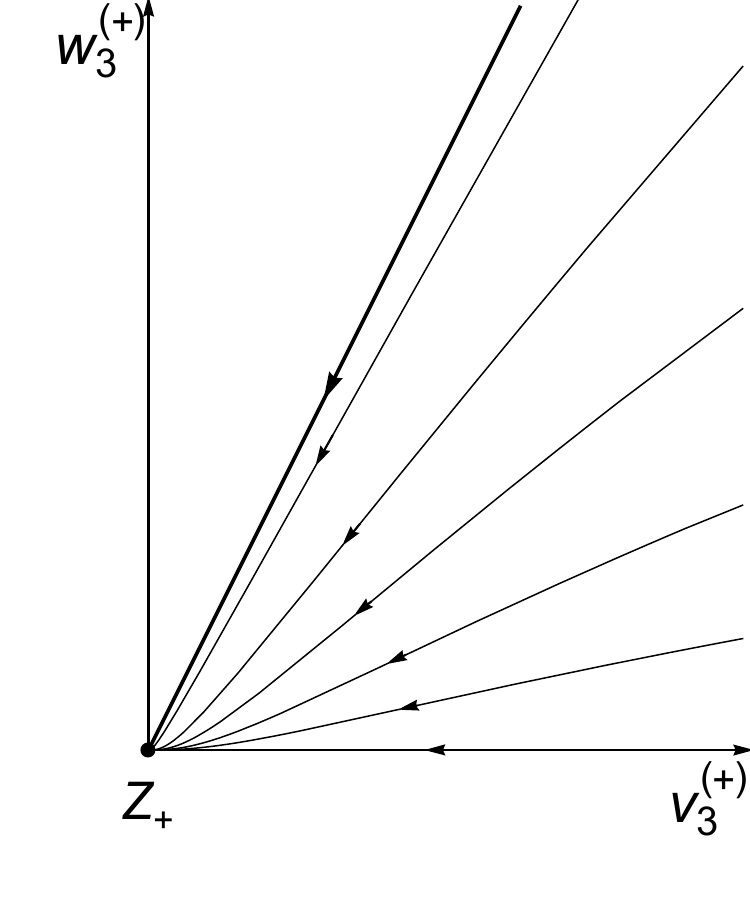}}\hspace{0.1cm}
\subfigure[ $5/3<\gamma_\mathrm{pf}<2$.]{\label{fig:s-direction+2}
\includegraphics[trim={1.5cm 2.1cm  0.1cm 0cm},clip,width=0.23\textwidth]{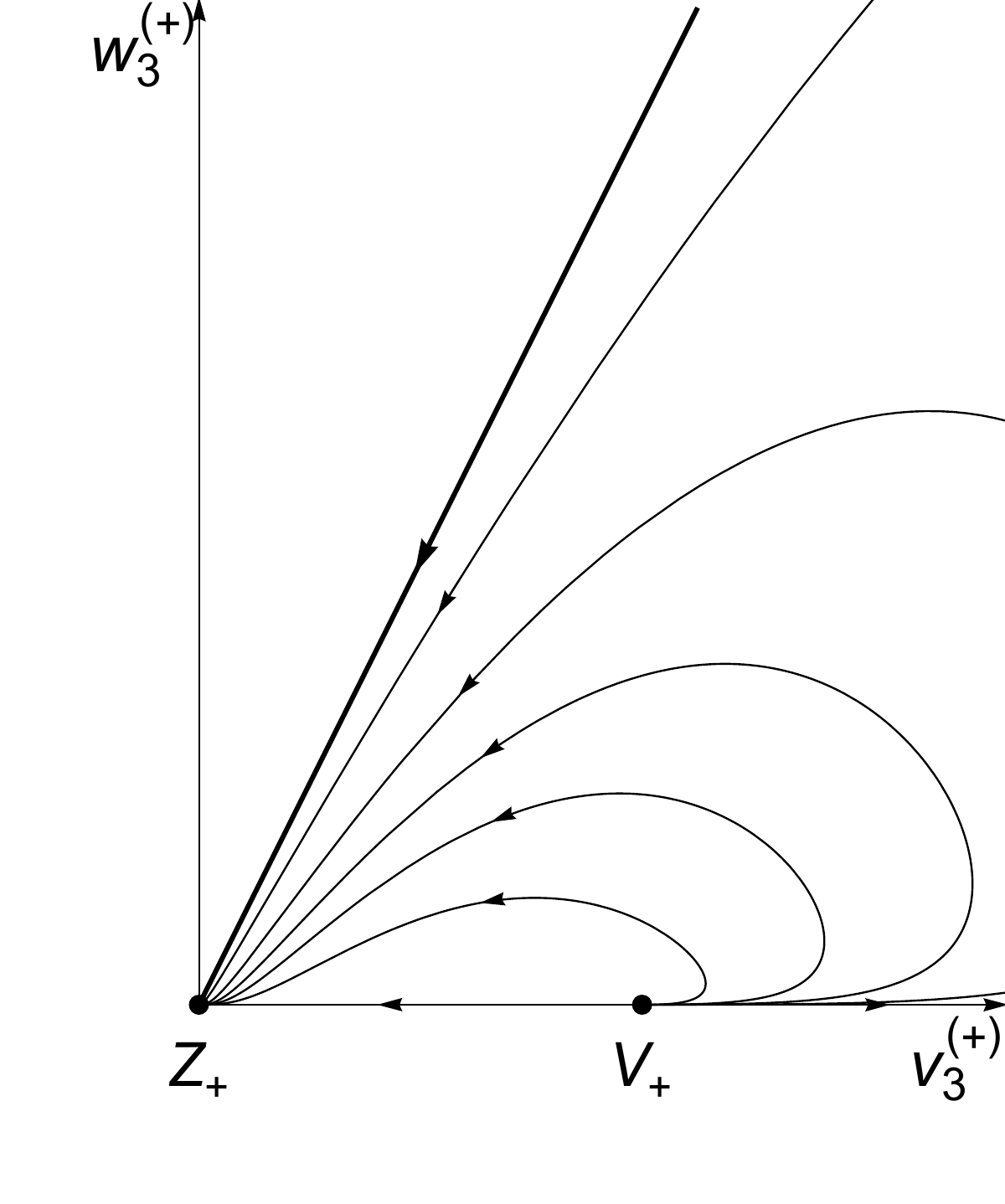}}	
	\caption{Physical regions on the $\{u_3^{(\pm)}=0\}$ invariant sets of the blow-ups in the positive $v$ of the fixed points $\mathrm{Q}_\pm$}
	\label{fig:vBupQ}
\end{figure}

Next, we proceed with the study of the flow in a neighbourhood of the sphere's equator, which is the region where $v^{(\pm)}_3$ and $w^{(\pm)}_3$ become unbounded. Since the physical region of the blow-up space has $w^{(\pm)}_3\leq2v^{(\pm)}_3$, in order to cover the boundary of the physical region on the equator we just need to consider the blow-up in the positive $v$-direction.

\paragraph{Blow-up in the positive $v$-direction.} In this case, and after canceling a common factor $u^{(\pm)}_1$, i.e. by changing the time variable $d/d\tau_{2\pm} = u^{(\pm)}_1 d/d\tau^{(\pm)}_1$, we arrive at the regular system
\begin{subequations}
	\begin{align}
	\frac{du^{(\pm)}_1}{d\tau^{(\pm)}_1} &= -u^{(\pm)}_1 \left\{ 1 \pm (-2+2w^{(\pm)}_1+u^{(\pm)^2}_1w^{(\pm)}_1)s^{(\pm)}_1 \right\}, \\
	\frac{dw^{(\pm)}_1}{d\tau^{(\pm)}_1} &=w^{(\pm)}_1 \Bigg\{ u^{(\pm)^2}_1w^{(\pm)}_1 \pm (-4+2w^{(\pm)}_1+2u^{(\pm)^2}_1w^{(\pm)}_1+u^{(\pm)^2}_1w^{(\pm)^2}_1)s^{(\pm)}_1 \nonumber \\
	&\quad -\left(1-\frac{3}{2}\gamma_\mathrm{pf}\right)(2-w^{(\pm)}_1-u^{(\pm)^2}_1w^{(\pm)}_1)\Bigg\}, \\
	\frac{ds^{(\pm)}_1}{d\tau^{(\pm)}_1} &=s^{(\pm)}_1\Bigg\{ (3-u^{(\pm)^2}_1 w^{(\pm)}_1) \pm (2+2w^{(\pm)}_1-u^{(\pm)^2}_1 w^{(\pm)}_1-u^{(\pm)^2}_1 w^{(\pm)^2}_1)s^{(\pm)}_1  \nonumber \\
	&\quad +\left(1-\frac{3}{2}\gamma_\mathrm{pf}\right)(2-w^{(\pm)}_1-u^{(\pm)^2}_1w^{(\pm)}_1)\Bigg\}.
	\end{align}
\end{subequations}
On the invariant set $\{u^{(\pm)}_1=0\}$, the equator is the invariant boundary subset $\{s^{(\pm)}_1=0\}$ of the physical state-space with $s^{(\pm)}_1>0$ and $0<w^{(\pm)}_1<2$, where $\{w^{(\pm)}_1=0\}$ and $\{w^{(\pm)}_1=2\}$ are the invariant boundaries which correspond in chart $\kappa^{(\pm)}_3$ to the invariant boundaries $\{w^{(\pm)}_3=0\}$ and $\{w^{(\pm)}_3=2w^{(\pm)}_3\}$, respectively. On these invariant boundaries we find the fixed points $\mathrm{V}_{\pm}$ and $\mathrm{dS}_0$ which in the present chart $\kappa^{(\pm)}_{1}$, are given by
\begin{subequations}
\begin{align}
\mathrm{V}_{\pm}:&\qquad (u^{(\pm)}_1,w^{(\pm)}_1,s^{(\pm)}_1)=\left(0,0,\pm\frac{3}{2}\left(\gamma_{\mathrm{pf}}-\frac{5}{3}\right)\right)\\
\mathrm{dS}_{0}:&\qquad (u^{(-)}_1,w^{(-)}_1,s^{(-)}_1)=\left(0,2,\frac{1}{2}\right).
\end{align}
\end{subequations}
The intersection of the invariant boundaries $\{w^{(\pm)}_1=0\}$ and $\{w^{(\pm)}_1=2\}$ with the equator $\{s^{(\pm)}_1=0\}$ are the two fixed points:
\begin{equation}
\mathrm{T}_{\pm} :\qquad  (u^{(\pm)}_1,w^{(\pm)}_1,s^{(\pm)}_1)=(0,0,0)
\end{equation}
with eigenvalues $-1$, $3(\gamma_{\mathrm{pf}}-\frac{2}{3})$, $-3(\gamma_{\mathrm{pf}}-\frac{5}{3})$ and associated eigenvectors 
$(1,0,0)$, $(0,1,0)$, $(0,0,1)$, as well as
\begin{equation}
\mathrm{S}_{\pm} :\qquad  (u^{(\pm)}_1,w^{(\pm)}_1,s^{(\pm)}_1)=(0,2,0)
\end{equation}
with eigenvalues $-1$, $-3(\gamma_{\mathrm{pf}}-\frac{2}{3})$, $3$ and eigenvectors $(1,0,0)$, $(0,1,0)$, $(0,0,1)$, respectively.

On the $\{u^{(-)}_1=0\}$ invariant set, $\mathrm{T}_{-}$ is a source (saddle) for $\gamma_{\mathrm{pf}}<5/3$ ($\gamma_{\mathrm{pf}}>5/3$).  Conversely, on the $\{u^{(+)}_1=0\}$ invariant set, $\mathrm{T}_+$ is a hyperbolic source (saddle) when $\gamma_{\mathrm{pf}}<5/3$ ($\gamma_{\mathrm{pf}}>5/3$). When $\gamma_{\mathrm{pf}}=5/3$, the fixed points $\mathrm{V}_{\pm}$ merge with $\mathrm{T}_\pm$ resulting in the appearance of the 1-dimensional centre manifold given by the $s^{(\pm)}_{1}$-axis. As a consequence, in this case, $\mathrm{T}_+$ has a 1-dimensional unstable manifold (the $w^{(+)}_{1}$-axis) and a 1-dimensional unstable centre manifold, while $\mathrm{T}_-$ has a 1-dimensional unstable manifold (the $w^{(-)}_{1}$-axis) and a 1-dimensional stable centre manifold. The fixed points $\mathrm{S}_{\pm}$ are hyperbolic saddles irrespective of the parameter $\gamma_{\mathrm{pf}}$. In turn the boundary of the physical region on the equator consists of the heteroclinic orbits $\mathrm{T}_{\pm}\rightarrow \mathrm{S}_{\pm}$. 
Figure~\ref{fig:s-w-Bup} shows the physical regions $\{0<w^{(\pm)}_1<2\,\wedge\, s^{(\pm)}_1>0\}$ and the trivial orbit structure of its invariant boundaries.

\begin{figure}[ht!]
	\centering
		\subfigure[ $\frac{2}{3}<\gamma_\mathrm{pf}<\frac{5}{3}$.]{\label{fig:v-direction-1}
			\includegraphics[trim={3cm 3cm 1.5cm 0cm},clip,width=0.23\textwidth]{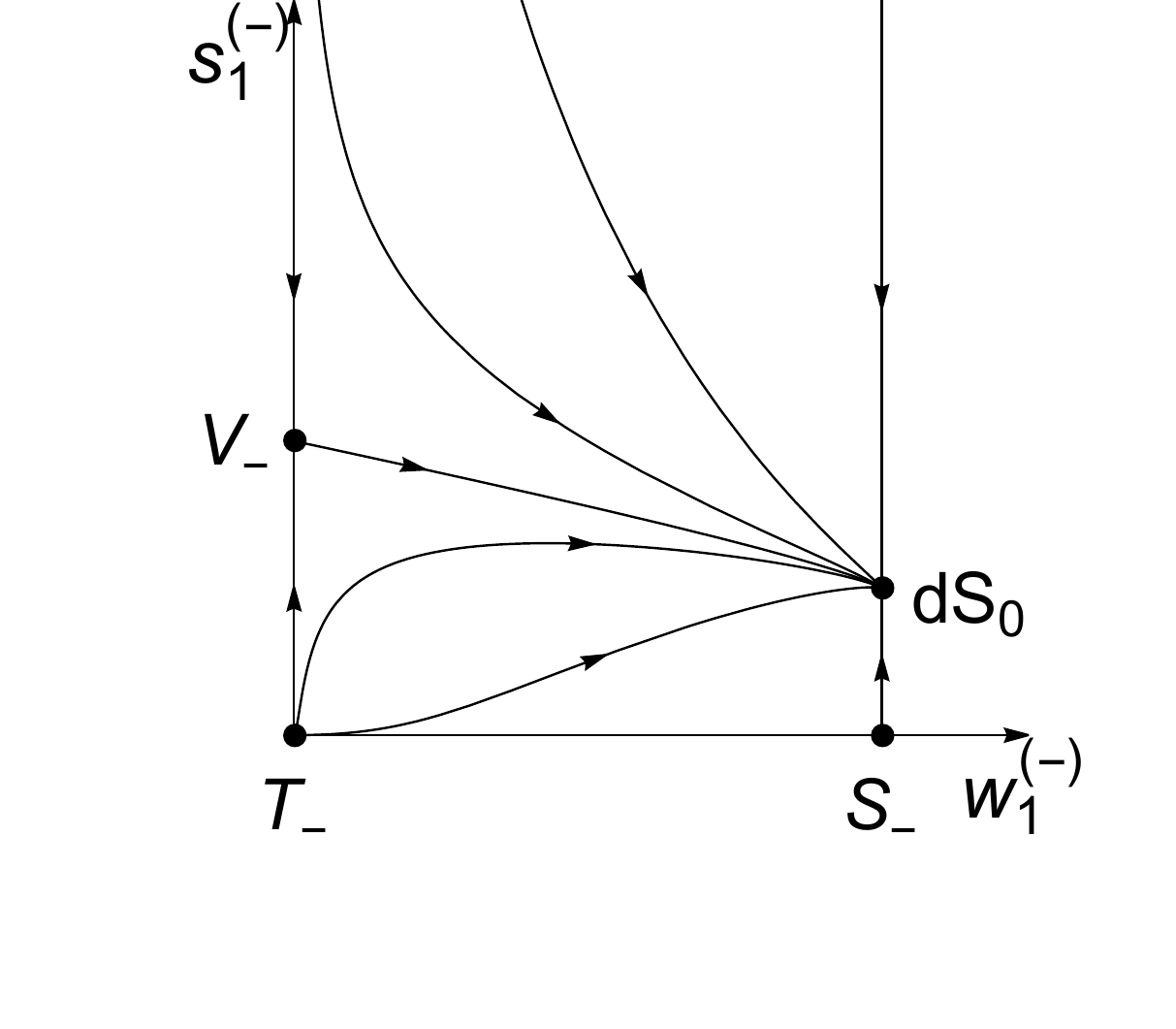}}
		\subfigure[ $\frac{5}{3}\leq\gamma_\mathrm{pf}<2$.]{\label{fig:v-direction-2}
			\includegraphics[trim={3cm 3cm 1.5cm 0cm},clip,width=0.23\textwidth]{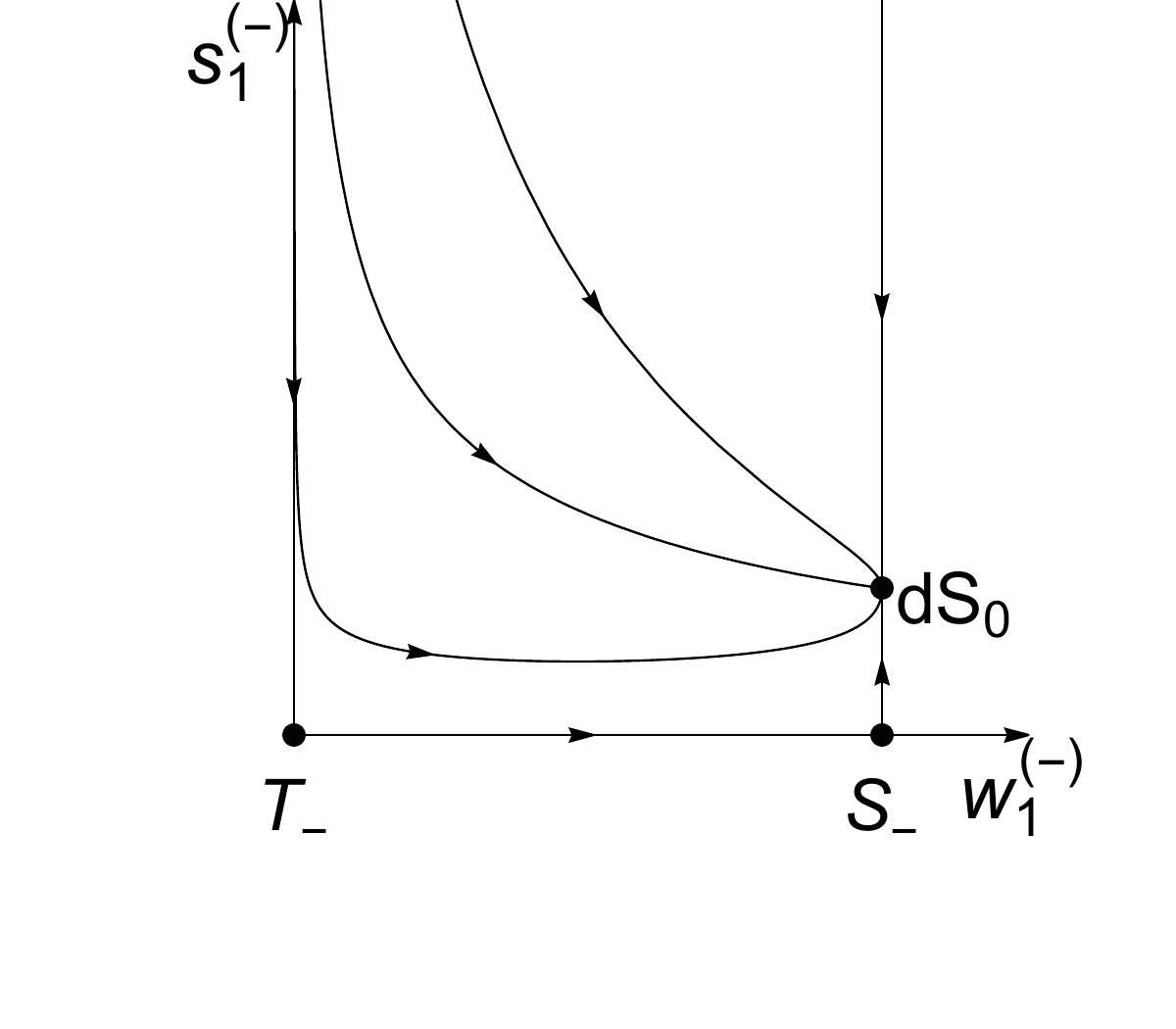}}
		\subfigure[ $\frac{2}{3}<\gamma_\mathrm{pf}\leq\frac{5}{3}$.]{\label{fig:v-direction+1}
			\includegraphics[trim={3cm 3cm 1.5cm 0cm},clip,width=0.23\textwidth]{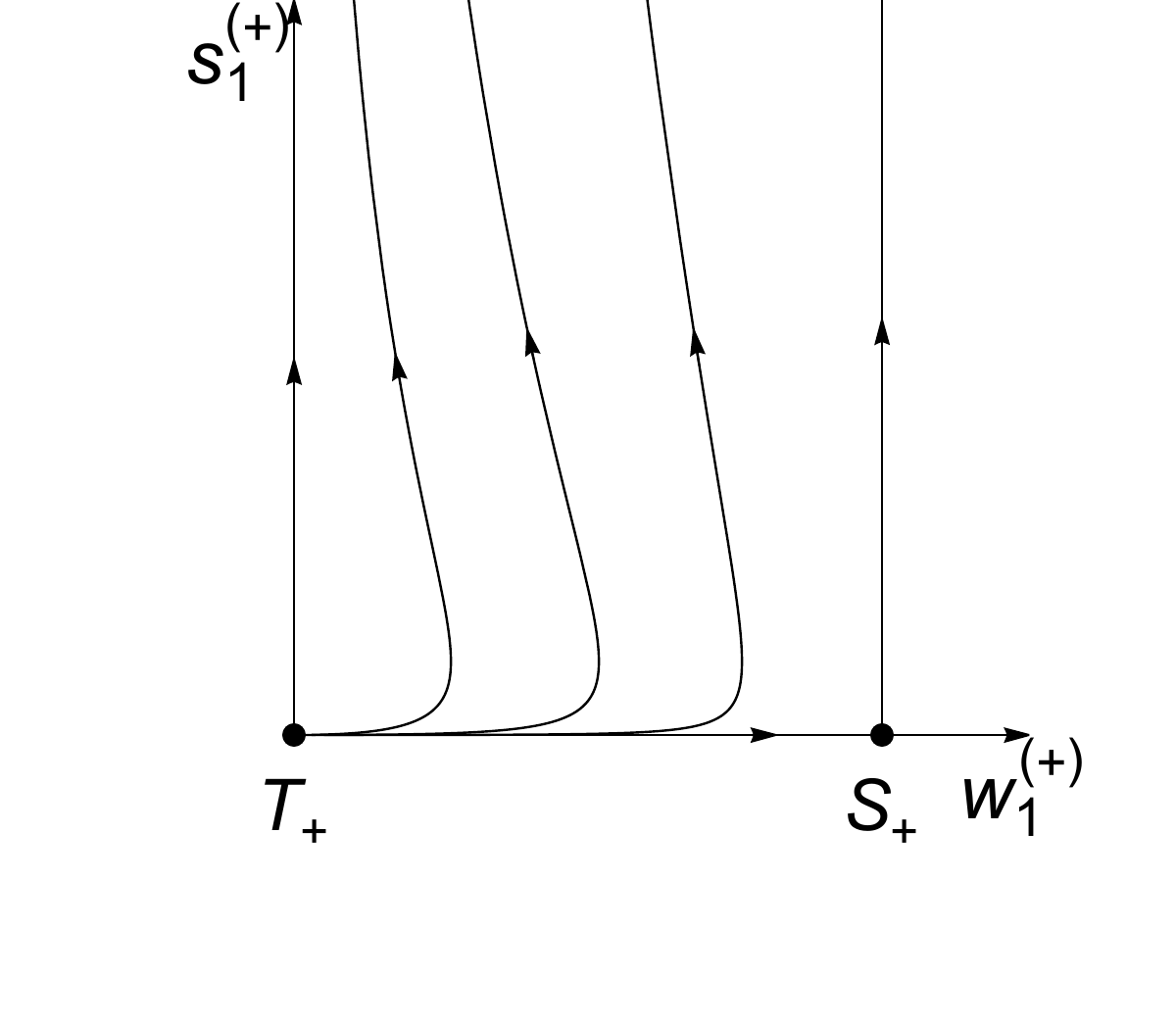}}
		\subfigure[ $\frac{5}{3}<\gamma_\mathrm{pf}<2$.]{\label{fig:v-direction+2}
			\includegraphics[trim={3cm 3cm 1.5cm 0cm},clip,width=0.23\textwidth]{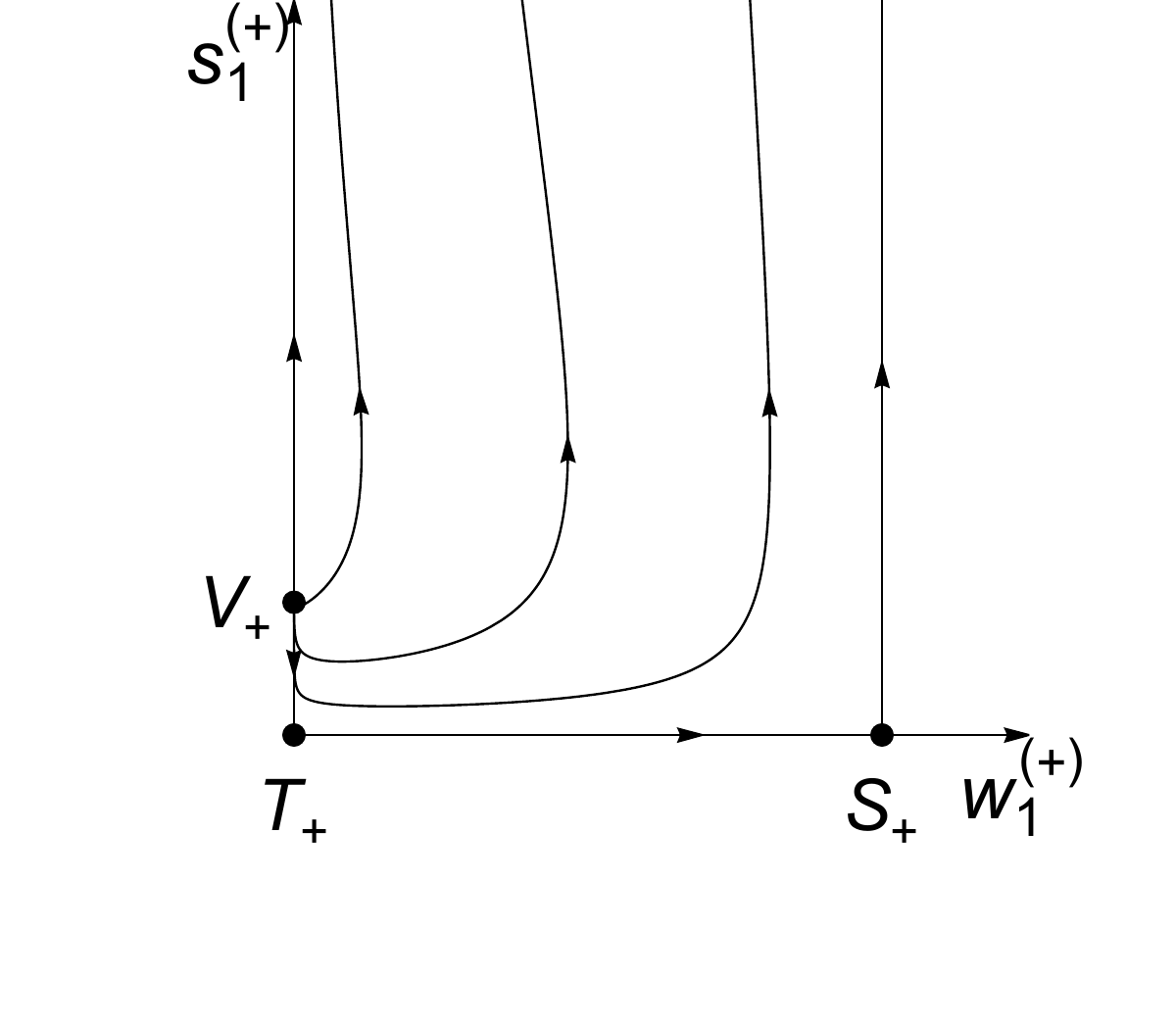}}	
	\caption{Physical region of the invariant set $\{u_1^{(\pm)}=0\}$ of the blow-ups in the positive $s-$directions of the fixed points $\mathrm{Q}_\pm$.}
	\label{fig:s-w-Bup}
\end{figure}
Hence the physical region of the blow-ups of each fixed point  $\mathrm{Q}_\pm$ is a compact region of the unit 2-sphere, whose boundary consists of three invariant subsets intersecting at the three fixed points $\mathrm{S}_\pm$, $\mathrm{T}_\pm$ and $\mathrm{Z}_\pm$. 
Having uncovered the boundary structure, we are now in position to describe the interior dynamics:
\begin{lemma}\label{LemmaCompQ_pm}
    The flow in the interior of the physical region of the blow-up locus of $\mathrm{Q}_{\pm}$ consists entirely of heteroclinic orbits connecting the fixed points on the boundary. 
    More precisely:~
    \begin{itemize}
    	\item[(i)] If $\gamma_{\mathrm{pf}}\in(\frac{2}{3},\frac{5}{3})$ then on the blow-up of $\mathrm{Q}_{-}$ there exists a separatrix $\mathrm{V}_-\rightarrow \mathrm{dS}_0$ which splits the physical region into two invariant regions consisting of heteroclinic orbits $\mathrm{Z}_{-}\rightarrow \mathrm{dS}_0$ and $\mathrm{T}_{-}\rightarrow \mathrm{dS}_0$, while the flow on the blow-up of $\mathrm{Q}_{+}$ consists of heteroclinic orbits $\mathrm{T}_{+}\rightarrow \mathrm{Z}_+$.
    	\item[(ii)] If $\gamma_{\mathrm{pf}}=\frac{5}{3}$ then the flow on the blow-up of $\mathrm{Q}_{-}$ consists of heteroclinic orbits $\mathrm{Z}_{-}\rightarrow \mathrm{dS}_0$ and the flow on the blow-up of $\mathrm{Q}_{+}$ consists of heteroclinic orbits $\mathrm{T}_{+}\rightarrow \mathrm{Z}_+$.
    	\item[(iii)] If $\gamma_{\mathrm{pf}}\in(\frac{5}{3},2)$ the flow on the blow-up of $\mathrm{Q}_{-}$ consists of heteroclinic orbits $\mathrm{Z}_{-}\rightarrow \mathrm{dS}_0$ and the flow on the blow-up of $\mathrm{Q}_{+}$ consists of heteroclinic orbits $\mathrm{V}_{+}\rightarrow \mathrm{Z}_+$.
    \end{itemize}  
\end{lemma}
\begin{proof}
 Since on the compact physical region of the blow-ups of $\mathrm{Q}_{\pm}$ there are no interior fixed points and no closed heteroclinic cycles can exists on the boundary, then the proof follows from the application of the Poincar\'e-Bendixson theorem and the local stability analysis of the semi-hyperbolic fixed points on the boundary. 
\end{proof}	
After the characterisation of the flow on the blow-up locus of the non-hyperbolic fixed points $\mathrm{Q}_\pm$ of Lemma~\ref{LemmaCompQ_pm}, what remains to be done in order to describe the flow in neighbourhood of $\mathrm{N}_0$, is to determine all possible $\alpha$ and $\omega$-limit sets for the orbits of Lemma~\ref{lemk3}.

We have seen, using chart $\kappa^{(0)}_1$, that the flow on the invariant set $\{z_1=0\}\cap\{u_1=0\}$, which describes an open set of the equator of $\mathbb{S}^2_{0}$ on the blow-up of $\mathrm{N}_0$ where $x_3$ becomes unbounded, has a single hyperbolic source $\mathrm{P}$ as shown in Figure~\ref{fig:blowupx}. 
When $y_3$ becomes unbounded, the flow is described by the blow-ups of $\mathrm{Q}_\pm$. Consider the charts $\kappa^{(\pm)}_{3}$ and the associated flow when restricted to the invariant boundaries $\{w^{(\pm)}_3=0\}$. Using a composition of transition charts, it is straightforward to check that the orbit originating from $\mathrm{P}$ defined by~\eqref{eq74} is simply given by
\begin{equation}\label{w3PtoV}
v^{(\pm)}_3=\mp \frac{2/3}{\gamma_\mathrm{pf}-\frac{5}{3}},
\end{equation}
where
\begin{equation}
\frac{dv^{(\pm)}_3}{d\tau^{(\pm)}_3}=0, \qquad \frac{du^{(\pm)}_3}{d\tau^{(\pm)}_3} =\pm 4u^{(\pm)}_3 \left(1+\frac{1}{2}\left(\frac{\gamma_\mathrm{pf}-\frac{4}{3}}{\gamma_\mathrm{pf}-\frac{5}{3}}\right)\right).
\end{equation}
This shows that there exists a heteroclinic orbit $\mathrm{P}\rightarrow \mathrm{V}_{-}$ if $\gamma_\mathrm{pf}<\frac{5}{3}$ or $\mathrm{P}\rightarrow \mathrm{V}_{+}$ if $\gamma_\mathrm{pf}>\frac{5}{3}$. When $\gamma_\mathrm{pf}=\frac{5}{3}$ these heteroclinic orbits merge with the equator which therefore consists of the two heteroclinic orbits $\mathrm{P}\rightarrow\mathrm{T}_{\pm}$. It is also straightforward to show that in the charts $\kappa^{(\pm)}_{3}$, the invariant subset $\{w^{(\pm)}_3=0\}\cap\{v^{(\pm)}_3=0\}$, which is the $u^{(\pm)}_3$-axis, corresponds to the boundary $\{u_3=0\}\cap\{x_3=0\}$ of the half-plane of Lemma~\ref{lemk3}, which therefore consists of the heteroclinic orbit $\mathrm{Z}_+\rightarrow\mathrm{Z}_-$. Figure~\ref{fig:s-direction2} shows the invariant sets $\{w^{(\pm)}_3=0\}$. Hence the orbits described in Lemma~\ref{lemk3} are contained in a compact region of $\mathbb{S}^2_{\pm}$ whose boundary consists of four invariant subsets which intersect at the fixed points $\mathrm{Z}_\pm$, $\mathrm{T}_\pm$. 
\begin{figure}[ht!]
	\centering
		\subfigure[$2/3<\gamma_\mathrm{pf} < 5/3$.]{\label{fig:w3-1}
			\includegraphics[trim={1.2cm 1cm 0cm 0cm},clip,width=0.23\textwidth]{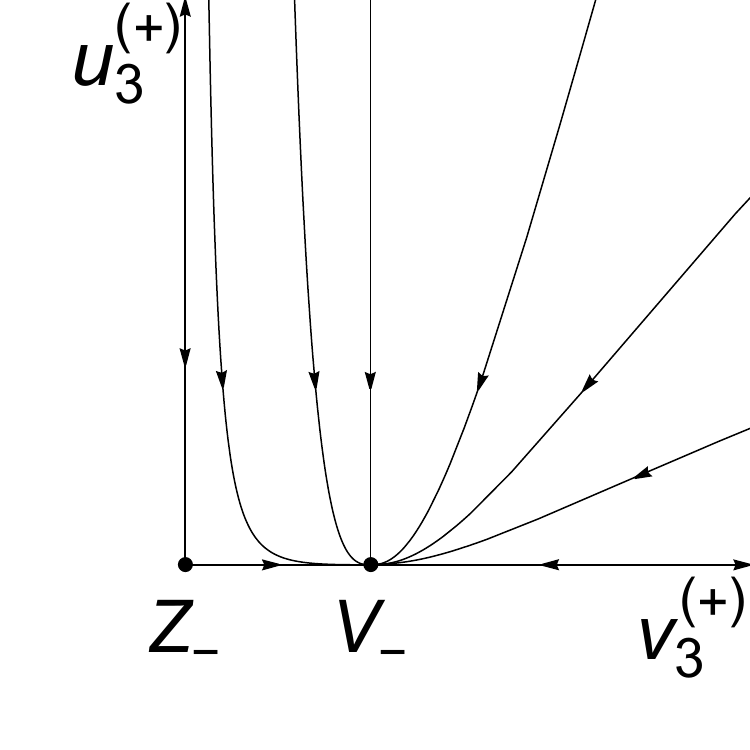}}
		\subfigure[$5/3\leq\gamma_\mathrm{pf}< 2$.]{\label{fig:w3-2}
			\includegraphics[trim={1.2cm 1.1cm 0cm 0cm},clip,width=0.23\textwidth]{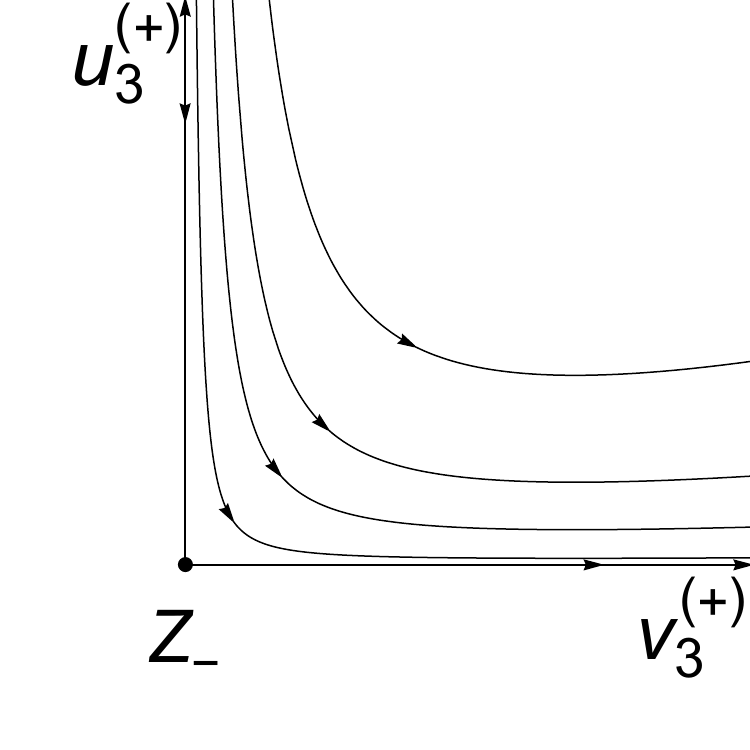}}
		\subfigure[$2/3<\gamma_\mathrm{pf}\leq5/3$.]{\label{fig:w3+1}
			\includegraphics[trim={0.5cm 0.7cm 0.1cm 0cm},clip,width=0.23\textwidth]{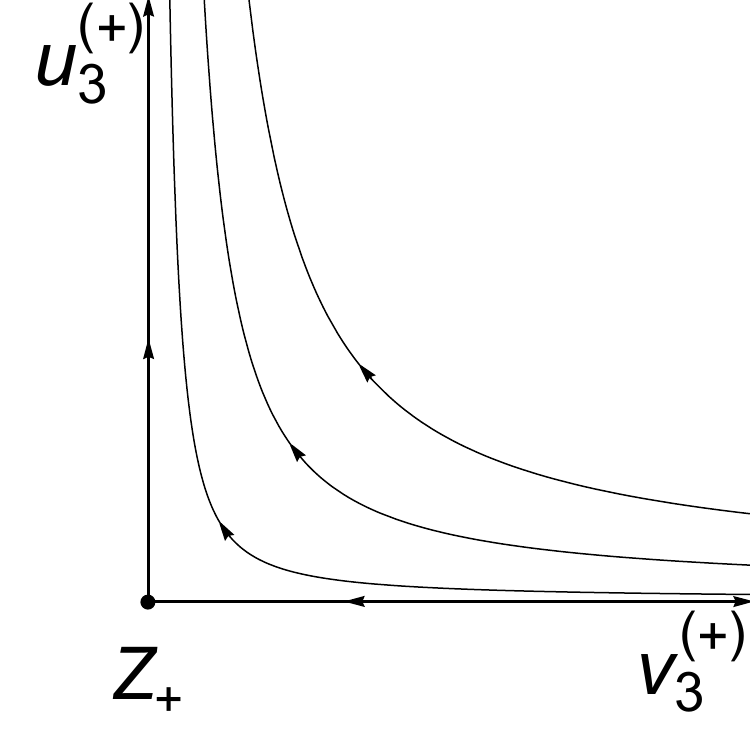}}
		\subfigure[$5/3<\gamma_\mathrm{pf}<2$.]{\label{fig:w3+2}
			\includegraphics[trim={0.5cm 0.5cm  0cm 0cm},clip,width=0.23\textwidth]{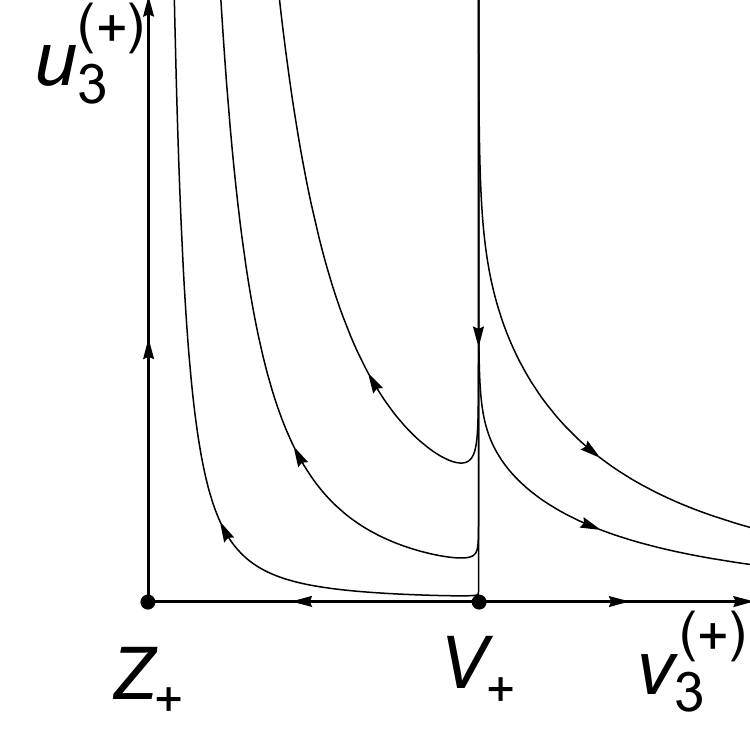}}	
	\caption{Invariant subsets $\{w_3^{(\pm)}=0\}$ of the blow-up in the positive $s-$direction of the fixed points $\mathrm{Q}_\pm$.}
	\label{fig:s-direction2}
\end{figure}
\begin{lemma}\label{LemComp0}
The flow on the interior of the compact region on the unit 2-sphere of the blow-up of $\mathrm{N}_0$ consists entirely of heteroclinic orbits connecting the fixed points at the boundary. More precisely:
\begin{itemize}
	\item[(i)] If $\gamma_{\mathrm{pf}}\in(\frac{2}{3},\frac{5}{3})$ then the flow consists of heteroclinic orbits $\mathrm{P}\rightarrow\mathrm{V}_-$;
	\item[(ii)] If $\gamma_{\mathrm{pf}}=\frac{5}{3}$ then the flow consists of heteroclinic orbits $\mathrm{P}\rightarrow \mathrm{T}_-$;
	\item[(iii)] If $\gamma_{\mathrm{pf}}\in(\frac{5}{3},2)$ there exists a separatrix $\mathrm{P}\rightarrow\mathrm{V}_+$ which divides the state-space into two invariant regions: one region consisting of heteroclinic orbits $\mathrm{P}\rightarrow\mathrm{T}_+$ and the other region consisting of heteroclinic orbits $\mathrm{P}\rightarrow\mathrm{T}_-$.
\end{itemize}  

\end{lemma}

\begin{proof}
The proof follows from Lemma~\ref{lemk3} which states that all possible attracting sets are located on the 1-dimensional invariant boundaries, together with the local stability of the semi-hyperbolic fixed points on these boundaries.
\end{proof}
Gathering all the above information, particularly lemmas~\ref{LemmaCompQ_pm} and~\ref{LemComp0}, leads to a complete picture of the \emph{sucessive blow-ups} leading to the desingularisation of $\mathrm{N}_0$. 
\begin{proposition}\label{BupN0}
	The flow in a neighbourhood  of the non-hyperbolic fixed point $\mathrm{N}_0$ is as depicted in Figure~\ref{fig:BUP3D_N0}, being the $\alpha$-limit point of a 2-parameter family of interior orbits when $\frac{4}{3}<\gamma_{\mathrm{pf}}<2$ (originating from $\mathrm{P}$), while no interior orbits converge to $\mathrm{N}_0$ if $\frac{2}{3}<\gamma_{\mathrm{pf}}\leq\frac{4}{3}$.
\end{proposition}
\begin{remark}\label{BupN0T0}
	The intersection with the $\{T=0\}$ invariant boundary is just the union of the spheres' equators of the successive blow-ups which therefore consists of the union of the heteroclinic orbits $\{\mathrm{P}\rightarrow\mathrm{T}_+\}\cup\{\mathrm{T}_+\rightarrow\mathrm{S}_+\}\cup\{\mathrm{P}\rightarrow\mathrm{T}_-\}\cup\{\mathrm{T}_-\rightarrow\mathrm{S}_-\}$.
\end{remark}
\begin{remark}\label{BupN0Om0}
	The intersection of the vacuum invariant boundary $\{\Omega_{\mathrm{pf}}=0\}$ with the successive blow-ups is simply the one-dimensional invariant set consisting of the union of the heteroclinic orbits $\{\mathrm{S}_+\rightarrow\mathrm{Z}_+\}\cup\{\mathrm{Z}_+\rightarrow\mathrm{Z}_-\}\cup\{\mathrm{Z}_-\rightarrow\mathrm{dS}_0\}\cup\{\mathrm{S}_-\rightarrow\mathrm{dS}_0\}$.
\end{remark}
\begin{figure}[!htb]
	\centering \vspace{-0.25cm}
		\subfigure[$2/3<\gamma_\mathrm{pf}<4/3$.]{
			\includegraphics[trim={2.5cm 34.7cm 3.5cm 2.7cm},clip,width=0.94\textwidth]{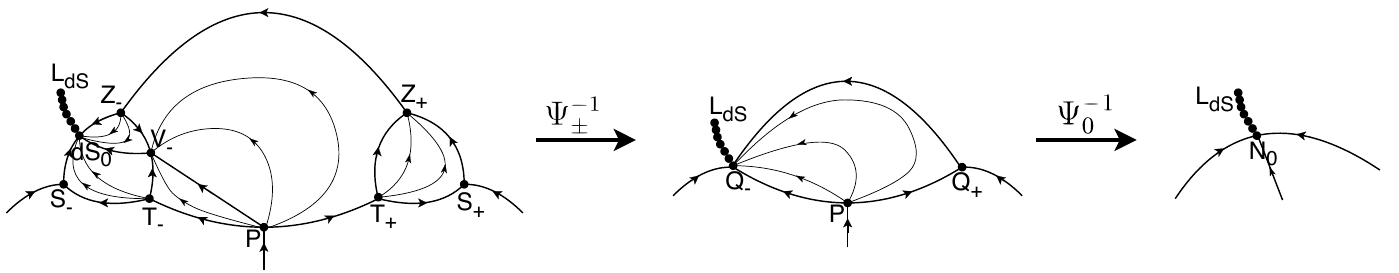}}\vspace{-0.2cm}
		\subfigure[$\gamma_\mathrm{pf}=4/3$.]{
			\includegraphics[trim={2.5cm 34.7cm 3.5cm 2.7cm},clip,width=0.94\textwidth]{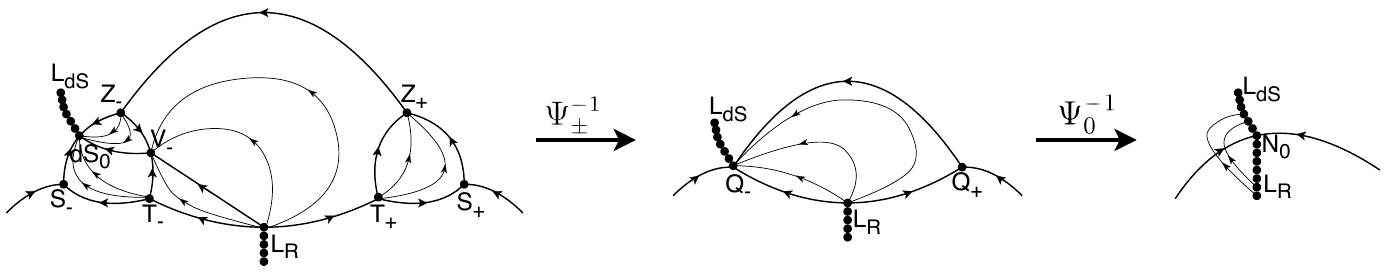}}\vspace{-0.2cm}
		\subfigure[$4/3<\gamma_\mathrm{pf}<5/3$.]{
			\includegraphics[trim={2.5cm 34.7cm 3.5cm 2.7cm},clip,width=0.94\textwidth]{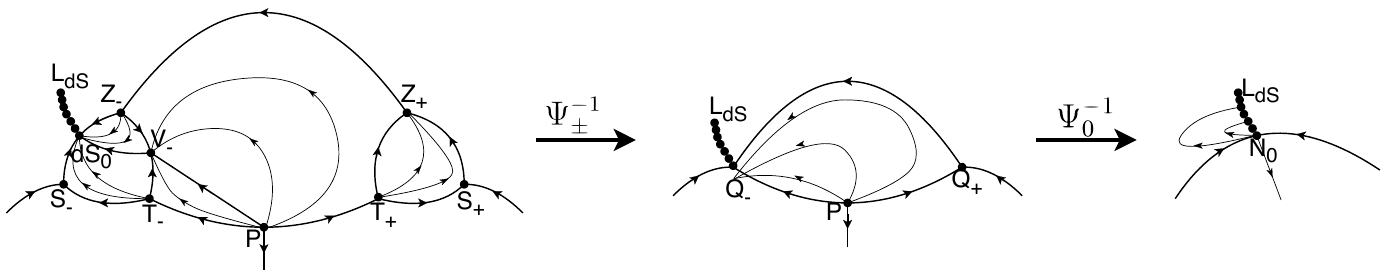}}\vspace{-0.2cm}
		\subfigure[$\gamma_\mathrm{pf}=5/3$.]{
			\includegraphics[trim={2.5cm 34.7cm 3.5cm 2.7cm},clip,width=0.94\textwidth]{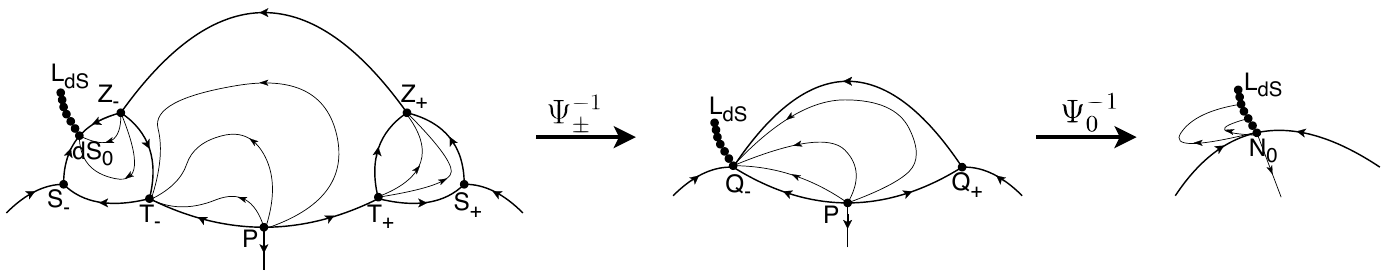}}\vspace{-0.2cm}
		\subfigure[$5/3<\gamma_\mathrm{pf}<2$.]{
			\includegraphics[trim={2.5cm 34.7cm 3.5cm 2.7cm},clip,width=0.94\textwidth]{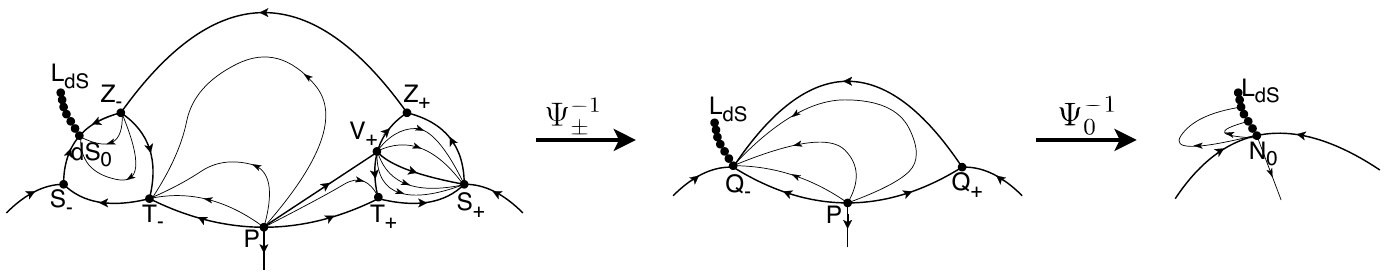}}
	\caption{Sucessive blow-ups on the desingularization  $\mathrm{N}_0$.}\label{fig:BUP3D_N0}
\end{figure}
%
\subsection{Blow-up of the fixed point $\mathrm{N}_1$}
\label{sec:BUP_N1}
In this section we perform the blow-up of the fixed point $\mathrm{N}_1$. Similarly to the blow-up of $\mathrm{N}_0$ we start by introducing the local variables
\begin{equation}
(\bar{X},\bar{S},\bar{T})=\left(1+X,S,\frac{1-T}{T}\right),
\end{equation}
that puts $\mathrm{N}_1$ at the origin. Then, after changing the time variable to $\bar{N}$ defined by
\begin{equation}
\frac{d}{d\bar{N}}=\frac{1}{T}\frac{d}{d\bar{t}}=\sqrt{12\alpha}\frac{d}{dt},
\end{equation}
the system of equations~\eqref{dynsysTXS} results in 
\begin{subequations}
	\begin{align}
	\frac{d\bar{X}}{d\bar{N}} &=\bar{X}S+\bar{T}(2-\bar{X})\left[\bar{X}(2-\bar{X})+\left(1-\frac{3}{2}\gamma_\mathrm{pf}\right)(\bar{X}(2-\bar{X})-S^2)\right], \\
	\frac{dS}{d\bar{N}} &=-(\bar{X}-S^2) -S\bar{T}\left[-(1-\bar{X})(2-\bar{X})+\left(1-\frac{3}{2}\gamma_\mathrm{pf}\right)\left(\tilde{X}(2-\tilde{X}) - S^2\right)\right], \\
	\frac{d\bar{T}}{d\bar{N}} &= - \bar{T}\left\{S+\bar{T}\left[(2-\bar{X})^2-\left(1-\frac{3}{2}\gamma_\mathrm{pf}\right)\left(\bar{X}(2-\bar{X}) - S^2\right)\right]\right\}.
	\end{align}
\end{subequations}
We proceed as in the previous section and employ the spherical blow-up method by transforming the fixed point $\mathrm{N}_1$ at the origin to the unit 2-sphere $\mathbb{S}^2_1=\{(x,y,z): x^2+y^2+z^2=1\}$. Then, defining
the blow-up space manifold $\mathcal{B}_1\:=\mathbb{S}^2_1\times[0,\bar{u}^{(1)}]$, for some fixed $\bar{u}^{(1)}>0$, together with the quasi-homogeneous blow-up map
\begin{equation}
\Psi_1\,:\quad \mathcal{B}_1\rightarrow \mathbb{R}^3,\qquad \Psi_1(x,y,z,u)=(u^2x,uy,u z) 
\end{equation} 
leads to a desingularisation of the non-hyperbolic fixed point on the blow-up locus $\{u=0\}$, after cancelling a common factor $u$ (i.e. by changing the time variable to $\tau$ defined by $d/d\tau= u^{-1}d/d\bar{N}$). Since $\Psi_1$ is a diffeomorphism outside of the sphere $\mathbb{S}^2_1\times\{u=0\}$, which corresponds to the fixed point $(0,0,0)$, the dynamics on the blow-up space $\mathcal{B}_1\setminus \{\mathbb{S}^2_1\times \{u = 0\}\}$ is topological conjugate to $\mathbb{R}^3\setminus \{0,0,0\}$. Instead of using standard spherical coordinates on $\mathcal{B}$, we will use different local charts $\kappa^{(1)}_i \,:\,\mathcal{B}\rightarrow\mathbb{R}^3$ and define the directional blow-up maps $\psi^{(1)}_i\,:\, \Psi_1\circ (\kappa^{(1)}_i)^{-1}$ for which the resulting state vector are simpler to analyse. 

Since we are interested only in the region where $\bar{X}\geq0$ and $\bar{T}\geq0$, we just need to consider four charts $\kappa^{(1)}_i$ such that
\begin{subequations}
	\begin{align}
	\psi^{(1)}_{1} &=(u^2_{1}, u_{1} y_{1},u_1 z_1), \\
	\psi^{(1)}_{2\pm} &=(u^2_{2\pm}x_{2\pm},\pm u_{2\pm},u_{2\pm} z_{2\pm}), \\
	\psi^{(1)}_{3} &=(u^2_{3}x_{3},u_{3}y_3,u_3),
	\end{align}
\end{subequations}
where $\psi^{(1)}_{1}$, $\psi^{(1)}_{2+}$, and $\psi^{(1)}_3$ are called the directional blow-ups in the positive $x$, $y$ and $z$ directions, respectively, and $\psi^{(1)}_{2-}$ is the blow-up in the negative $y$-direction. In these coordinates the equator of the sphere is located at infinity,
which will be analysed using both charts $\kappa^{(1)}_1$ and $\kappa^{(1)}_2$. Using the transition charts $\kappa^{(1)}_{ij}=\kappa^{(1)}_{j}\circ(\kappa^{(1)}_{i})^{-1}$ we can then identify special invariant subsets and obtain a global picture of the blow-up solution space.
\paragraph{Blow-up in the positive $z$-direction.} For the blow-up in the positive $z$-direction, and after cancelling a common factor $u_3$ (i.e. by changing the time variable $d/d\bar{N} = u_3 d/d\tau_3$), we obtain the regular dynamical system
\begin{subequations}
	\begin{align}
	\frac{dx_3}{d\tau_3} &=3x_3\left(y_3+(2-u^2_3 x_3)^2\right)+\left(1-\frac{3}{2}\gamma_\mathrm{pf}\right)\left(x_3 (2-u^2_3x_3)-y^2_3\right)(2-3u^2_3 x_3), \\
	\frac{dy_3}{d\tau_3} &=-x_3+2y^2_3+y_3\left[(2-u^2_3 x_3)(3-2u^2_3x_3)-2 \left(1-\frac{3}{2}\gamma_\mathrm{pf}\right)\left(x_3 (2-u^2_3x_3)-y^2_3\right)u^2_3\right], \\
	\frac{du_3}{d\tau_3} &=-u_3\left[y_3+(2-u^2_3 x_3)^2- \left(1-\frac{3}{2}\gamma_\mathrm{pf}\right)\left(x_3 (2-u^2_3x_3)-y^2_3\right)u^2_3\right].
	\end{align}
\end{subequations}
Since $\bar{X}\geq0$, we will only consider the flow in the half-plane  $x_3\geq0$. Moreover on the $\{u_3=0\}$ invariant plane, the vacuum invariant boundary $\Omega_\mathrm{pf}=0$ is given in  by the parabola $y^2_3-2x_3=0$ which is invariant as follows from 
\begin{equation}
\frac{d}{d\tau_3}\left(y^2_3-2x_3\right)= \left[12+4y_3+4\left(1-\frac{3}{2}\gamma_\mathrm{pf}\right)\right](y^2_3-2x_3).
\end{equation}
Hence the physical region $\Omega_\mathrm{pf}>0$ on the plane consists of the unbounded region $\{x_3>0 \wedge y^2_3-2x_3>0\}$. 
The above system has two fixed points on the invariant parabola. One fixed point at the vertice (the north pole of $\mathbb{S}^2$)
\begin{equation}
\mathrm{R}_1\,:\quad (x_3,y_3,u_3)=(0,0,0)
\end{equation}
with eigenvalues $\lambda_1=-2(3\gamma_{\mathrm{pf}}-8)$, $\lambda_2=6$, $\lambda_3=-4$ and eigenvectors $(2(3\gamma_{\mathrm{pf}}-5),1,0)$, $(0,1,0)$, $(0,0,1)$, respectively. The fixed point $\mathrm{R}_1$ is a hyperbolic saddle, with unstable manifold the $\{u_3=0\}$ invariant plane. Therefore $\mathrm{R}_1$ is the $\alpha$-limit point of a 1-parameter set of orbits in the interior of the parabola. The other fixed point is
	\begin{equation}
	\mathrm{dS}_1\,:\quad (x_3,y_3,u_3)=(8,-4,0)
	\end{equation}
	 with eigenvalues $\lambda_1=-6$, $\lambda_2=-6\gamma_{\mathrm{pf}}$ and $\lambda_3=0$ and eigenvectors $(-4,1,0)$, $(2(3\gamma_{\mathrm{pf}}-5),1,0)$ and $(0,0,1)$.
	On $\{u_3=0\}$, $\mathrm{dS}_1$ has eigenvalues with negative real part being a hyperbolic sink. The zero eigenvalue is related to the fact that $\mathrm{dS}_1$ is the extension of the line of fixed points $\mathrm{L}_\mathrm{dS}$ to the invariant set $\{u_3=0\}$.  Figure~\ref{fig:blowupzN1} shows the invariant parabola and its orbit structure. 
	\begin{figure}[ht!]
	\centering
		\includegraphics[trim={0cm 1.1cm 0cm 0cm},clip,width=0.4\textwidth]{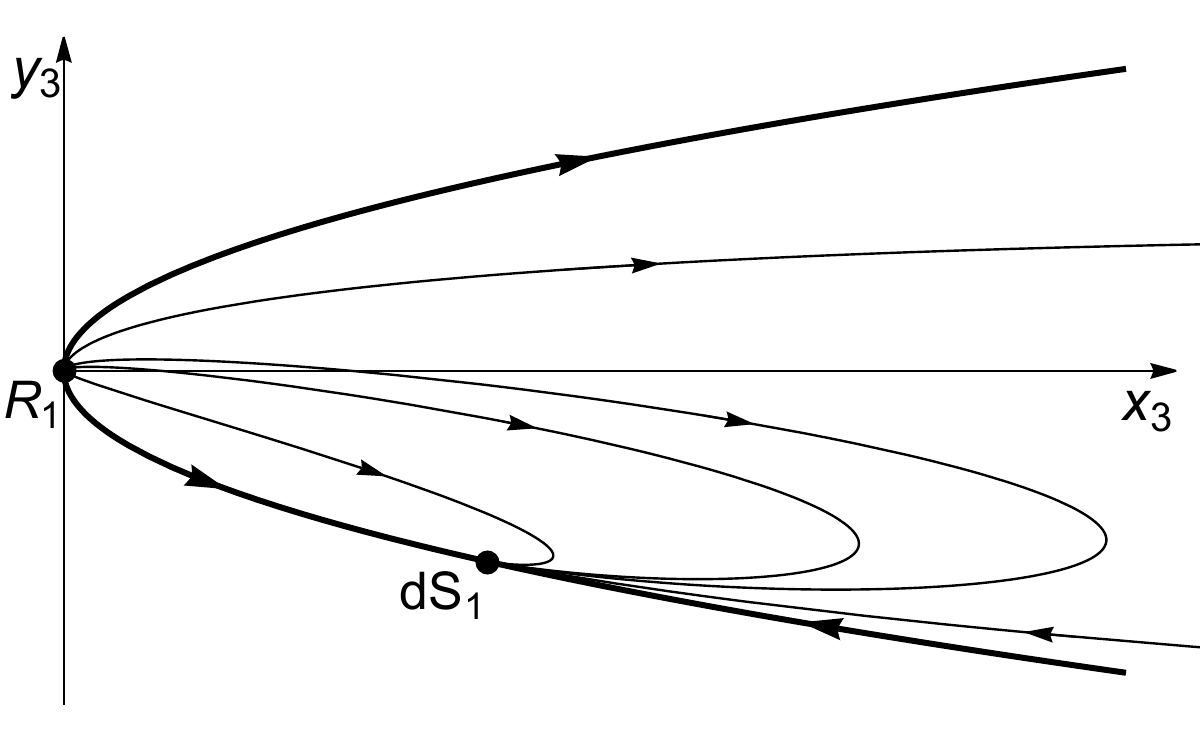}
		\caption{The invariant plane $\{u_3=0\}$ of the blow-up in the $z$-direction of the fixed point $\mathrm{N}_1$, showing the physical region defined by $x_3>0$ and $y^2_3<2x_3$.}\label{fig:blowupzN1}
	\end{figure}
\paragraph{Blow-up in the positive $x$-direction.} 
In order to understand the flow on the physical region and its extension to the equator of the 2-sphere, we turn to the blow-up in the positive $x$-direction. Using chart $\kappa^{(1)}_1$ and 
after canceling a common factor $u_1$, i.e. by changing the time variable $d/d\bar{N} = u_1 d/d\tau_1$, we obtain the regular dynamical system
\begin{subequations}
	\begin{align}
		\frac{du_1}{d\tau_1} &= \frac{u_1}{2}\left\{ y_1+z_1(2-u^2_1)\left[ 2-u^2_1+\left(1-\frac{3}{2}\gamma_\mathrm{pf}\right)\left(2-u^2_1-y^2_1\right) \right] \right\},\\
		\frac{dy_1}{d\tau_1} &= \frac{1}{2}\left\{ -2+y^2_1-y_1z_1\left[ u^2_1(2-u^2_1)+\left(1-\frac{3}{2}\gamma_\mathrm{pf}\right)\left(2-u^2_1-y^2_1\right)(2+u^2_1) \right] \right\},\\
		\frac{dz_1}{d\tau_1} &=-\frac{1}{2}z_1\left\{ 3y_1+z_1\left[ 3(2-u^2_1)^2+\left(1-\frac{3}{2}\gamma_\mathrm{pf}\right)\left(2-u^2_1-y^2_1\right)(2-3u^2_1) \right] \right\}.
	\end{align}
\end{subequations}
On the $\{ u_1=0 \}$ invariant plane, the vacuum invariant subset $\Omega_{\mathrm{pf}}=0$ consists of the two disjoint invariant sets $\{y_1=\pm\sqrt{2}\}$. The equator of $\mathbb{S}^2_1$ corresponds to the invariant boundary $\{z_1=0\}$. The physical region of the plane is therefore given by $\{-\sqrt{2}< y_1<\sqrt{2}\wedge z_1> 0\}$ and contains three fixed points located at the invariant boundaries: The fixed point $\mathrm{dS}_1$ at $(u_1,y_1,z_1)=\left( 0,-\sqrt{2},\frac{1}{2\sqrt{2}} \right)$, and on $\{u_1=0\}$, is a hyperbolic sink, and the two equivalent fixed points at the intersection of the vacuum boundary with the equator:
\begin{equation}
\mathrm{W}_\pm \,:\quad  (u_1,y_1,z_1)=\left(0,\pm\sqrt{2},0 \right).
\end{equation}
The linearised system around $\mathrm{W}_\pm$ has eigenvalues $\lambda_1=\pm\frac{1}{\sqrt{2}}$, $\lambda_2=\pm \sqrt{2}$ and $\lambda_3=\mp \frac{3}{\sqrt{2}}$ and associated eigenvectors $(1,0,0)$, $(0,1,0)$ and $(0,0,1)$, being hyperbolic saddles on $\{ u_1=0\}$. Figure~\ref{fig:blowupxN1} shows the orbit structure on the vacuum invariant boundaries and the equator. The physical region is therefore a slice of $\mathbb{S}^2$ with boundary consisting of three invariant subsets whose intersection are the fixed points $\mathrm{R}_1$, and $\mathrm{W}_\pm$.
\begin{figure}[ht!]
	\centering
			\includegraphics[trim={0cm 1.4cm 0cm 0cm},clip,width=0.4\textwidth]{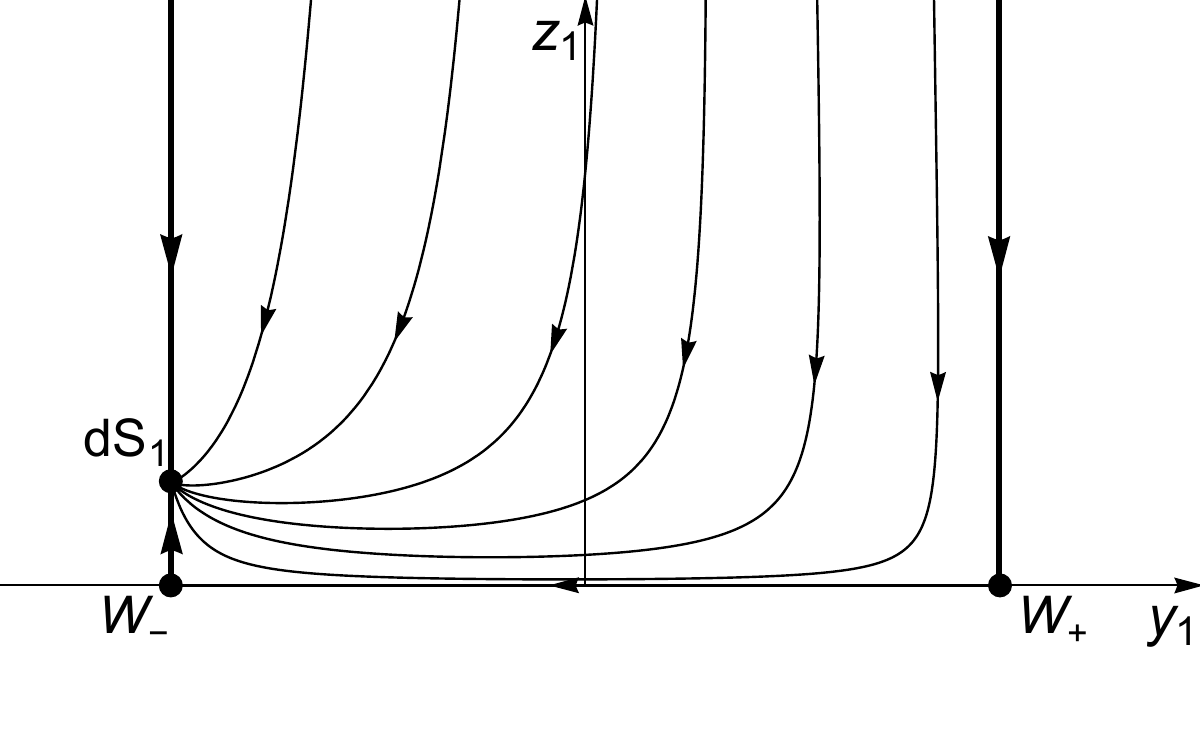}	
	\caption{The invariant plane $\{u_1=0\}$ of the blow-up in the $x$-direction of the fixed point $\mathrm{N}_1$, showing the physical region defined by $y^2_1\leq2$ and $z_1\geq0$.}\label{fig:blowupxN1}
\end{figure}
\begin{lemma}\label{LemComp1}
	The flow in the interior of the physical region on the unit 2-sphere of the blow-up of $\mathrm{N}_1$ consists of the heteroclinic orbit $\mathrm{R}_1\rightarrow \mathrm{dS}_1$ as depicted in Figure~\ref{fig:BUP3D_N1}.
\end{lemma}
\begin{proof}
	The proof follows by the application of the Poincar\'e-Bendixson theorem and the local stability properties of the fixed points which are all located at the invariant boundaries of the physical region.
\end{proof}
\begin{remark}\label{BupN1T1}
	The intersection with the $\{T=1\}$ invariant boundary is just the equator of the blow-up 2-sphere, which therefore consists of the heteroclinic orbit $\{\mathrm{W}_+\rightarrow\mathrm{W}_-\}$. On $\{z_1=0\}$ the invariant plane, $\mathrm{W}_+$ is a source and $\mathrm{W}_-$ is a sink.
\end{remark}
\begin{remark}\label{BupN1Om0}
	The intersection of the vacuum invariant boundary $\{\Omega_{\mathrm{pf}}=0\}$ with the blow-up of $\mathrm{N}_1$ is simply the one-dimensional invariant set consisting of the union of the heteroclinic orbits: $\{\mathrm{R}_1\rightarrow\mathrm{W}_+\}\cup\{\mathrm{R}_1\rightarrow\mathrm{dS}_1\}\cup\{\mathrm{W}_-\rightarrow\mathrm{dS}_1\}$.
\end{remark}
We summarise the previous analysis in the following Proposition:
\begin{proposition}\label{BupN1}
	The flow in a neighbourhood of the non-hyperbolic fixed point $\mathrm{N}_1$ is as depicted in Figure~\ref{fig:BUP3D_N1}, and no interior orbit in $\bar{\mathbf{S}}_\mathrm{J}$ has the fixed point $\mathrm{N}_1$ as $\omega$-limit set.
\end{proposition}
\begin{figure}[ht!]
	\centering
		\includegraphics[trim={3cm 22cm 5cm 3cm},clip,width=0.65\textwidth]{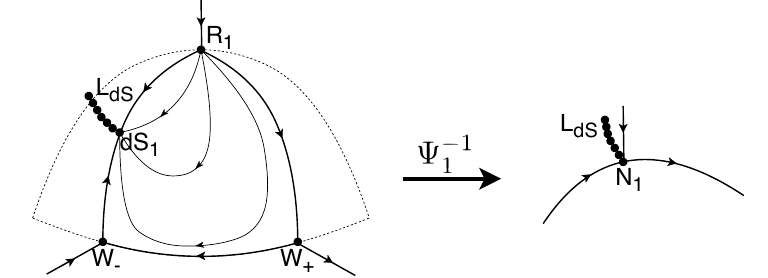}
	\caption{Blow-up space of $\mathrm{N}_1$.}\label{fig:BUP3D_N1}
\end{figure}
%
%
\section{Dynamics in the Einstein frame}\label{subsec:dynsysein}
We now turn to the analysis of the dynamics in the Einstein frame. Under the conformal transformation~\eqref{fr} and~\eqref{EFscalarfield}, with the choice $f(R)=\alpha R^2$, $\alpha>0$, i.e. $p=2$ in~\eqref{Monomial}-\eqref{lambdaRp}, the potential is a positive constant $V(\phi)=V_0$ and $\lambda=0.$
Moreover, as before, we consider a fluid with linear equation of state
\begin{equation}
\tilde{p}_\mathrm{pf}=(\tilde{\gamma}_\mathrm{pf}-1)\tilde{\rho}_\mathrm{pf},
\end{equation}
where the fluid energy-density and pressure in the Einstein frame are relate to those on the Jordan frame by equation~\eqref{SETFluid}. In particular we have $\tilde{\gamma}_\mathrm{pf}=\gamma_\mathrm{pf}$. 
The system~\eqref{Edota}-\eqref{Edotrho} results in the following evolution equations
\begin{subequations}
	\begin{align}
	\frac{d\tilde{a}}{d\tilde{t}} &= \tilde{H}\tilde{a}, \label{EdotaMon}\\
	\frac{d\tilde{H}}{d\tilde{t}}  & = - \frac12\left(\left(\frac{d\phi}{d\tilde{t}}\right)^2 + \gamma_{\mathrm{pf}}\tilde{\rho}_{\mathrm{pf}}\right), \label{Ray1Mon} \\
	\frac{d^2\phi}{d\tilde{t}^2}&=-3\tilde{H}\frac{d\phi}{d\tilde{t}}+2\sqrt{\frac{2}{3}} \left(1-\frac{3}{4}\gamma_{\mathrm{pf}}\right)\tilde{\rho}_{\mathrm{pf}}, \label{waveEqMon} \\
	\frac{d\tilde{\rho}_{\mathrm{pf}}}{d\tilde{t}} & = -\left[3\gamma_{\mathrm{pf}}\tilde{H} + 2\sqrt{\frac{2}{3}}\left(1-\frac{3}{4}\gamma_{\mathrm{pf}}\right)\frac{d\phi}{d\tilde{t}}\right]\tilde{\rho}_{\mathrm{pf}}, \label{EdotrhoMon}
	\end{align}
\end{subequations}
and the constraint
\begin{equation} \label{Gauss1Mon} 
3\tilde{H}^2 =
\left[\frac{1}{2}\left(\frac{d\phi}{d\tilde{t}}\right)^2 + V_{0}+\tilde{\rho}_{\mathrm{pf}}\right].
\end{equation}
The equations~\eqref{Ray1Mon}, \eqref{waveEqMon} and~\eqref{EdotrhoMon} determine a closed system of first order equations for $(\tilde{H},\frac{d\phi}{d\tilde{t}},\tilde{\rho}_{\mathrm{pf}})$ which, due to the constraint~\eqref{Gauss1Mon}, yield a dynamical system describing a flow on a 2-dimensional state-space for each value of the parameter $\gamma_\mathrm{pf}$. 
Note that this dynamical system consists of a \emph{reduced} system of equations, where the equations for the scale factor $\tilde{a}$ and for the field $\phi$ decouple.  Once the reduced 2-dimensional system of first order equations has been solved, then equation~\eqref{EdotaMon} yields $\tilde{a} \propto \exp(\int d\tilde{t} \tilde{H})$) and $\phi = \int (d\phi/d\tilde t)d\tilde{t}$. The vacuum state-space forms a 1-dimensional invariant boundary $\tilde{\rho}_{\mathrm{pf}}=0$ of the matter state-space. The above system of equations is not suitable for a \emph{global} analysis since its state-space lacks compactness. So, in the next section, we will formulate the Einstein frame system as a dynamical system on a compact state-space.
%
\subsection{The reduced state-space}
The reduced system of equations can be cast into a global dynamical system introducing the following set of \emph{Hubble} normalised variables
\begin{equation}
\label{DepVar}
  \left(\Omega_\Lambda, \Sigma_{\phi},\tilde{\Omega}_\mathrm{pf}\right) = \left(\frac{V_0}{3\tilde{H}^2} ,\frac{\frac{d\phi}{d\tilde{t}}}{\sqrt{6}\tilde{H}} , \frac{\tilde{\rho}_\mathrm{pf}}{3\tilde{H}^2}\right),
\end{equation}
together with $e$-fold time $\tilde{N}=\ln{(\tilde{a}/\tilde{a}_0)}$, where $\tilde{a}_0$ is constant. This results in the 2-dimensional flow
\begin{subequations}
        \label{2DdynMatter}
	\begin{align}
	\frac{d\Sigma_{\phi}}{d\tilde{N}} &=-(2-\tilde{q})\Sigma_{\phi}+2\left( 1-\frac{3}{4}\gamma_{\mathrm{pf}} \right)\tilde{\Omega}_\mathrm{pf}, \\ 
	\frac{d\tilde{\Omega}_\mathrm{pf}}{d\tilde{N}} &=2\left[1+\tilde{q}-\frac{3}{2}\gamma_{\mathrm{pf}}-2\left( 1-\frac{3}{4}\gamma_{\mathrm{pf}} \right)\Sigma_{\phi} \right]\tilde{\Omega}_\mathrm{pf},
	\end{align}
\end{subequations}
where we used the constraint
\begin{equation}
\Omega_\Lambda =1- \Sigma^{2}_{\phi}-\tilde{\Omega}_\mathrm{pf}\,>0,
\end{equation}
to globally solve for $\Omega_{\Lambda}$, and the Einstein frame deceleration parameter $\tilde{q}$ which, following from 
\begin{equation}
\frac{d\tilde{H}}{d\tilde{t}}=-(1 + \tilde{q})\tilde{H}^2
\end{equation}
and~\eqref{Ray1Mon}, is expressed as
\begin{equation}\label{qscalar}
\tilde{q} := - 1 + 3\Sigma^2_\phi +\frac{3}{2} \gamma_{\mathrm{pf}}\tilde{\Omega}_{\mathrm{pf}}.
\end{equation} 
It is also useful to consider the auxiliary equation
\begin{equation}
\frac{d\Omega_\Lambda}{d\tilde{N}} =3 \left(2\Sigma_{\phi}^2+\gamma_{\mathrm{pf}}\tilde{\Omega}_{\mathrm{pf}} \right)\Omega_\Lambda, 
\end{equation}
which shows that $\Omega_{\Lambda}=0$ is an invariant boundary of the Einstein frame state-space, describing the dynamics of a \emph{free} scalar field and a perfect fluid with linear equation of state. Due to the constraint, the physical relatively compact state-space $\mathbf{S}_\mathrm{E}$ is defined by 
\begin{equation}
\mathbf{S}_\mathrm{E}=\left\{(\Sigma_\phi,\tilde{\Omega}_\mathrm{pf})\in \mathbb{R}^2\, :  1-\tilde{\Omega}_{\mathrm{pf}}-\Sigma^{2}_{\phi}>0,\quad \tilde{\Omega}_{\mathrm{pf}}>0\right\}.
\end{equation}
The reduced state-space $\mathbf{S}_\mathrm{E}$ can be regularly extended to include its invariant boundaries $\tilde{\Omega}_\mathrm{pf}=0$ and $\Omega_{\Lambda}=1-\Sigma^2_\phi-\tilde{\Omega}_\mathrm{pf}=0$, resulting in the compact state-space $\bar{\mathbf{S}}_\mathrm{E}$. 

All fixed points are located at the boundaries and can be found in Table~\ref{FP_Einstein}. In the reduced state-space $\bar{\mathbf{S}}_\mathrm{E}$ there are four isolated fixed points. The de-Sitter fixed point $\mathrm{dS}$ located at $(\Sigma_{\phi},\tilde{\Omega}_{\mathrm{pf}})=(0,0)$, with $\Omega_{\Lambda}=1$ and $\tilde{q}_{\mathrm{dS}}=-1$, is a hyperbolic sink and hence the $\omega$-limit point of a 1-parameter family of interior orbits. This fixed point is associated with the vacuum de-Sitter solution
\begin{equation}
\mathrm{dS}:\qquad \tilde{\rho}_\mathrm{pf}=0,\quad \frac{d\phi}{d\tilde{t}}=0, \quad \tilde{H}=\sqrt{\frac{\Lambda}{3}}, \quad \tilde{a}=a_0 e^{\sqrt{\frac{\Lambda}{3}}\tilde{t}}.
\end{equation}
The two \emph{kinaton} fixed points $\mathrm{K}_{\pm}$ located at $(\Sigma_{\phi},\tilde{\Omega}_{\mathrm{pf}})=(\pm1,0)$, with $\Omega_{\Lambda}=0$ and $\tilde{q}_{\mathrm{K}}=2$, are hyperbolic sources except $\mathrm{K}_{-}$, when $\gamma_\mathrm{pf}\geq 5/3$, as it becomes a saddle (centre saddle when $\gamma_\mathrm{pf}= 5/3$). Hence $\mathrm{K}_+$ is the $\alpha$-limit point of a 1-parameter set of interior orbits and $\mathrm{K}_{-}$ is the $\alpha$-limit point of a 1-parameter set of interior orbits when $\gamma_{\mathrm{pf}}<5/3$, while no interior orbit originates from $\mathrm{K}_{-}$ for $5/3\leq\gamma_{\mathrm{pf}}<2$. These two fixed points correspond to the well-known self-similar solutions
\begin{equation}
\mathrm{K}_{\pm}:\qquad \phi = \pm\sqrt{\frac{2}{3}}\ln{\left(\frac{c_0}{\tilde{t}}\right)} , \quad \frac{d\phi}{d\tilde{t}} = \mp\sqrt{\frac{2}{3}}\tilde{t}^{-1} , \quad \tilde{H} = \frac{1}{3\tilde{t}},\quad \tilde{a}= \tilde{t}^{1/3},
\end{equation}
with $c_0$ constant. In addition, there exists a \emph{kinetic-matter} fixed point $\mathrm{KM}$ located at $(\Sigma_{\phi},\tilde{\Omega}_{\mathrm{pf}})=\left(\frac{3\gamma_{\mathrm{pf}}-4}{3(\gamma_{\mathrm{pf}}-2)},\frac{4(5-3\gamma_{\mathrm{pf}})}{9(\gamma_{\mathrm{pf}}-2)^2}\right)$, with $\Omega_{\Lambda}=0$ and $\tilde{q}_{\mathrm{KM}}=\frac{2}{3}/(2-\gamma_{\mathrm{pf}})$, which is a hyperbolic saddle and the $\alpha$-limit point of a unique interior orbit. This fixed point corresponds to the explicit solution
\begin{equation}\label{ScSol}
\begin{split}
\mathrm{KM}:\qquad &\tilde{\rho}_\mathrm{pf} = \frac{4(\frac{5}{3}-\gamma_{\mathrm{pf}})}{(\frac{8}{3}-\gamma_{\mathrm{pf}})^2}\frac{1}{\tilde{t}^2},\qquad\tilde{H} = \frac{2-\gamma_\mathrm{pf}}{8/3-\gamma_\mathrm{pf}}\frac{1}{\tilde{t}},\qquad \tilde{a}= \tilde{t}^{\frac{2-\gamma_\mathrm{pf}}{8/3-\gamma_\mathrm{pf}}},\\
&\phi = \phi_* + \frac{\sqrt{6}(4/3-\gamma_\mathrm{pf})}{8/3-\gamma_\mathrm{pf}}\ln{\left(\tilde{t}\right)} , \quad \frac{d\phi}{d\tilde{t}} = \frac{\sqrt{6}(4/3-\gamma_\mathrm{pf})}{8/3-\gamma_\mathrm{pf}}\frac{1}{\tilde{t}}
\end{split}
\end{equation}
with $\phi_*\in\mathbb{R}$.
\begin{theorem}
\label{theorem-a}
The orbit structure on $\mathbf{\bar S}_\mathrm{E}$ is as depicted in Figures~\ref{fig:mainfig}, i.e. if $\gamma_\mathrm{pf}<5/3$ there exists a separatrix $\{\mathrm{KM}\rightarrow\mathrm{dS}\}$ which splits the state space in two invariant regions, one consisting of heteroclinic orbits $\mathrm{K}_+\rightarrow\mathrm{dS}$, and the other of heteroclinic orbits $\mathrm{K}_-\rightarrow\mathrm{dS}$, while if $\gamma_{\mathrm{pf}}>5/3$ then all interior orbits originate from $\mathrm{K}_-$ and end at $\mathrm{dS}$.
\end{theorem}
\begin{proof}
The proof follows by the local stability analysis of the fixed points, which are \emph{all} located on the invariant boundaries, and the application of the Poincar\'e-Bendixson theorem.
\end{proof}
\begin{figure}[ht!]
	\centering
	\subfigure[$\frac{2}{3}< \gamma_{\mathrm{pf}} <\frac{4}{3}$.]{
		\includegraphics[trim={1.7cm 1.2cm 0.55cm 0.55cm},clip,width=0.35\textwidth]{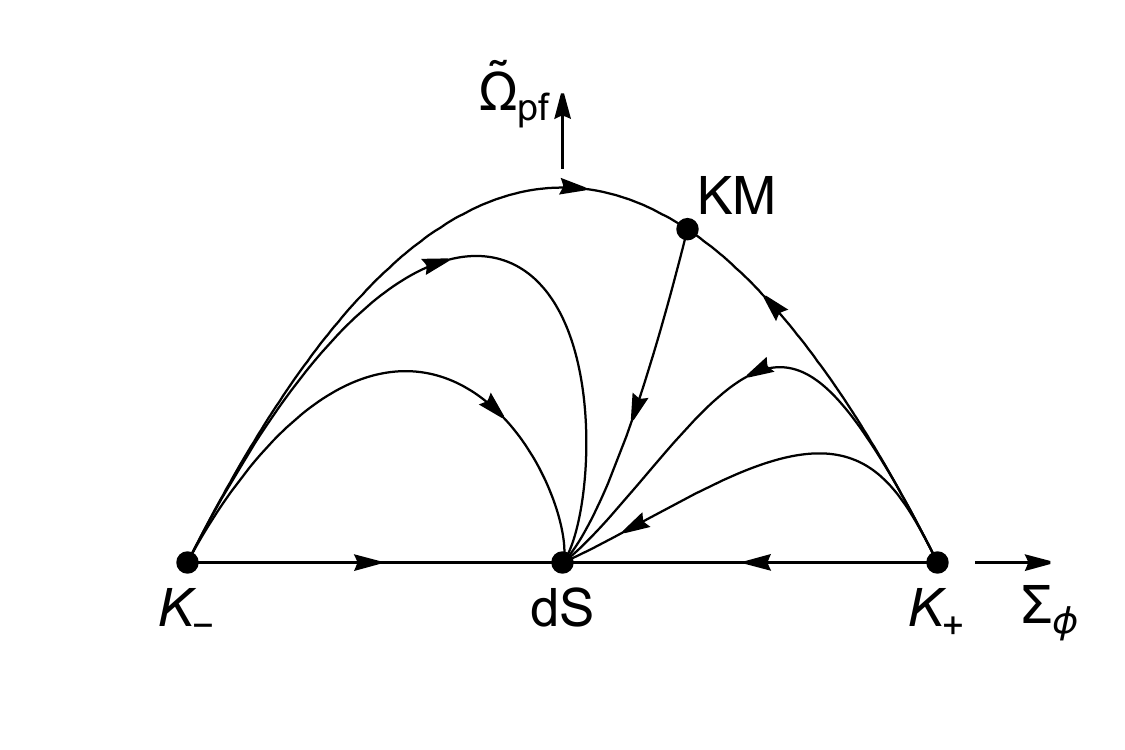}
		\label{fig:subfig1}
	}
	\subfigure[$\gamma_{\mathrm{pf}} =\frac{4}{3}$.]{
		\includegraphics[trim={1.7cm 1.2cm 0.55cm 0.55cm},clip,width=0.35\textwidth]{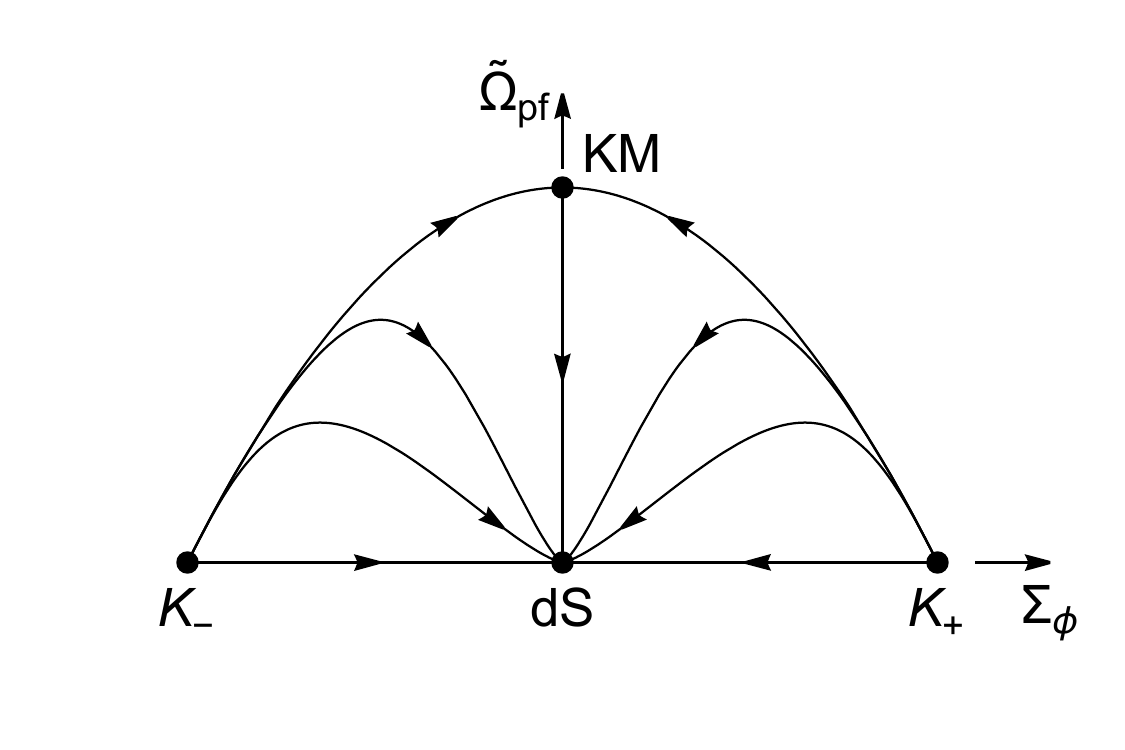}
		\label{fig:subfig2}
	}
	\subfigure[$\frac{4}{3}<\gamma_{\mathrm{pf}} <\frac{5}{3}$.]{
		\includegraphics[trim={1.7cm 1.2cm 0.55cm 0.55cm},clip,width=0.35\textwidth]{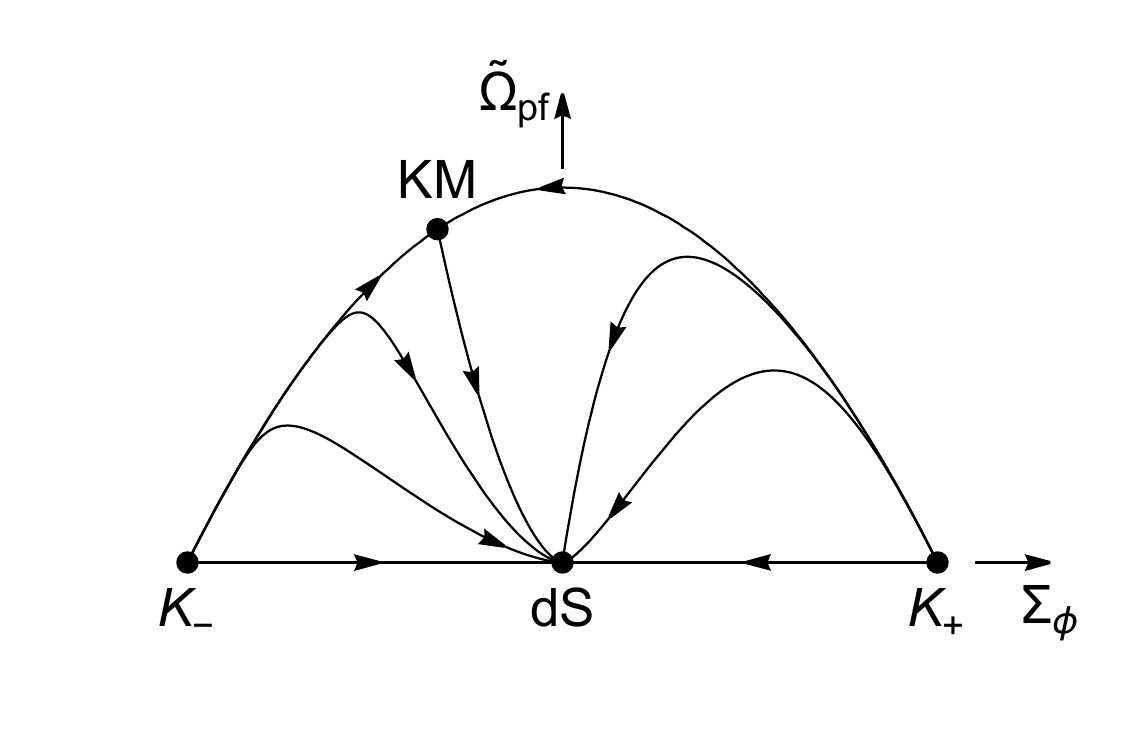}
		\label{fig:subfig3}
	}
	\subfigure[$\frac{5}{3}\leq \gamma_{\mathrm{pf}} <2$.]{
		\includegraphics[trim={1.7cm 1.2cm 0.55cm 0.55cm},clip,width=0.35\textwidth]{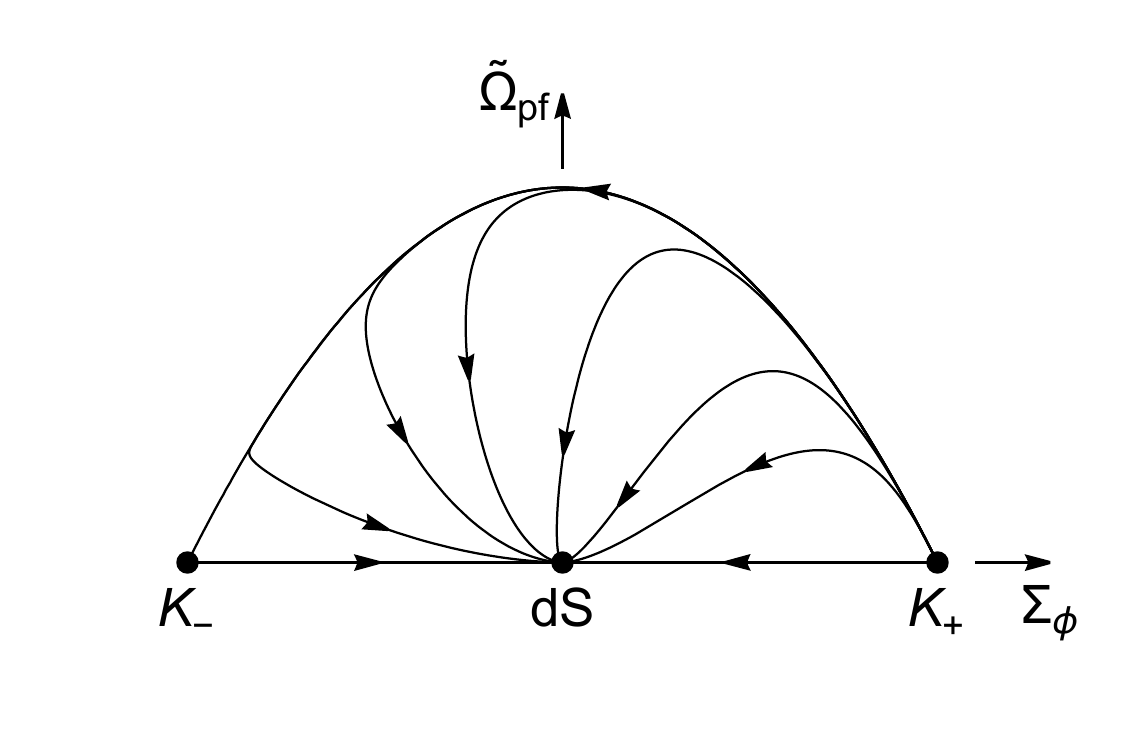}
		\label{fig:subfig4}
	}
	\caption{The reduced Einstein state-space $\bar{\mathbf{S}}_{\mathrm{E}}$ for $f(R)=\alpha R^2$ and a perfect fluid with linear equation of state.}
	\label{fig:mainfig}
\end{figure}
\subsection{The skew-product state-space and fixed points}
In order to better situate the global results of the Einstein frame, given by Theorem~\ref{theorem-a}, into the global Jordan state-space, it is useful to consider the full Einstein frame state-space by adding the scalar field variable $\phi$. To obtain a compact state-space one introduces a regular bounded and monotone (and hence invertible) scalar field variable $\bar{\phi}(\phi)$ satisfying 
\begin{equation}
\lim_{\phi\rightarrow\pm\infty}\bar{\phi}=\bar{\phi}_\pm, \quad \lim_{\phi\rightarrow\pm\infty}\frac{d\bar{\phi}}{d \phi}=0.
\end{equation}
The evolution equation for $\bar{\phi}$ is therefore given by
\begin{equation}
\frac{d \bar{\phi}}{d\tilde{N}}=\sqrt{6} \frac{d\bar{\phi}}{d \phi} \Sigma_\phi.
\end{equation}
We then consider the skew product flow on the product space $[\bar{\phi}_-,\bar{\phi}_+]\times \bar{\mathbf{S}}_\mathrm{E}$, in particular a projection map which projects the orbits in the product space into orbits on the reduced state-space $\bar{\mathbf{S}}_\mathrm{E}$ and which commutes with the flow.

For the problem at hand a suitable choice of the bounded scalar field variable is
 \begin{equation}
\bar{\phi}=\frac{1}{1+e^{- \sqrt{\frac{2}{3}} \phi}}, \quad \frac{d \bar{\phi}}{d \phi}=  \sqrt{\frac{2}{3}} \bar{\phi}(1-\bar{\phi}), 
\end{equation}
so that $\bar{\phi} \to 0$ when $\phi \to -\infty$, and $\bar{\phi} \to 1$ when $\phi \to +\infty$, and
\begin{equation}\label{EvEqbarPhi}
	\frac{d \bar{\phi}}{d\tilde{N}}= 2  \bar{\phi}(1-\bar{\phi})\Sigma_\phi.
\end{equation}

The fixed points on the skew-product state-space $[0,1]\times \bar{\mathbf{S}}_\mathrm{E}$ can be found in Table~\ref{FP_Einstein}. 
\begin{table}
	\begin{center}
		\resizebox{\textwidth}{!}{
		\begin{tabular}{|c|c|c|c|c|c|c|c|}
			\hline            & & & & & & &\\ [-2ex]
			\begin{tabular}{c} {\bf Fixed}\\{\bf points} \end{tabular} & $\bar{\phi}$ & $\Sigma_\phi$ & $\tilde{\Omega}_{\mathrm{pf}}$ & $\Omega_\Lambda$ &{\bf Eigenvalues} & {\bf Eigenvectors} & {\bf Restrictions} \\ [1ex]
			\hline\hline      & & & & & & & \\[-2ex]
			$\mathrm{K}^{0}_{+}$    & $0$ & $+ 1$ & $0$ & $0$ & \begin{tabular}{c} $2$,\\ $6$, \\$2$\\ \end{tabular} & \begin{tabular}{c} $e_1$,\\ $e_2$,\\ $e_2-2e_3$\end{tabular} & \\[1ex] \hline
			$\mathrm{K}^{0}_{-}$    & $0$ & $- 1$ & $0$ & $0$ & \begin{tabular}{c} $-2$,\\ $6$, \\$-2\left( 3\gamma_{\text{pf}}-5 \right)$\\[1ex] \end{tabular} & \begin{tabular}{c} $e_1$,\\ $e_2$,\\ $e_2+2e_3$\end{tabular} & \\[1ex] \hline
			$\mathrm{K}^{1}_{+}$    & $1$ & $+ 1$ & $0$ & $0$ & \begin{tabular}{c} $-2$,\\ $6$, \\$2$\\ \end{tabular} & \begin{tabular}{c} $e_1$,\\ $e_2$,\\ $e_2-2e_3$\end{tabular} & \\[1ex]\hline
			$\mathrm{K}^{1}_{-}$    & $1$ & $- 1$ & $0$ & $0$ & \begin{tabular}{c} $2$,\\ $6$, \\$-2\left( 3\gamma_{\text{pf}}-5 \right)$\\[1ex] \end{tabular} & \begin{tabular}{c} $e_1$,\\ $e_2$,\\ $e_2+2e_3$\end{tabular} & \\[1ex]\hline 
			$ \mathrm{KM}^0$  & $0$ & $\frac{3\gamma_{\mathrm{pf}}-4}{3(\gamma_{\mathrm{pf}}-2)}$ & $\frac{4(5-3\gamma_{\mathrm{pf}})}{9(\gamma_{\mathrm{pf}}-2)^2}$ & $0$ & \begin{tabular}{c} $\frac{2(3\gamma_{\mathrm{pf}}-4)}{3(\gamma_{\mathrm{pf}}-2)}$,\\[1ex]  $\frac{2(3\gamma_{\mathrm{pf}}-8)}{3(\gamma_{\mathrm{pf}}-2)}$,\\[1ex] $-\frac{2(3\gamma_{\mathrm{pf}}-5)}{3(\gamma_{\mathrm{pf}}-2)}$\\[1ex] \end{tabular} & \begin{tabular}{c} $e_1$,\\ $a d e_2-b e_3$,\\ $a e_2-2d e_3$ \end{tabular} & $\gamma_{\mathrm{pf}}<\frac{5}{3}$ \\[1ex] \hline
		    $\mathrm{KM}^1$ & $1$ & $\frac{3\gamma_{\mathrm{pf}}-4}{3(\gamma_{\mathrm{pf}}-2)}$ & $\frac{4(5-3\gamma_{\mathrm{pf}})}{9(\gamma_{\mathrm{pf}}-2)^2}$ & $0$ & \begin{tabular}{c} $- \frac{2(3\gamma_{\mathrm{pf}}-4)}{3(\gamma_{\mathrm{pf}}-2)}$,\\[1ex] $\frac{2(3\gamma_{\mathrm{pf}}-8)}{3(\gamma_{\mathrm{pf}}-2)}$,\\[1ex] $-\frac{2(3\gamma_{\mathrm{pf}}-5)}{3(\gamma_{\mathrm{pf}}-2)}$\\[1ex] \end{tabular} & \begin{tabular}{c} $e_1$,\\ $a d e_2-b e_3$,\\ $a e_2-2d e_3$ \end{tabular} & $\gamma_{\mathrm{pf}}<\frac{5}{3}$\\[1ex] \hline
			${\mathrm{\tilde L}}_\mathrm{R}$    & $\bar{\phi}_*$ & $0$ & $1$ & $0$ & \begin{tabular}{c} $0$,\\ $-3$,\\ $-4$\\ \end{tabular} & \begin{tabular}{c} $e_1$,\\ $c e_1+e_2$,\\ $e_3$ \end{tabular}  & \begin{tabular}{c} $\gamma_\mathrm{pf}=\frac{4}{3}$,\\ $\bar\phi_*\in(0,1)$ \end{tabular} \\ [1ex]\hline
			${\mathrm{\tilde L}}_\mathrm{dS}$    & $\bar{\phi}_*$ & $0$ & $0$ & $1$ & \begin{tabular}{c} $0$,\\ $-3$,\\ $-3\gamma_\mathrm{pf}$ \end{tabular} & \begin{tabular}{c} $e_1$,\\ $c e_1+e_2$,\\ $c d e_1+\gamma_{\mathrm{pf}}d e_2+f e_3$ \end{tabular} & $\bar{\phi}_*\in(0,1)$ \\[1ex]
			\hline
		\end{tabular}}
	\end{center}\vspace{-0.5cm}
	\caption{Fixed points in skew-product state-space $[0,1]\times\bar{\mathbf{S}}_\mathrm{E}$. We introduced the notation $a=3(\gamma_{\mathrm{pf}}-2)$, $b=2(\gamma_{\mathrm{pf}}+2)(3\gamma_{\mathrm{pf}}-5)$, $c=-\frac{2}{3}\bar{\phi}_*(1-\bar{\phi}_*)$, $d=(3\gamma_{\mathrm{pf}}-4)$ and $f=6\gamma_{\mathrm{pf}}(\gamma_{\mathrm{pf}}-1)$. }
	\label{FP_Einstein}
\end{table}
From~\eqref{EvEqbarPhi} it follows that $\bar{\phi}$ is monotonically increasing (decreasing) when $\Sigma_{\phi}>0$ ($\Sigma_{\phi}<0$). Therefore the fixed points on the reduced state-space with $\Sigma_{\phi}=0$ result in lines of fixed points on the skew-product state-space which are parameterised by constant values of $\bar{\phi}$. Hence there are two heteroclinic orbits $\mathrm{K}^{0}_+\rightarrow\mathrm{K}^{1}_{+}$ and $\mathrm{K}^{1}_-\rightarrow\mathrm{K}^{0}_-$ for  $\Sigma_\phi=\pm1$, i.e. at the intersection of the invariant boundaries $\Omega_\Lambda=0$ and $\tilde{\Omega}_\mathrm{pf}=0$. \\

The vacuum invariant boundary  $\tilde{\Omega}_\mathrm{pf}=0$ on the skew-product state-space has an attracting normally hyperbolic line of de-Sitter fixed points $\tilde{\mathrm{L}}_\mathrm{dS}$.  The interior flow therefore consists of heteroclinic orbits $\mathrm{K}^{1}_- \rightarrow \mathrm{L}_{\mathrm{dS}}$ and $\mathrm{K}^{0}_+ \rightarrow \mathrm{L}_{\mathrm{dS}}$ as shown in Figure~\ref{fig:Vaccum-EF}.
\begin{figure}[ht!]
	\centering
		\includegraphics[trim={1.8cm 1.7cm 0.55cm 1cm},clip,width=0.42\textwidth]{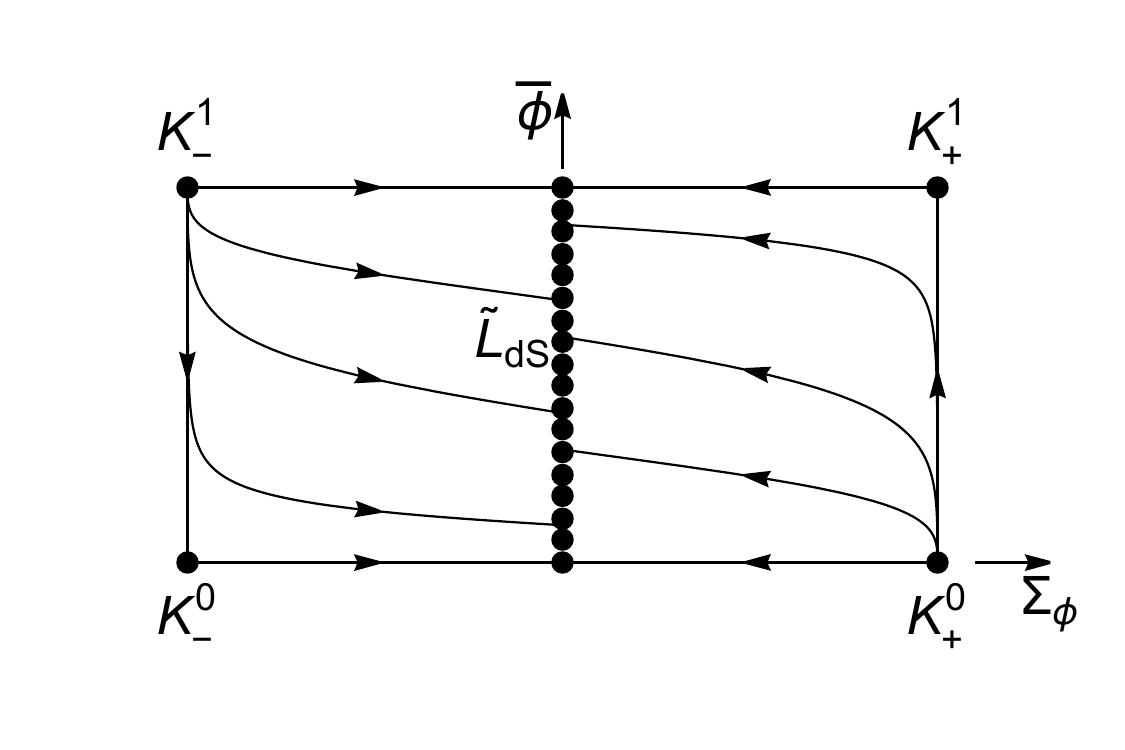}	
	\caption{The vacuum invariant boundary, $\tilde{\Omega}_\mathrm{pf}=0$, in the Einstein frame.}
	\label{fig:Vaccum-EF}
\end{figure}
In turn, the invariant boundary $\{\Omega_\mathrm{\Lambda}=0\}$ on the skew-product state-space has separatrices given by heteroclinic orbits  $\mathrm{KM}^{0}\rightarrow\mathrm{KM}^{1}$ if $\gamma_{\mathrm{pf}}<4/3$, or $\mathrm{KM}^{1}\rightarrow\mathrm{KM}^{0}$ if $4/3<\gamma_{\mathrm{pf}}<5/3$, while if $\gamma_{\mathrm{pf}}=4/3$ then $\Sigma_\phi=0$ and there is a line of fixed points $\tilde{\mathrm{L}}_\mathrm{R}$. These split the state-space into two invariant regions consisting on heteroclinic orbits $\mathrm{K}^{1}_-\rightarrow\mathrm{KM}^{1}$ and $\mathrm{K}^{0}_+\rightarrow\mathrm{KM}^{1}$. When $5/3\leq\gamma_{\mathrm{pf}}<2$ the invariant boundary $\{\Omega_\mathrm{\Lambda}=0\}$ consists of heteroclinic orbits $\mathrm{K}^{0}_+\rightarrow \mathrm{K}^{0}_-$, see Figure~\ref{fig:Lambda-bound}.

\begin{figure}[ht!]
	\centering
		\subfigure[$2/3<\gamma_{\mathrm{pf}}< 4/3$.]{\label{fig:Lambda-boundary1}
			\includegraphics[trim={1.8cm 1.2cm 0.55cm 1cm},clip,width=0.35\textwidth]{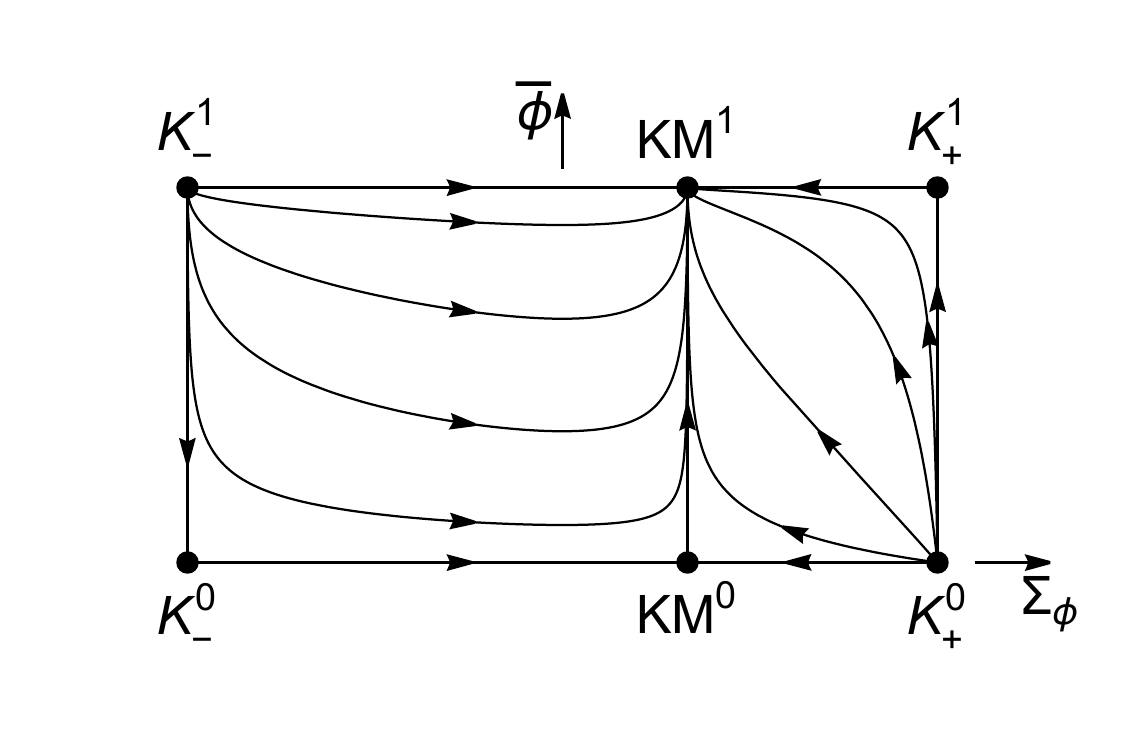}}
		\subfigure[$\gamma_{\mathrm{pf}}=4/3$.]{\label{fig:Lambda-boundary2}
			\includegraphics[trim={1.8cm 1.2cm 0.55cm 1cm},clip,width=0.35\textwidth]{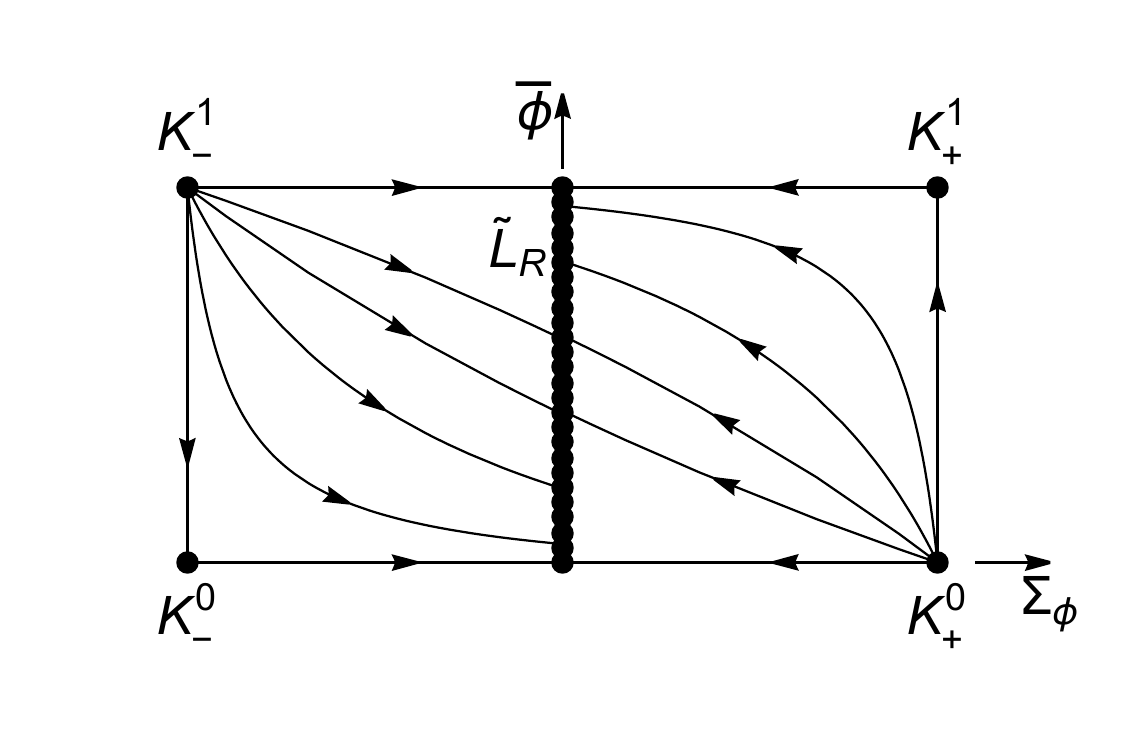}} \\
		\subfigure[$4/3<\gamma_{\mathrm{pf}}<5/3$.]{\label{fig:Lambda-boundary3}
			\includegraphics[trim={1.8cm 1.2cm 0.55cm 1cm},clip,width=0.35\textwidth]{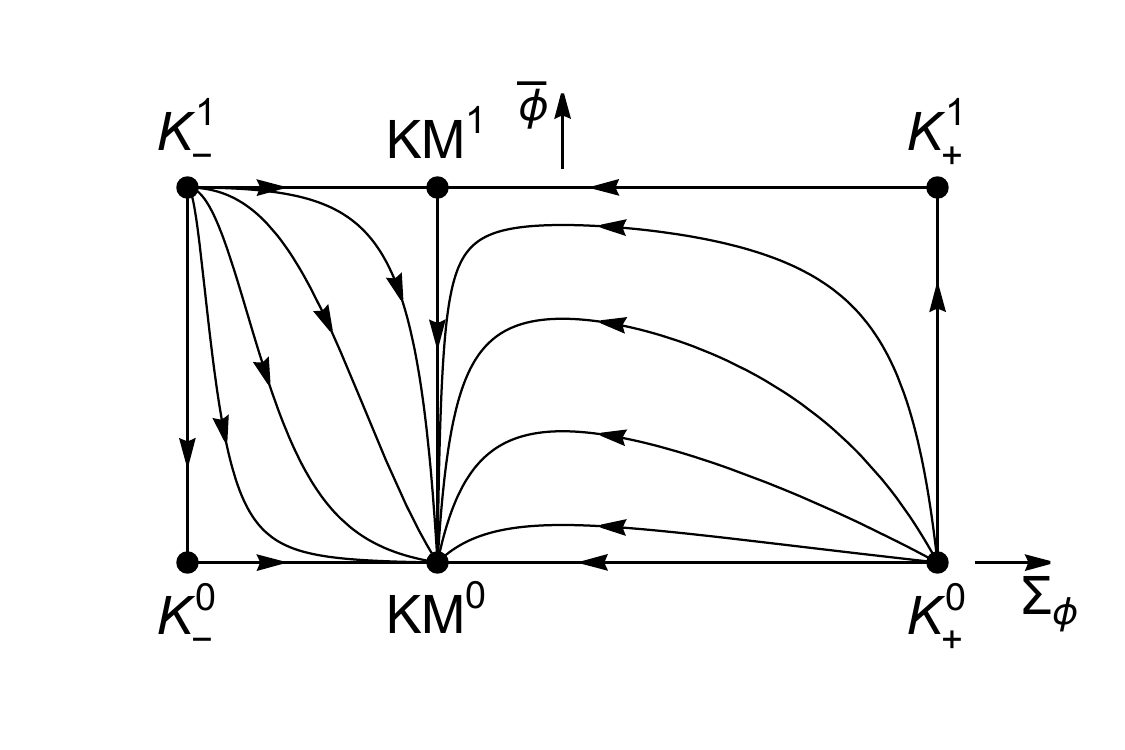}}
		\subfigure[$5/3\leq\gamma_{\mathrm{pf}}<2$.]{\label{fig:Lambda-boundary4}
			\includegraphics[trim={1.8cm 1.2cm 0.55cm 1cm},clip,width=0.35\textwidth]{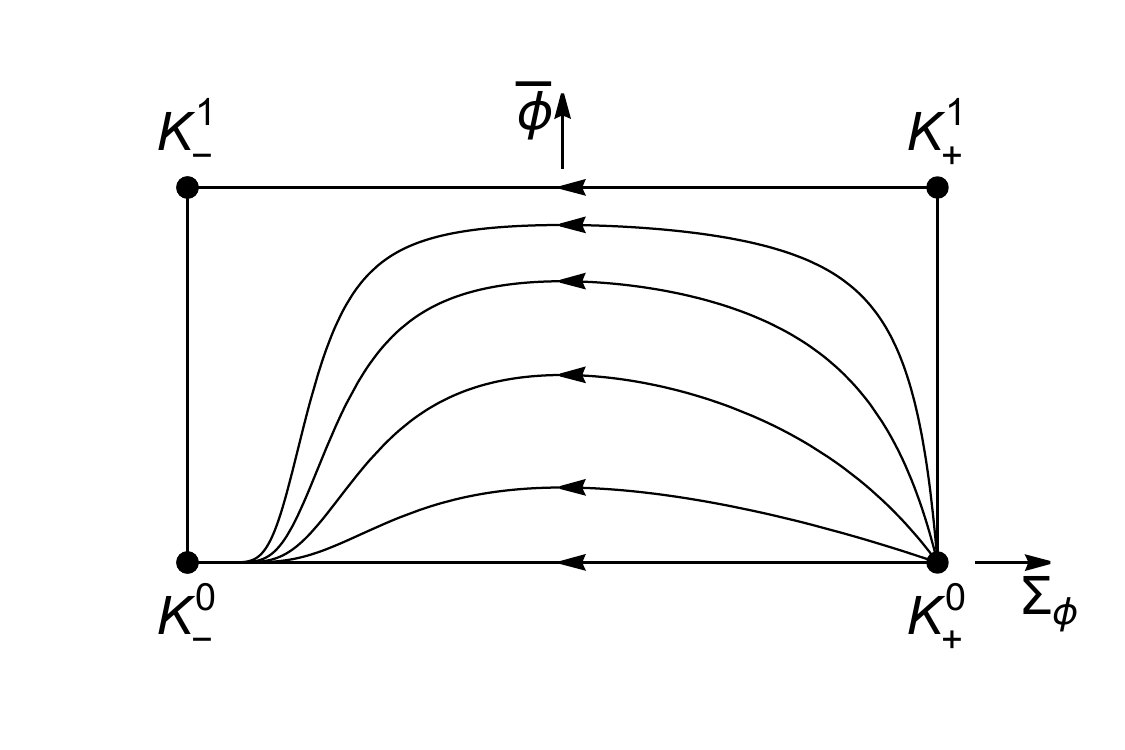}}
	\caption{Invariant boundary $\{\Omega_\Lambda=0\}$ in the Einstein frame skew-product state-space $[0,1]\times\bar{\mathbf{S}}_\mathrm{E}$.}
	\label{fig:Lambda-bound}
\end{figure}

The dynamics on the global Einstein state-space $[0,1]\times \mathbf{\bar S}_\mathrm{E}$ is therefore easily deduced. First notice that for $\gamma_{\mathrm{pf}}=4/3$ there exists a \emph{separatrix surface} consisting of a 1-parameter family of heteroclinic orbits connecting each fixed point on $\tilde{\mathrm{L}}_\mathrm{R}$ to each fixed point on $\tilde{\mathrm{L}}_\mathrm{dS}$ with constant $\bar{\phi}_*\in[0,1]$. This codimension one surface which is the unstable manifold of $\tilde{\mathrm{L}}_\mathrm{R}$ splits the state-space into two invariant regions consisting of heteroclinic orbits $\mathrm{K}^{0}_+\rightarrow\tilde{\mathrm{L}}_\mathrm{dS}$ and $\mathrm{K}^{1}_-\rightarrow\tilde{\mathrm{L}}_\mathrm{dS}$, see Figure~\ref{fig:subfig3D2}. If $2/3<\gamma_{\mathrm{pf}}<4/3$ then the separatrix surface consists of a 1-parameter family of heteroclinic orbits $\mathrm{KM}^{0}\rightarrow \tilde{\mathrm{L}}_\mathrm{dS}$ while, if $4/3<\gamma_{\mathrm{pf}}<5/3$, the separatrix surface consists of a 1-parameter family of heteroclinic orbits $\mathrm{KM}^{1}\rightarrow \tilde{\mathrm{L}}_\mathrm{dS}$. Finally if $5/3\leq\gamma_{\mathrm{pf}}<2$ all solutions originate from $\mathrm{K}^{0}_+$ and end at $\tilde{\mathrm{L}}_\mathrm{dS}$, see Figure~\ref{fig:mainfig2}.
\begin{figure}[htbp]
	\centering
	\subfigure[$\frac{2}{3}< \gamma_{\mathrm{pf}} <\frac{4}{3}$.]{
		\includegraphics[trim={3.25cm 4.25cm 5.5cm 0.25cm},clip,width=0.43\textwidth]{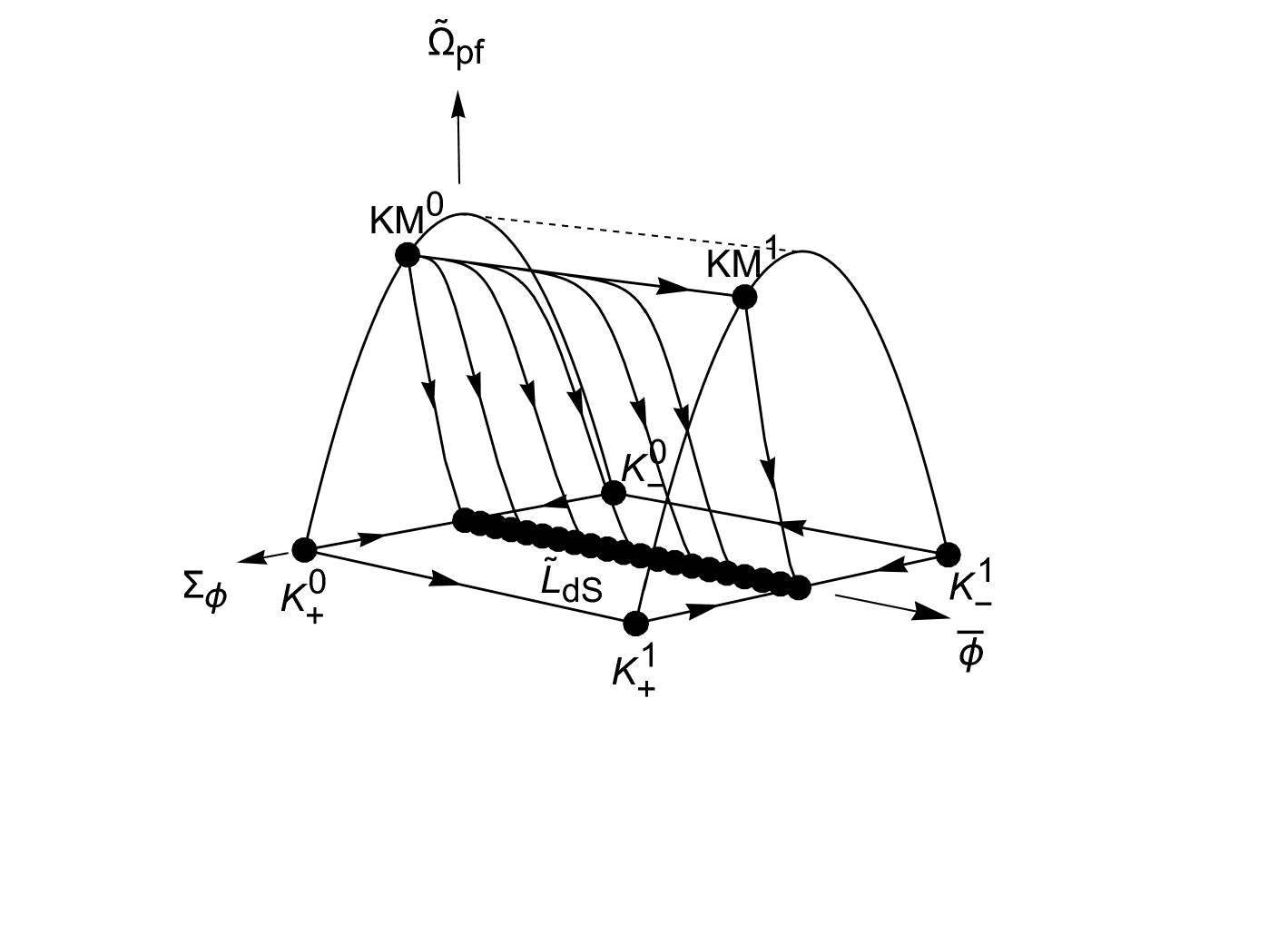}
		\label{fig:subfig3D1}
	}
	\subfigure[$\gamma_{\mathrm{pf}} =\frac{4}{3}$.]{
		\includegraphics[trim={3.75cm 4.25cm 5.2cm 0.25cm},clip,width=0.43\textwidth]{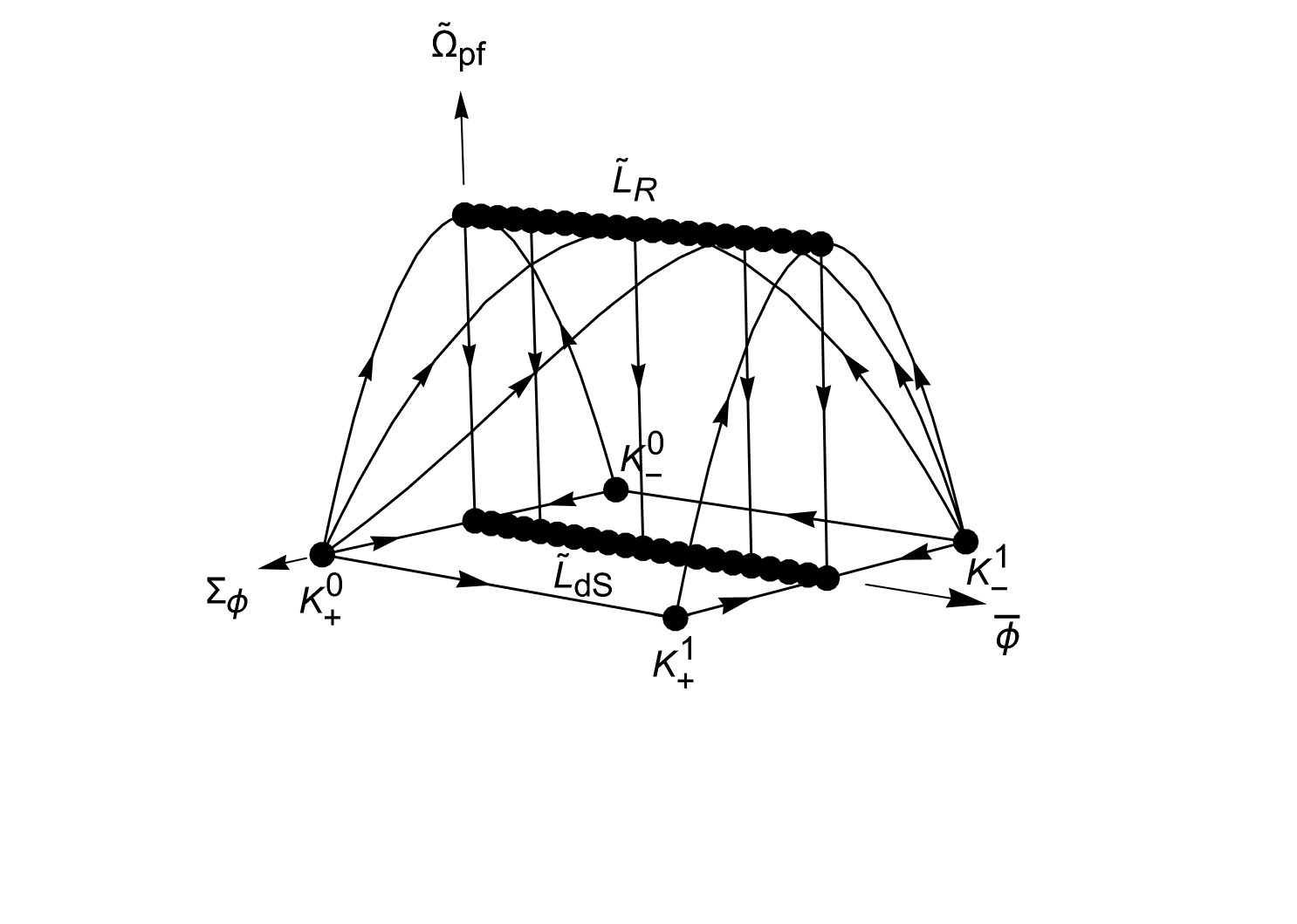}
		\label{fig:subfig3D2}
	}
	\\
	\subfigure[$\frac{4}{3}<\gamma_{\mathrm{pf}} <\frac{5}{3}$.]{
		\includegraphics[trim={3.75cm 4.25cm 5.2cm 0.25cm},clip,width=0.43\textwidth]{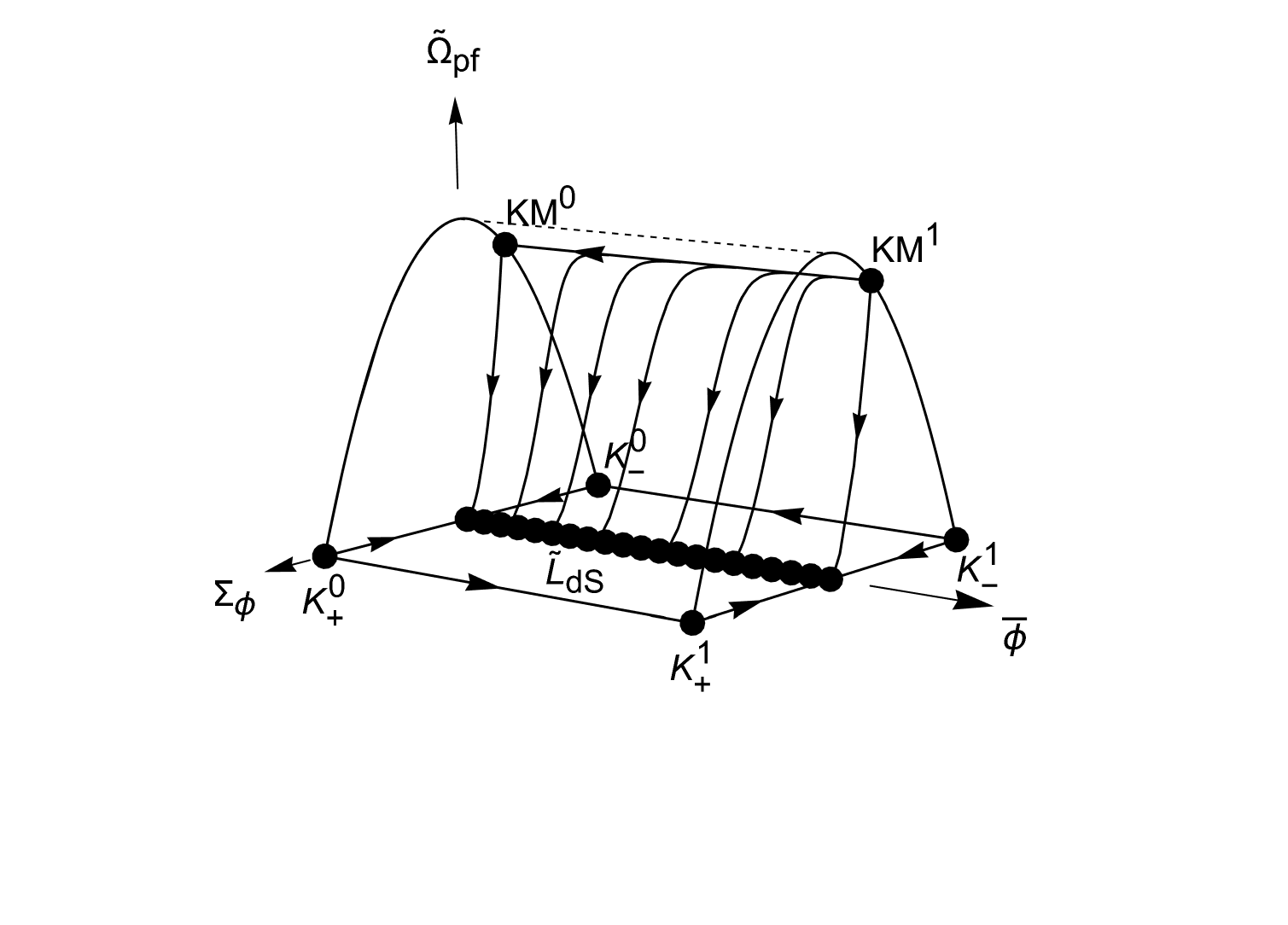}
		\label{fig:subfig3D3}
	}
	\subfigure[$\frac{5}{3}\leq \gamma_{\mathrm{pf}} <2$.]{
		\includegraphics[trim={3.75cm 4.25cm 5.2cm 0.25cm},clip,width=0.43\textwidth]{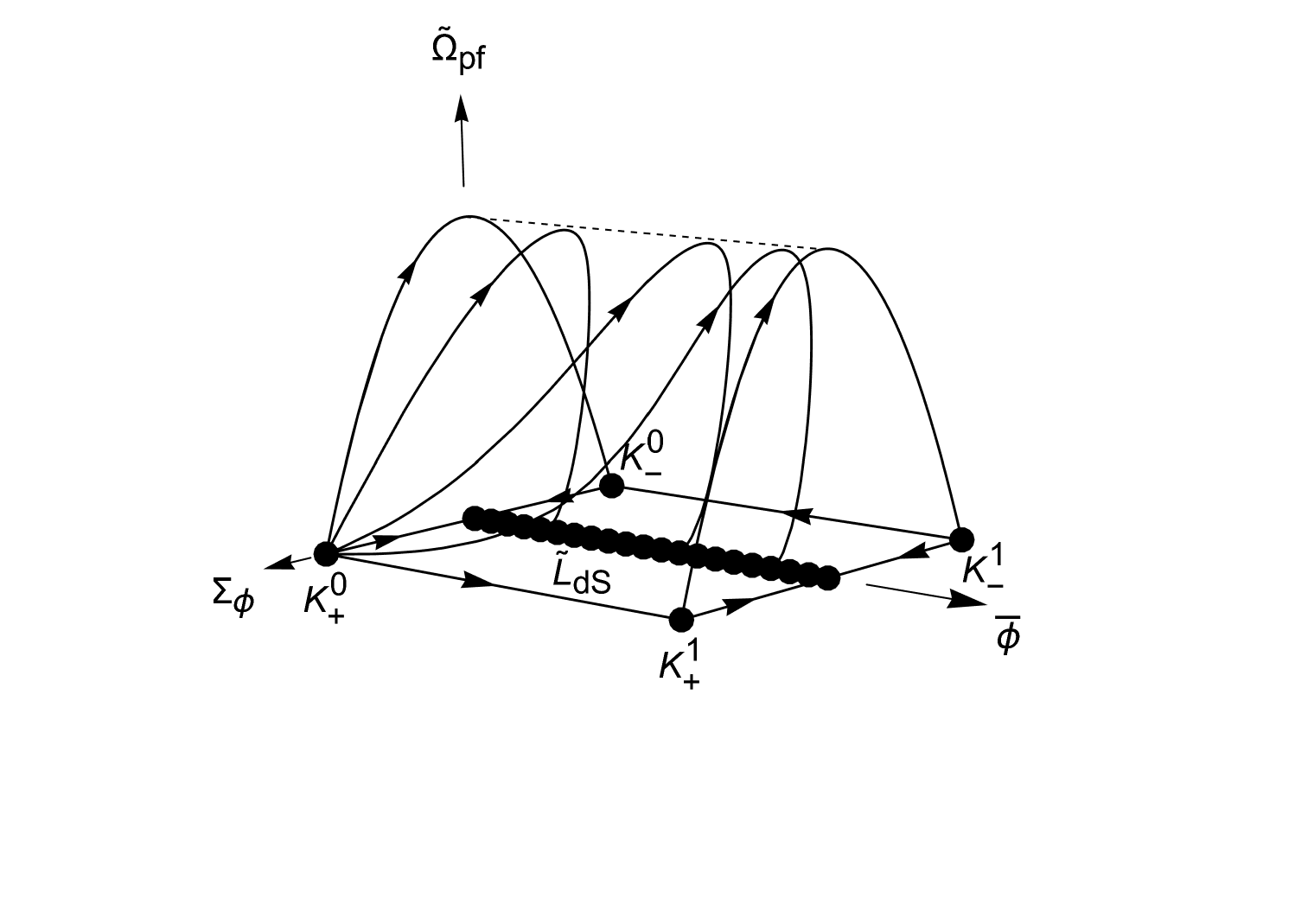}
		\label{fig:subfig3D4}
	}
	\caption{Global dynamics on the Einstein frame skew-product state-space for $f(R)=\alpha R^2$ and a perfect fluid with linear equation of state.}
	\label{fig:mainfig2}
\end{figure}
%
\section{Situating the Einstein state-space in the Jordan state-space}\label{sec:Map}
In this section we discuss how solutions in the Jordan frame are mapped to the Einstein frame. In order to do that we compute explicit expressions for the asymptotic solutions.
Firstly, we will need to relate the two global state-space variables.
The variables in Einstein frame are given in terms of the variables in Jordan frame as follows: 
\begin{subequations}
	\begin{align}\label{EtoJ}
	\bar{\phi} &=\frac{2\alpha R}{1+2\alpha R}= -\frac{2(1-T)S}{T-2(1-T)S},\\
	\Sigma_{\phi} &=\frac{\dot{R}}{ \dot{R}+2HR }=\frac{(1+X)T+(1-T)(1-X)S}{(1+X)T-(1-T)(1-X)S}, \label{Sigma}\\
	\tilde{\Omega}_{\mathrm{pf}} &=\frac{2 R \rho_{\text{pf}}}{3\alpha \left( \dot{R}+2 H R \right)^2}=\frac{-4T(1-T)S  \Omega_{\mathrm{pf}}}{\left( (1+X)T-(1-T)(1-X)S \right)^2},\\
	\Omega_\Lambda &= \frac{R^3}{3\left( \dot{R}+2 H R \right)^2}=\frac{-4T(1-T)S^3}{\left( (1+X)T-(1-T)(1-X)S \right)^2},
	\end{align}
\end{subequations}
which are only valid for $F>0$, i.e. $S<0$. Using these expressions we can write the deceleration parameter in the Einstein frame as
\begin{equation}\label{tildeq}
\tilde{q} = -1+\frac{3((1+X)T+(1-T)(1-X)S)^2-6\gamma_\mathrm{pf}T(1-T)S  \Omega_{\mathrm{pf}}}  {\left( (1+X)T-(1-T)(1-X)S \right)^2}.
\end{equation}

\subsection{Vacuum solutions}
Let us consider the vacuum solutions with $\rho_\mathrm{pf}=0$. In this case $\tilde{\rho}_\mathrm{pf}=0$ and the global dynamics in the Jordan and Einstein frames is shown in figures~\ref{fig:VBJF} and \ref{fig:Vaccum-EF}, respectively. In the Jordan frame \emph{all} solutions originate from $\mathrm{R}_0$ and end at $\mathrm{L}_\mathrm{dS}$. The line of de-Sitter fixed points in the Jordan frame with $q=-1$ has $S<0$ and is therefore globally mapped to the Einstein frame, where it corresponds to the line of de-Sitter fixed points $\tilde{\mathrm{L}}_\mathrm{dS}$ with $\tilde{q}=-1$, as follows from~\eqref{tildeq}. The values for the fixed point $\mathrm{L}_\mathrm{dS}$ are given in Table~\ref{FP}. For the fixed point $\mathrm{R}_0$ we get $\tilde{q}=2$, corresponding to the kinaton fixed points $\mathrm{K}^{0}_{\pm}$ and $\mathrm{K}^{1}_{\pm}$. The linearised solutions starting at $\mathrm{R}_0$ are given by
\begin{equation}
\label{asympt-tbar}
\theta(\bar{t})=\pi+C_\theta e^{2\bar{t}}, \qquad T(\bar{t})= C_T e^{4\bar{t}},
\end{equation}
where $C_\theta$ and $C_T>0$ are constants and $\bar t$ was defined in \eqref{tbar}. Hence we get, as $\bar{t}\rightarrow-\infty$,
\begin{equation}
R(t) = -\frac{C_\theta}{\alpha C_T}e^{-2\bar{t}}, \qquad \dot{R}(t) =\frac{C_\theta }{\alpha\sqrt{3\alpha}C^2_T}e^{-6\bar{t}}, \qquad H(t) =\frac{1}{\sqrt{3\alpha}C_T}e^{-4\bar{t}}.  
\end{equation} 
Moreover from \eqref{tbar} and \eqref{asympt-tbar}, we get
\begin{equation}
t=\sqrt{12\alpha}\int T(\bar{t}) d\bar{t}= \frac{\sqrt{3\alpha}}{2}C_T e^{4\bar{t}}+c,
\end{equation}
where we set $c=0$ so that $t\rightarrow 0^{+}$ as $\bar{t}\rightarrow-\infty$. This means that $R\rightarrow\pm \infty$, depending on the sign of $C_\theta$ and $R\rightarrow 0$ if $C_\theta=0$ which corresponds to the heteroclinic orbit $\mathrm{R}_0\rightarrow\mathrm{N}_1$, with $R=\dot{R}=0$ and $H=1/(2t)$. This splits the solutions which are globally  mapped to the Einstein frame ($C_\theta<0$) from those that conformally incomplete ($C_\theta>0$). Solutions with $C_\theta<0$ correspond to solutions in the Einstein frame originating from the kinaton source at $\bar{\phi}=1$ ($\phi\rightarrow +\infty$) with $\Sigma_{\phi}=-1$, i.e. to $\mathrm{K}^{1}_{-}$. In turn the solutions with $C_\theta>0$ eventually cross the $F=0$ ($R=0$) surface and therefore correspond to solutions originating from $\bar{\phi}=0$ ($\phi\rightarrow -\infty$) with $\Sigma_{\phi}=1$, i.e. from the source $\mathrm{K}^{0}_{+}$. 
%
\subsection{Perfect fluid solutions}
For perfect fluid solutions in the Jordan frame,  Theorem~\ref{GlobalTheoJordan} states that \emph{all} interior orbits originate from $\mathrm{R}_0$, $\mathrm{L}_\mathrm{R}$ or $\mathrm{N}_0$ depending on whether $2/3<\gamma_{\mathrm{pf}}<4/3$, $\gamma_{\mathrm{pf}}=4/3$ or $4/3<\gamma_{\mathrm{pf}}<2$, respectively, and end at the line of de-Sitter fixed points $\mathrm{L}_{\mathrm{dS}}$. We now show how the solutions in the Jordan state-space relate to those in the skew-product Einstein state-space.

If $\gamma_\mathrm{pf}\in(\frac{4}{3},2)$ then all interior orbits in $\bar{\mathbf{S}}_\mathrm{J}$ originate from the fixed point $\mathrm{P}$ on the blow-up of $\mathrm{N}_0$,  see Section \ref{BUP_N0}. Using chart $\kappa^{(0)}_1$, defined in \eqref{ki}, the Einstein frame state variables are given by
\begin{equation}\label{EinsteinBUPN0}
\bar{\phi}=\frac{-2y_1}{u_1 z_1-2y_1}, \qquad \Sigma_\phi = -\frac{1+(2-u_1)\frac{z_1}{y_1}}{1-(2-u_1)\frac{z_1}{y_1}}, \quad \tilde{\Omega}_\mathrm{pf}=-4\frac{\frac{z_1}{y_1}(2-u_1-u_1 y^2_1)}{(1-(2-u)\frac{z_1}{y_1})^2}.
\end{equation}
From~\eqref{tildeq}, it then follows that
\begin{equation}
\tilde{q}=-1+ \frac{3\left(1+(2-u_1)\frac{z_1}{y_1}\right)^2-6\gamma_{\mathrm{pf}}\frac{z_1}{y_1}(2-u_1-u_1 y^2_1)}{\left(1-(2-u)\frac{z_1}{y_1}\right)^2}.
\end{equation}
If $\gamma\neq5/3$, the linearised solutions around $\mathrm{P}$ are given by
\begin{subequations}
	\begin{align}
	u_1(\tau_1)&= C_u e^{(3\gamma_{\mathrm{pf}}-4)\tau_1},\\ 
	y_1(\tau_1)&=\left(C_y+\frac{2}{5-3\gamma_{\mathrm{pf}}}C_z\right) e^{\tau_1}-\frac{2}{5-3\gamma_{\mathrm{pf}}}C_ze^{3(2-\gamma_{\mathrm{pf}})\tau_1}  ,\\
	z_1(\tau_1)&=C_z e^{3(2-\gamma_{\mathrm{pf}})\tau_1},
	\end{align}
\end{subequations}
where $C_u,C_z>0$ and $C_y$ are constants. Inserting these expressions in~\eqref{EinsteinBUPN0}, leads to
\begin{subequations}
	\begin{align}
	\bar{\phi}&= \frac{
		2 C_y (5 - 3\gamma_{\mathrm{pf}}) + 4 C_z - 4 C_z e^{(5 - 3\gamma_{\mathrm{pf}})\tau_1}}{2 C_y (5 - 3\gamma_{\mathrm{pf}}) + 4 C_z - 4 C_z e^{(5 - 3\gamma_{\mathrm{pf}})\tau_1} 
		+ C_u C_z e^{3(5 - 3\gamma_{\mathrm{pf}})\tau_1}},\\
	\Sigma_\phi&= -\frac{(5 - 3\gamma_{\mathrm{pf}}) C_y + 2 C_z  + 2(4 - 3\gamma_{\mathrm{pf}}) C_z e^{(5 - 3\gamma_{\mathrm{pf}})\tau_1}- (5 - 3\gamma_{\mathrm{pf}}) C_u C_z e^{\tau_1}}{ (5 - 3\gamma_{\mathrm{pf}}) C_y + 2 C_z - 6(2 - \gamma_{\mathrm{pf}}) C_z e^{(5 - 3\gamma_{\mathrm{pf}})\tau_1}+ (5 - 3\gamma_{\mathrm{pf}}) C_u C_z e^{\tau_1}},\\
	\tilde{\Omega}_\mathrm{pf}&=\frac{4 C_z\left((5 - 3\gamma_\mathrm{pf}) C_y + 2 C_z - 2 C_z e^{(5 - 3\gamma_\mathrm{pf})\tau_1}\right)}{5 - 3\gamma_\mathrm{pf}} \nonumber\times\\
	&\times \frac{\left(-2 (5 - 3\gamma_\mathrm{pf})^2 e^{(2 - 3\gamma_\mathrm{pf})\tau_1}+ C_u e^{3(1 - 2\gamma_\mathrm{pf})\tau_1}\left(2 C_z e^{5\tau_1}- \left((5 - 3\gamma_\mathrm{pf}) C_y + 2 C_z\right) e^{3\gamma_\mathrm{pf} \tau_1}\right)^2\right)}{\left(-2 C_z - (5 - 3\gamma_\mathrm{pf}) C_y+ C_u C_y e^{\tau_1}+ 6 (2 - \gamma_\mathrm{pf}) C_z e^{(5 - 3\gamma_\mathrm{pf})\tau_1}\right)^2}.
	\end{align}
\end{subequations}

For $C_y>0$, solution trajectories are initially in the region $y_1>0$, i.e. $S>0$, and thereby are not globally mapped to the Einstein state-space. Instead, these solutions eventually cross the $S=0$ surface in the Jordan state-space and, in this case, they are associated with the 2-parameter family orbits in the Einstein frame converging to $\mathrm{K}^{0}_{+}$.

For $C_y<0$, if $\gamma_{\mathrm{pf}}\in(\frac{5}{3},2)$ then $(\bar{\phi},\Sigma_{\phi},\tilde{\Omega}_\mathrm{pf})\rightarrow (0,1,0)$, as $\tau_1\rightarrow -\infty$, and the 2-parameter family of interior orbits originating from $\mathrm{P}$ in the Jordan state-space with $S<0$ are also identified with a 2-parameter family of orbits originating from the source $\mathrm{K}^{0}_+$, which is the only $\alpha$-limit point for \emph{all} interior orbits in the Einstein skew-product state-space, see Figure~\ref{fig:subfig3D4}. 

If $\gamma_{\mathrm{pf}}<5/3$, there is a special heteroclinic orbit $\mathrm{P}\rightarrow\mathrm{V}_-$,  only partially covered by this chart, which is located at the invariant set $\{u_1=0\}$ and is given by $z_1/y_1=\frac{3}{2}(\gamma_{\mathrm{pf}}-5/3)$. Inserting the previous expression into~\eqref{EinsteinBUPN0} yields
$\Sigma_\phi=(\gamma_{\mathrm{pf}}-4/3)/(\gamma_{\mathrm{pf}}-2)$, $\tilde{\Omega}_\mathrm{pf}=4(5/3-\gamma_{\mathrm{pf}})/3(\gamma_{\mathrm{pf}}-2)^2$ and $\tilde{q}=\frac{2/3}{2-\gamma_{\mathrm{pf}}}$. Moreover $\bar{\phi}=1$ which corresponds to the saddle fixed point $\mathrm{KM}^{1}$ on the Einstein state-space. As a check, we verified that the same result is obtained if instead one uses the chart $\kappa^{(\pm)}_{3}$ on the blow-up of $\mathrm{Q}_\pm$ and restricts to the invariant subset $\{w_{3\pm}=0\}$, where the part of the heteroclinic orbit is given by~\eqref{w3PtoV}. Note that only $\mathrm{V}_-$ is mapped to the Einstein state-space which has $F>0$, i.e. $S<0$ ($R>0$). Recall that $\mathrm{V}_-$ only exists for $2/3<\gamma_{\mathrm{pf}}<5/3$, while $\mathrm{V}_+$ exists for $5/3<\gamma_{\mathrm{pf}}<2$ and has no counterpart in the Einstein frame.

For $\gamma_{\mathrm{pf}}\in(\frac{4}{3},\frac{5}{3})$ there is a 1-parameter family of orbits emanating from $\mathrm{P}$ obtained by setting $C_y=-\frac{2C_z}{5-3\gamma_{\mathrm{pf}}}$ for which we get  $(\bar{\phi},\Sigma_{\phi},\tilde{\Omega}_{\mathrm{pf}})\rightarrow\left(1,\frac{\gamma_{\mathrm{pf}}-\frac{4}{3}}{2-\gamma_{\mathrm{pf}}},\frac{4(5-3\gamma_{\mathrm{pf}})}{9(2-\gamma_{\mathrm{pf}})^2}\right)$ as $\tau_1\rightarrow-\infty$. These orbits are therefore identified with the 1-parameter family of interior orbits in the Einstein state-space originating from $\mathrm{KM}^{1}$. The  limit $C_z\rightarrow0$ corresponds to the heteroclinic sequence $\mathrm{P}\rightarrow \mathrm{V}_-\rightarrow\mathrm{dS}_0$ along the explicit orbit $\mathrm{P}\rightarrow \mathrm{V}_-$ identified above. This 1-parameter set of solutions are therefore identified with the separatrix surface in the Einstein frame originating from $\mathrm{KM}^1$ and ending at $\tilde{\mathrm{L}}_\mathrm{dS}$, see Figure~\ref{fig:subfig3D3} and Figure~\ref{fig:JORDEIJ3}. 
If $C_y<0$ and $C_y\neq-\frac{2C_z}{5-3\gamma_{\mathrm{pf}}}$, then $(\bar{\phi},\Sigma_{\phi},\tilde{\Omega}_{\mathrm{pf}})\rightarrow (1,-1,0)$ as $\tau_1\rightarrow-\infty$, and the 2-parameter family of orbits converging to $\mathrm{P}$ in the future invariant region $y_1<0$, are identified with the 2-parameter family of interior orbits in the Einstein skew-product state-space that converge to the source $\mathrm{K}^{1}_{-}$. 

If $\gamma_{\mathrm{pf}}=4/3$, then \emph{all} interior orbits in $\mathbf{S}_\mathrm{J}$ originate from the line of fixed points $\mathrm{L}_\mathrm{R}$. In the chart $\kappa^{(0)}_1$, the line  $\mathrm{L}_\mathrm{R}$ is parametrised by constant values of $u_1=u_*\in(0,2)$ and the solutions converging to the line have asymptotics, as $\tau_1\rightarrow-\infty$, given by:
\begin{subequations}
	\begin{align*}
	u_1(\tau_1)&= u_{*},\\ 
	y_1(\tau_1)&=\left(C_y+(2-u_*)C_z\right)e^{\tau_1}-(2-u_{*})C_z e^{2\tau_1},\\
	z_1(\tau_1)&=C_z e^{2\tau_1}. 
	\end{align*}
\end{subequations}
Inserting the previous expression in~\eqref{EinsteinBUPN0}, leads to
\begin{equation}
	\begin{aligned}
	\bar{\phi}=&\frac{C_y+(2-u_*)C_z-(2-u_{*})C_z e^{\tau_1}}{C_y+(2-u_*)C_z-(2-\frac{u_{*}}{2})C_z e^{\tau_1}}, \quad \Sigma_\phi=-\frac{C_y+(2-u_*)C_z}{C_y+(2-u_*)C_z-2(2-u_{*})C_z e^{\tau_1}}, \\
	\tilde{\Omega}_\mathrm{pf}=&-4C_z e^{\tau_1}\left(C_y+(2-u_*)C_z-(2-u_*)C_ze^{\tau_1} \right) \times \\
	&\times  \frac{\left(2-u_*-u_* e^{2\tau_1}C_y+(2-u_*)C_z-(2-u_*)C_z e^{\tau_1} \right)^2}{\left(C_y+(2-u_*)C_z-2(2-u_*)C_z e^{\tau_1}\right)^2}.
	\end{aligned}
\end{equation}

The 2-parameter family of orbits with $C_y>0$ are initially in the region $S>0$ ($F<0$) and eventually cross the surface $S=0$, thereby originating from the kinaton fixed point $\mathrm{K}^{0}_+$ with $\bar{\phi}=0$ and $\Sigma_{\phi}=1$ on the Einstein state-space. 

For the 2-parameter family with $C_y<0$ and $C_y\neq-(2-u_*)C_z$, it follows that  $(\bar{\phi},\Sigma_{\phi},\tilde{\Omega}_\mathrm{pf})\rightarrow (1,-1,0)$ as $\tau_1\rightarrow-\infty$, thereby corresponding to the 2-parameter set of solutions in the Einstein state-space originating from the source $\mathrm{K}^{1}_-$ and that are globally mapped to the Einstein frame state-space.

If $C_y=-(2-u_*)C_z$, then $(\bar{\phi},\Sigma_{\phi},\tilde{\Omega}_\mathrm{pf})\rightarrow (1-\frac{u_*}{4-u_*},0,1)$ as $\tau_1\rightarrow-\infty$, which corresponds to the 1-parameter set of solutions originating from the line of fixed points $\tilde{\mathrm{L}}_\mathrm{R}$ in the Einstein frame parameterised by $\bar{\phi}=\bar{\phi}_*$, $\bar{\phi}_*\in(0,1)$. Therefore the 1-parameter family of orbits $\mathrm{L}_\mathrm{R}\rightarrow\mathrm{L}_\mathrm{dS}$ in the Jordan state-space obtained by setting $C_y=-(2-u_*)C_z$ corresponds to the $1$-parameter family of orbits forming the separatrix surface $\tilde{\mathrm{L}}_\mathrm{R}\rightarrow\tilde{\mathrm{L}}_\mathrm{dS}$ of the Einstein state-space. The limit $u_*\rightarrow 0$ corresponds to the special heteroclinic orbit $\mathrm{P}\rightarrow \mathrm{V}_-$, see Figure~\ref{fig:subfig3D2} and Figure~\ref{fig:JORDEIJ2}, while the limit $u_*\rightarrow 2$ yields the heteroclinic orbit $\mathrm{R}_0\rightarrow \mathrm{N}_1$, more precisely $\mathrm{R}_1$ on the blow-up of $\mathrm{N}_1$.

For $2/3<\gamma_{\mathrm{pf}}<4/3$ the solution trajectories originating from the hyperbolic source $\mathrm{R}_0$, have the following asymptotics as $\bar{t}\rightarrow -\infty$ 
\begin{subequations}
	\begin{align}
	X(\bar{t}) &= -1 + C_Xe^{8\left( 1-\frac{3}{4}\gamma_{\text{pf}} \right)\bar{t}}, \\
	S(\bar{t}) &= C_S e^{2\bar{t}},\\ 
	T(\bar{t}) &= C_T e^{4\bar{t}}, 
	\end{align}
\end{subequations}
where $C_X>0$, $C_S$ and $C_T>0$ are constants. Solutions with $C_S<0$ have $S<0$ ($F>0$) initially and are globally mapped to the Einstein frame state-space where they originate from the fixed point $\mathrm{K}^{1}_-$. Conversely, solutions with $C_S>0$ have $S>0$ initially and must eventually cross $S=0$ in the Jordan frame, hence originating from the boundary $\bar{\phi}=0$. 

Finally, to  obtain the correct limit when $C_S=0$, we must go to higher-order on the asymptotic expansions. Setting $\bar{X}=X+1$ the system~\eqref{dynsysTXS} is given up to second-order by 
\begin{subequations}
\begin{align}
\bar{X}^{\prime}&= -2(3\gamma_\mathrm{pf}-4)\bar{X}+2(3\gamma_\mathrm{pf}-4)\bar{X}^2+(3\gamma_\mathrm{pf}-2)S^2+2(3\gamma_\mathrm{pf}-4)\bar{X} T + \text{h.o.},\\
S^\prime&= 2 S + (3\gamma_\mathrm{pf}-5)\bar{X} S-2S T-\bar{X} T + \text{h.o.},\\
T^\prime&= 4T-8T^2+3(\gamma_\mathrm{pf}-2)\bar{X} T + \text{h.o.}.
\end{align}
\end{subequations}	
Considering $Y(\bar{t})=(\bar{X},S,T)$, the previous dynamical systems takes the form $Y'(\bar{t})=A Y(\bar{t})+Q(Y(\bar{t}))$, where $A=\text{diag}\left(8(1-3/4\gamma_{\mathrm{pf}}),2,4\right)$ and $Q$ collects the second order terms. We seek an approximate solution as a perturbative expansion given by the sum of successive orders, $Y(\bar{t})= Y^{(1)}(\bar{t})+Y^{(2)}(\bar{t})+\text{h.o}$, where $Y^{(1)}$ is the first order solution and $Y^{(2)}$ represents the second order correction. Neglecting the nonlinear terms, the first order solution satisfies $Y^{(1)'}(\bar{t})= A Y^{(1)}(\bar{t})$, yielding $Y^{(1)}(\bar{t})=\left(C_Xe^{8\left(1-\frac{3}{4}\gamma_{\mathrm{pf}}\right)\bar{t}}, C_S e^{2\bar{t}}, C_T e^{4\bar{t}} \right)$. The second order correction is obtained by substituting the expansion into the full system and retaining quadratic terms, leading to the linear inhomogeneous problem $Y^{(2)'}(\bar{t})=A Y^{(2)}(\bar{t})+Q\left( Y^{(1)}(\bar{t}) \right)$. Solving this system and setting $C_S=0$, we obtain
\begin{subequations}
\begin{align}
X(\bar{t})&=-1+C_Xe^{8\left(1-\frac{3}{4}\gamma_{\mathrm{pf}}\right)\bar{t}}\left[1-\left(1-\frac{3}{4}\gamma_\mathrm{pf}\right)C_T e^{4\bar{t}}-C_Xe^{8\left(1-\frac{3}{4}\gamma_{\mathrm{pf}}\right)\bar{t}}  \right], \\
S(\bar{t}) &= -\frac{C_X C_T}{2(5-3\gamma_\mathrm{pf})} e^{8\left(1-\frac{3}{4}\gamma_{\mathrm{pf}}\right)\bar{t}}e^{4\bar{t}}, \\
T(\bar{t}) &=C_T e^{4\bar{t}}\left[1-2C_T e^{4\bar{t}}+\frac{3(\gamma_{\mathrm{pf}}-2)C_X}{2(4-3\gamma_{\mathrm{pf}})} e^{8\left(1-\frac{3}{4}\gamma_{\mathrm{pf}}\right)\bar{t}}   \right],	
\end{align}
\end{subequations}	
so that plugging in~\eqref{EtoJ} leads to $(\bar{\phi},\Sigma_{\phi},\tilde{\Omega}_\mathrm{pf})\rightarrow(0,\frac{\frac{4}{3}-\gamma_{\mathrm{pf}}}{2-\gamma_{\mathrm{pf}}},\frac{4(5-3\gamma_{\mathrm{pf}})}{9(2-\gamma_{\mathrm{pf}})^2})$ as $\tau_1\rightarrow-\infty$. Thus, the 1-parameter family of orbits originating from $\mathrm{R}_0$ tangentially to $S=0$ (obtained by setting $C_S=0$), corresponds to the 1-parameter set of orbits originating from $\mathrm{KM}^{0}$ on the Einstein skew product state-space, see Figure~\ref{fig:subfig3D1} and Figure~\ref{fig:JORDEIJ1}.

\begin{figure}[h!]
	\centering
	\subfigure[$2/3<\gamma_\mathrm{pf}<4/3$.]{
		\includegraphics[trim={1.4cm 28.2cm 17.1cm 2.7cm},clip,width=0.32\textwidth]{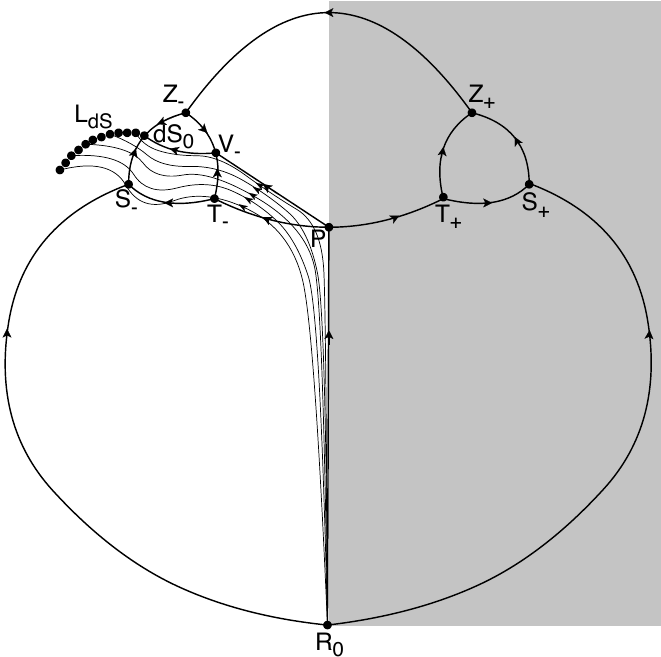}\label{fig:JORDEIJ1}}
	\subfigure[$\gamma_\mathrm{pf}=4/3$.]{
		\includegraphics[trim={1.2cm 28.1cm 17.3cm 2.8cm},clip,width=0.32\textwidth]{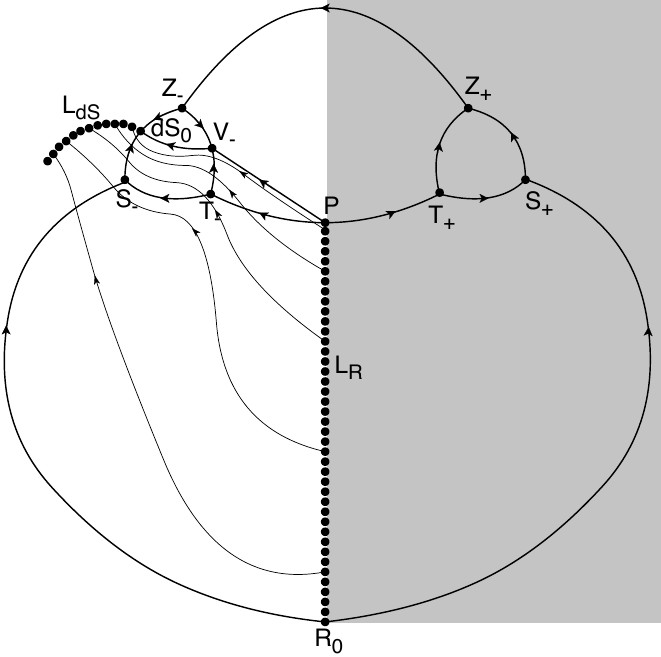}\label{fig:JORDEIJ2}}
	\subfigure[$4/3<\gamma_\mathrm{pf}<5/3$.]{
		\includegraphics[trim={1.6cm 28.9cm 16.8cm 2cm},clip,width=0.32\textwidth]{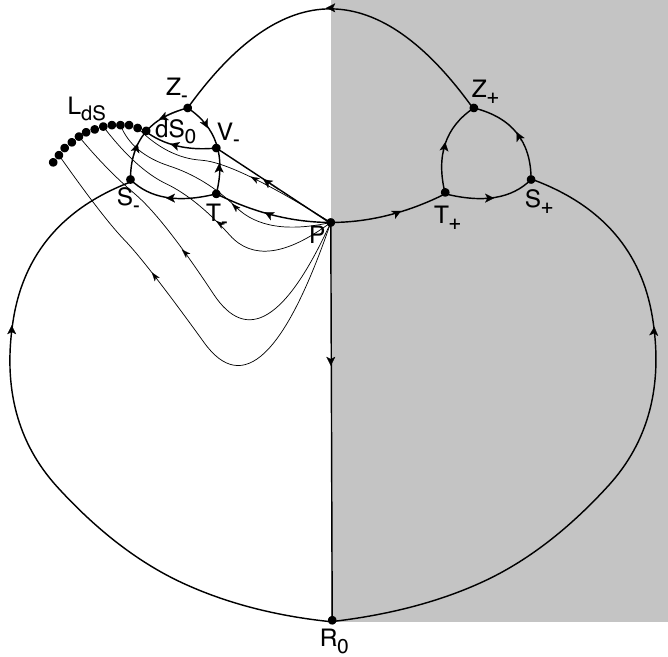}\label{fig:JORDEIJ3}}
	\caption{The figures show the 1-parameter set of orbits in the Jordan frame that correspond to the 1-parameter set in the Einstein frame originating from $\mathrm{KM}^{0}$, $\tilde{\mathrm{L}}_\mathrm{R}$ and $\mathrm{KM}^1$.}\label{fig:JORDEIJ}
\end{figure}


\section*{Acknowledgments}
The authors'  research is partially supported by Portuguese Funds through FCT (Fundação para a Ciência e a Tecnologia) within the Projects UIDB/04459/2020\\ (https://doi.org/10.54499/UID/04459/2025) and UID/00013/2025\\ (https://doi.org/10.54499/UID/00013/2025) as well as European Commission funds H2020-MSCA-2022-SE within Project EinsteinWaves, Grant Agreement~101131233. M.L. thanks support from Fundo Regional da Ciência e Tecnologia and Azores Government through the Ph.D. Fellowship M3.1.a/F/031/2022.

\end{document}